\documentclass[a4paper,UKenglish,cleveref, autoref, thm-restate]{lipics-v2021}
\newcommand{\aI}{\mathcal{I}}

\newcommand{\empt}{\Box}

\newcommand{\kl}{D_\mathrm{KL}}

\newcommand{\RN}{\mathbb{R}}
\newcommand{\NN}{\mathbb{N}}
\newcommand{\ZN}{\mathbb{Z}}

\newcommand{\aR}{\mathcal{R}}

\newcommand{\pse}{\mathrm{PE}}
\newcommand{\lse}{\mathrm{LE}}
\newcommand{\use}{\mathrm{UE}}

\newcommand{\prob}{\mathbb{P}}
\newcommand{\expe}{\mathbb{E}}

\newcommand{\indi}{\mathbf{1}}

\newcommand{\conf}{\delta}

\newcommand{\sgn}{\sigma}
\newcommand{\dsgn}{\omega}

\newcommand{\Var}{\mathrm{Var}}

\newcommand{\eval}[1]{\llbracket  #1 \rrbracket}

\newcommand{\dfp}{r}
\newcommand{\dff}{s}

\newcommand{\dsum}{\mathrm{S}}

\newcommand{\edsum}{\mathbb{S}}
\newcommand{\eset}{\mathcal{U}}
\newcommand{\einter}{\mathcal{I}}

\newcommand{\filt}{\mathcal{F}}

\usepackage{dsfont} % to use \mathds
\usepackage{mathtools} % for \coloneqq
\usepackage{algorithm}
\usepackage{algpseudocode}
\usepackage{bm} % for \bm
\usepackage{stmaryrd} % for \llbracket, \rrbracket
\usepackage{booktabs} % for \toprule, \midrule, \bottomrule

\newcommand{\1}{\mathds{1}} %indicator function

\newcommand{\moni}{\mathcal{M}}
\title{Monitoring Discounted Sum Properties} %TODO Please add

\author{Filip Cano}{Institute of Science and Technology Austria}{filip.cano@ist.ac.at}{}{}

\author{Thomas A. Henzinger}{Institute of Science and Technology Austria}{tah@ist.ac.at}{}{}

\author{Konstantin Kueffner}{Institute of Science and Technology Austria}{konstantin.kueffner@ist.ac.at}{}{}

\author{N. Ege Saraç}{CISPA Helmholtz Center for Information Security, Germany}{ege.sarac@cispa.de}{}{}

\authorrunning{F. Cano, T.A. Henzinger, K. Kueffner, and N.E. Saraç} %TODO mandatory. First: Use abbreviated first/middle names. Second (only in severe cases): Use first author plus 'et al.'

\Copyright{Filip Cano, Thomas A. Henzinger, Konstantin Kueffner, N. Ege Saraç} %TODO mandatory, please use full first names. LIPIcs license is "CC-BY";  http://creativecommons.org/licenses/by/3.0/


 %TODO mandatory; please add comma-separated list of keywords

\category{} %optional, e.g. invited paper

\relatedversion{} %optional, e.g. full version hosted on arXiv, HAL, or other respository/website
\acknowledgements{This work has been supported by the European Research Council under Grant No.: ERC-2020-AdG 101020093.}%optional

\nolinenumbers %uncomment to disable line numbering

\EventEditors{Ana Sokolova and Patrick Totzke}
\EventNoEds{2}
\EventLongTitle{37th International Conference on Concurrency Theory (CONCUR 2026)}
\EventShortTitle{CONCUR 2026}
\EventAcronym{CONCUR}
\EventYear{2026}
\EventDate{September 1--4, 2026}
\EventLocation{Liverpool, UK}
\EventLogo{}
\SeriesVolume{391}
\ArticleNo{32}
\begin{document}

\maketitle

\begin{abstract}
Runtime monitoring of quantitative signals faces a fundamental trade-off between volatility and over-aggregation: instantaneous observations are noisy, while long-run averages obscure local structure. Localisation measures such as discounted averages offer a principled middle ground, yet remain poorly understood in runtime verification. This paper studies discounted sums from a monitoring perspective, in both deterministic and stochastic settings. We formalize the discounted monitoring problem and show that exact, sound monitoring of discounted sums cannot be achieved with finite memory. To overcome this impossibility, we introduce $\varepsilon$-approximately sound monitoring, deriving explicit bounds on memory and observation requirements. We then extend the framework to stochastic processes via expected discounted sums, defining pointwise and uniform $(\varepsilon,\delta)$-soundness notions, establishing statistical optimality, and proving impossibility beyond a precision threshold.
We also formalize the resource complexity of deterministic discounted monitoring via affine register machines and prove a tight worst-case lower bound.
Finally, we present a specification language for arithmetic expressions over multiple discounted sums with synchronous and asynchronous semantics, and evaluate our approach on practical scenarios including algorithmic fairness.
\end{abstract}
\section{Introduction}
Runtime monitoring is a central problem in modern computer science, with applications across systems, networks, machine learning, security, and cyber-physical systems~\cite{bartocci2018introduction,sanchez2019survey}. 
Traditionally, it has focused either on instantaneous properties or long-run behaviors. However, for quantitative signals, such as CPU temperature, model accuracy, or algorithmic fairness, there is a tension: single observations are noisy, while long-run aggregates can obscure local structure. We illustrate this trade-off with the following motivating examples.

\begin{enumerate}[(i)]
    \item \emph{CPU temperature:} monitoring the temperature of a CPU is essential to prevent hardware damage; the conundrum: a single spike that dissipates quickly may raise an undue alarm, yet aggregating over long stretches of time can hide sustained overheating.
    \item \emph{Model accuracy:} monitoring a model’s accuracy is required to ensure sustained performance as distributions shift; the conundrum: one bad prediction is not indicative of the model’s performance, yet a long-run average may fail to adapt to distribution shifts.
    \item \emph{Algorithmic fairness:} monitoring the fairness of an arbiter may uncover uneven issuing of grants between two clients: in isolation, every grant issued to one client is maximally unfair to the other, yet a limit average can hide the deliberate starvation of a client.
\end{enumerate}
A natural way to resolve this conundrum is to use \emph{localisation measures}, such as the window average, discounted average, or, more generally, kernel-weighted averages. These are essentially convolutions that smooth volatile signals by weighting values around a given time point inversely with their distance.
While common in time series analysis, signal processing, and image classification, localisation measures are underexplored in runtime verification.
In this paper, we study a (seemingly) trivial and widely used localisation measure—the discounted average—from a runtime verification perspective.
Discounted monitoring faces two fundamentally different sources of uncertainty: epistemic uncertainty from unseen past and future values, and statistical uncertainty from noisy observations. This paper treats both explicitly and shows that each leads to distinct impossibility thresholds.

\vspace{-0.5em}\subparagraph{Deterministic Discounted Monitoring.}
We assume a bi-infinite stream of numerical values $w=(x_i)_{i\in\mathbb Z}$. Centered at time index $t\in\ZN$, we consider a bi-directional discounted sum, with discount factor $r\in [0,1)$ for the past and $s\in [0,1)$ for the future:
\begin{equation}
     S_t^{r,s}(w)\coloneqq
     \underbrace{\textstyle\sum_{i=1}^\infty r^i \cdot x_{t-i}}_{\text{past}}
     \;+\;
     \underbrace{x_t}_{\text{present}}
     \;+\;
     \underbrace{\textstyle\sum_{i=1}^\infty s^i \cdot x_{t+i}}_{\text{future}}.
\end{equation}
\vspace{-1em}
\begin{example}[Decision Fairness]
    The stream $w$ is a sequence of binary decisions, where $x_t=1$ indicates that client~$1$ was granted access to the resource at time $t$, while $x_t=0$ indicates the same for client~$0$.
    Suppose we are at time $t$ and client~$0$ was granted the resource (so $x_t=0$).
    Normalising the discounted sum by  $\lambda^{r,s}\coloneqq \tfrac{r}{1-r}+\tfrac{s}{1-s}+1$, we obtain the discounted average. 
    Then $S_t^{r,s}(w)/\lambda^{r,s}$ quantifies how fairly the arbiter treated the clients around time $t$, where $0.5$ indicates perfect fairness. Setting $s=0$ yields a purely backward-looking notion of fairness, whereas setting $r=0$ yields a forward-looking notion.
    In particular, for forward-looking fairness, if the arbiter initially favours client~$1$, the discount factor dictates how much it must compensate later to remain fair, if possible at all.
\end{example}
In discounted monitoring, time progresses and we slide the discounted sum across the value stream, yielding the sequence of discounted property values.
Without loss of generality, our monitors start observing the stream at time $0$. After each new observation at time $n$, the monitor must assess, for every past index $t\in[0;n]$, whether the discounted sum lies inside or outside a target interval $\aI$:
\begin{align*}
     \cdots,\; S_{t-2}^{r,s}(w),\; S_{t-1}^{r,s}(w),\; &S_t^{r,s}(w),\; S_{t+1}^{r,s}(w),\; S_{t+2}^{r,s}(w),\; \cdots\\
     &\uparrow \text{\small bound error}\\
     w_0, w_1, \cdots, w_{n-1}, w_n \;\rightarrow\; &\text{$\moni$onitor.}\;\rightarrow\; (\in \aI?) \text{-verdict}     
\end{align*}
The fundamental challenge is \emph{epistemic uncertainty}: the monitor must reason about a function over a bi-infinite stream after observing only $n$ values.
Within this setting, we propose a monitor that maintains, for each time index $t$, an \emph{uncertainty set} containing all possible values the full discounted sum at time $t$ could take. 
Although the uncertainty sets shrink as more information becomes available, we show that the time it takes to shrink enough to be completely inside or outside the target interval is unbounded.
As a solution, we study monitors with a tolerance of $\varepsilon$ in their verdict, and obtain an upper bound on the number of observations required to reach a decisive verdict, enabling monitoring with bounded resources at the cost of precision. 
We call these $\varepsilon$-approximate monitors.

\vspace{-0.5em}\subparagraph{Statistical Discounted Monitoring.}
In many applications, the value stream is a realisation of a stochastic process $W=(X_i)_{i\in\mathbb Z}$. 
In such settings, we are often interested in expected quantities rather than noisy observations, e.g., the latent temperature or the arbiter’s underlying propensity.
This leads to our second object of study: the \emph{expected discounted sum}, defined as the bi-directional discounted sum evaluated over conditional expectations,
\begin{equation}
     \edsum_t^{r,s}(W)\coloneqq
     \underbrace{\textstyle\sum_{i=1}^\infty r^i \expe_{t-i-1}(X_{t-i})}_{\text{past}}
     \;+\;
     \underbrace{\expe_{t-1}(X_t)}_{\text{present}}
     \;+\;
     \underbrace{\textstyle\sum_{i=1}^\infty s^i \expe_{t+i-1}(X_{t+i})}_{\text{future}}.
\end{equation}

\begin{example}[Bias Fairness]
    Consider a stochastic arbiter that generates its decision $X_t$ at time $t$ by tossing a coin with probability $P_t$, where $P_t$ is chosen based on the history up to time $t$.
    This yields two processes: the hidden bias process $U=(P_t)_{t\in\ZN}$ and the observed outcome process $W=(X_t)_{t\in\ZN}$, with $\expe_{t-1}(X_t)=P_t$.
    If we want to assess the arbiter around time $t$ independently of chance, we should evaluate discounted sums (or averages) with respect to the bias process, e.g., $S_t^{r,s}(U)=\edsum_t^{r,s}(W)$ rather than $S_t^{r,s}(W)$.
\end{example}

In statistical discounted monitoring, we slide both the expected and observed discounted sums over the stochastic stream, yielding expected and observed discounted property values.
As before, monitors start observing the realised sequence at time $0$. After each new observation at time $n$, the monitor must assess, for every past index $t\in[0;n]$, whether the \emph{expected} discounted sum lies inside or outside a target interval $\aI$ with high probability.
\begin{align*}
   \cdots,\; \edsum_{t-2}^{r,s}(W),\; \edsum_{t-1}^{r,s}(W),\; &\edsum_t^{r,s}(W),\; \edsum_{t+1}^{r,s}(W),\; \edsum_{t+2}^{r,s}(W),\; \cdots\\
   &\uparrow \text{\small bound deviation}\\
   \cdots,\; S_{t-2}^{r,s}(W),\; S_{t-1}^{r,s}(W),\; &S_t^{r,s}(W),\; S_{t+1}^{r,s}(W),\; S_{t+2}^{r,s}(W),\; \cdots\\
   &\uparrow \text{\small bound error}\\
   W_0, W_1, \cdots, W_{n-1}, W_n \;\rightarrow\; &\text{$\moni$onitor.}\;\rightarrow\; (\in \aI?) \text{-verdict}
\end{align*}
Apart from epistemic uncertainty, in this setting the extra challenge is \emph{statistical inference}: the monitor has to estimate also the deviation between the observed and expected discounted sums.
Because of the probabilistic nature of the monitored quantities,
the monitor's verdict is accompanied by an error probability $\delta$.
Within this setting, we extend $\varepsilon$-approximate monitoring to the probabilistic setting, and define three notions of $(\varepsilon,\delta)$-soundness with increasing strength of guarantees: pointwise, local, and uniform; depending on whether the probability of an error is given for each point after a fixed release time (pointwise), for each point with flexible release time (local), or for the whole execution (uniform). As the strength of the guarantee increases, the monitor's precision decreases. We show that the deviation bounds used by our pointwise sound monitors are minimax optimal, which establishes that statistical discounted monitoring is impossible beyond a certain precision.

\vspace{-0.5em}\subparagraph{Generalization to Arithmetic Expressions.}
Many relevant quantitative properties are naturally expressed as relations between discounted sums rather than as single aggregates. 
To support discounted versions of such properties, we introduce a richer specification language in which a monitor observes multiple outcome streams and monitors an arithmetic expression over their discounted sums.
We distinguish between synchronous and asynchronous semantics for combining discounted sums.
\begin{example}[Demographic Parity]
\label{ex:dem-parity}
    So far we have assumed that both clients always request access to the resource. We model the more general setting using the binary sequences representing:
    $w^{r_1}$ the requests of client 1, $w^{g_1}$ the grants client 1, $w^{r_2}$ the requests client 2, and $w^{g_2}$ the grants client 2. All indexed by the same global clock.
    Then we can express the discounted acceptance rate, i.e., a discounted form of demographic parity~\cite{dwork2012fairness}, as the difference between the ratios $\dsum_{t}^{r,s}(w^{g_1})/\dsum_{t}^{r,s}(w^{r_1})$ and 
    $\dsum_{t}^{r,s}(w^{g_2})/\dsum_{t}^{r,s}(w^{r_2})$. 
\end{example}

\vspace{-0.5em}\subparagraph{Register Complexity.}
The $\varepsilon$-approximate monitor from the deterministic setting yields a finite observation horizon; let $\tau$ denote the number of additional observations after which each monitored position receives a sound verdict up to the $\varepsilon$ boundary region.
This horizon is first a delay bound: a verdict for time $t$ may require waiting until time $t+\tau$.
It also serves as a memory bound: during this waiting period, several positions may remain unresolved simultaneously, and the monitor maintains a separate running discounted sum for each of them.
The delay bound, however, only says when verdicts are guaranteed to arrive; it does not characterize how much information must be stored before those verdicts are reached.
This leads to the register-complexity question: can the overlapping discounted sums in general be represented more compactly than by one running sum per pending position?

To make this question precise, we formalize monitors as \emph{affine register machines (ARMs)}: finite-state monitors equipped with real-valued registers, affine updates, and linear guards.
In this model, we answer this question negatively by proving a worst-case lower bound already for future-only monitoring.
For every $k$, there is a future-only monitor with horizon $k$ that cannot be monitored by any ARM with fewer than $k$ registers.

\vspace{-0.5em}\subparagraph{Experimental Evaluation.}
Finally, we implement our monitors on scenarios mirroring the motivating examples, and investigate empirically
how actual register use compares to our theoretical upper bounds, and how register usage differs between synchronous and asynchronous interpretations of arithmetic properties.
\vspace{-0.5em}\subparagraph{Contributions.}
The main contributions of this paper are as follows:
\begin{enumerate}
    \item A formalization of the discounted monitoring problem, and a proof that exact, sound discounted monitoring cannot be achieved with finite memory.
    \item An $\varepsilon$-approximately sound notion of discounted monitoring, accompanied by explicit bounds on the required number of observations and registers.
    \item An extension of discounted monitoring to stochastic processes via expected discounted sums, including pointwise and uniform $(\varepsilon,\delta)$-soundness notions, as well as results on statistical optimality and impossibility beyond a precision threshold.
    \item A specification language for arithmetic expressions over multiple discounted sums, supporting both synchronous and asynchronous semantics.
    \item A formalization of monitors as affine register machines and a tight bound on the number of registers required for sound $\varepsilon$-approximate monitoring in the future-only case.
    \item An implementation and empirical evaluation of the proposed monitors on realistic case studies, analyzing performance and efficiency relative to the proven worst-case guarantees.
\end{enumerate}

\section{Monitoring Discounted Sums}
\label{sec:monitoring-discounted-sums}

\subsection{Discounted Sum Property}

The main object of study in this paper is the discounted sum operator $S_t^{r,s}$, where $r,s\in [0,1)$ are the past and future discount factors and $t\in \ZN$ is a time index.
It acts on bi-infinite sequences $w=(x_t)_{t\in \ZN}$ taking values in a nonempty bounded set $\mathcal R \subset \RN$ of diameter $d_{\mathcal R}\coloneqq \sup \mathcal R - \inf \mathcal R$.
We call a tuple $(\mathcal R, r,s)$ an \emph{input setting}.
Intuitively, $S_t^{r,s}$ is centered at $t$, discounting past values ($i < t$) by $r$ and future values ($j > t$) by $s$.
Formally, for $n_l \leq t \leq n_u$, let $w_{n_l:n_u} = (x_{n_l}, \dots, x_{n_u})$ and define $S_t^{r,s}(w_{n_l:n_u}) \coloneqq \sum_{i=1}^{t-n_l} r^i x_{t-i} + x_t + \sum_{i=1}^{n_u-t} s^i x_{t+i}$.
For a bi-infinite sequence $w$, define $S_t^{r,s}(w) \coloneqq \sum_{i=1}^{\infty} r^i x_{t-i} + x_t + \sum_{i=1}^{\infty} s^i x_{t+i}$.

\vspace{-0.5em}\subparagraph{Basic properties.}
We prove some basic properties of discounted sums. 
First, the discounted sum converges, and consecutive sums differ by at most $d_{\mathcal R}$. 
This follows from the definition.

\begin{lemma}
    \label{lem:convergingsums}
    Let $(\mathcal R, r,s)$ be an input setting, $w\in \mathcal R^{\mathbb Z}$, and $t\in\ZN$.
    Then $S_t^{r,s}(w)$ exists and 
    $\big| S^{r,s}_t(w) - S^{r,s}_{t-1}(w)\big|\leq d_{\mathcal R}$.
\end{lemma}

Second, we characterize the values a discounted sum can take.
Formally, let $\mathrm{Sums}_{\mathcal R}^{r,s}\coloneqq\{ S_t^{r,s}(w)\::\: w\in\mathcal R^{\mathbb Z}\}$; note that this set does not depend on $t$.
Its shape depends on whether $\mathcal R$ has gaps that are too large relative to the discount factors.

\begin{definition}[Proper input setting]
    Let $(\mathcal R, r,s)$ be an input setting. 
    The \emph{normalization factor} $\lambda^{r,s}$ and the \emph{maximum gap} $\Delta_\mathcal R$ are defined as
    \begin{equation}
    \lambda^{r,s} \coloneqq 1+\frac{r}{1-r}+\frac{s}{1-s}, \qquad
    \Delta_\mathcal R \coloneqq
    \sup\Bigl(\{0\}\cup
    \{\,b-a:\ [a,b]\subseteq [\inf\mathcal R,\sup\mathcal R]\setminus \mathcal R\,\}\Bigr).
    \end{equation}
    We say that $(\mathcal R, r,s)$ is \emph{proper} if $\Delta_{\mathcal R}\leq \min\!\left\{\frac{r}{1-r},\frac{s}{1-s}\right\} d_{\mathcal R}$.
\end{definition}

The convention $\sup\{0\}=0$ ensures $\Delta_{\mathcal R}=0$ whenever $\mathcal R$ has no gaps in its interval hull, so interval-valued domains are automatically proper.
For non-interval domains such as $\mathcal R=\{0,1\}$, properness is a substantive condition: it requires the discount factors to be large enough to interpolate across the gaps.

\begin{lemma}
\label{lem:possible-sums}
    Let $(\mathcal R, r,s)$ be an input setting.
    \begin{enumerate}
    \item[(1)] $\mathrm{Sums}_{\mathcal R}^{r,s}\subseteq 
    \lambda^{r,s} \cdot [\inf\mathcal R,\, \sup\mathcal R]$.
    \item[(2)] If $(\mathcal R, r,s)$ is proper, then 
    $\mathrm{Closure}(\mathrm{Sums}_{\mathcal R}^{r,s}) =
    \lambda^{r,s} \cdot [\inf\mathcal R,\, \sup\mathcal R]$.
    \item[(3)] If $(\mathcal{R}, r, s)$ is proper and $\mathcal{R}$ is compact, then $\mathrm{Sums}_{\mathcal{R}}^{r,s} = \lambda^{r,s} \cdot [\inf \mathcal{R}, \sup \mathcal{R}]$.
    \end{enumerate}
\end{lemma}

The three statements isolate the role of each assumption.
Boundedness of $\mathcal{R}$ gives the interval over-approximation in~(1).
Properness ensures that the closure of the set of sums fills this interval, as stated in~(2).
Compactness lifts closure equality to equality in~(3).
For a non-closed domain, the endpoints need not be attainable; under compactness they are attained by the constant sequences at $\inf \mathcal{R}$ and $\sup \mathcal{R}$.

\vspace{-0.5em}\subparagraph{Sums vs. Averages.}
In some settings it is more natural to think of a discounted average than a discounted sum. 
A weighted average is a weighted sum whose weights are non-negative and sum to $1$.
Since the weights of a discounted sum add up to $\lambda^{r,s}$, an input setting $(\mathcal R, r,s)$ can be reframed for averages by considering the equivalent setting $(\mathcal R/\lambda^{r,s}, r, s)$ and rescaling each sequence $w = (x_t)_{t\in\mathbb Z}\in \mathcal R^{\mathbb Z}$ to $w' = (x_t/\lambda^{r,s})_{t\in\mathbb Z} \in (\mathcal R/\lambda^{r,s})^{\mathbb Z}$.

\subsection{Monitoring Problem}

A monitor observes a bi-infinite stream of values $w=(x_i)_{i\in\mathbb{Z}}$ from time~$0$ onward.
After observing a finite prefix $w_{0:n}$, it is queried only on observed positions $t\in\{0,\ldots,n\}$.
Formally, a monitor computes a function
$\moni\colon \mathcal{R}^* \times \mathbb{N} \to \{\bot, ?, \top\}$.
Only pairs $(u,t)$ with $t<|u|$ are queried.
The verdicts $\bot$ and $\top$ indicate that $S_t^{r,s}(w)$ lies outside or inside the target interval $\mathcal{I}$, respectively, and $?$ indicates that the verdict is inconclusive.

\subsubsection{Sound Monitors}

\begin{definition}[Sound monitor]
    Let $(\mathcal{R}, r, s)$ be an input setting and let $\mathcal{I}$ be a target interval.
    A monitor is \emph{sound} if, for every $w \in \mathcal{R}^{\mathbb{Z}}$, $n \in \mathbb{N}$, and $t \in \{0,\ldots,n\}$,
    \begin{equation}
        \label{eq:soundmonitor}
        \moni(w_{0:n}, t) = \bot \implies S_t^{r,s}(w) \notin \mathcal{I}
        \quad \text{and} \quad
        \moni(w_{0:n}, t) = \top \implies S_t^{r,s}(w) \in \mathcal{I}.
    \end{equation}
\end{definition}

\vspace{-0.5em}\subparagraph{Uncertainty.}
Since the monitor receives values incrementally, it typically returns an inconclusive verdict ($?$) until the observed prefix contains enough information to commit to a final $\top$ or $\bot$.
For $0 \leq t \leq n$, the \emph{uncertainty set} of a finite prefix $w_{0:n}$ collects the values $S_t^{r,s}$ can take across all extensions of $w_{0:n}$:
$
U_t(w_{0:n}) \coloneqq
\left\{ S_t^{r,s}(w') : w'\in \mathcal{R}^{\mathbb{Z}},\ w'_{0:n} = w_{0:n}\right\}.
$
A sound monitor therefore satisfies
$
\moni(w_{0:n}, t) = \top \implies U_t(w_{0:n}) \subseteq \mathcal{I}
$
and
$\moni(w_{0:n}, t) = \bot \implies U_t(w_{0:n}) \cap \mathcal{I} = \emptyset.
$
Similarly as in Lemma~\ref{lem:possible-sums}, for proper input settings the closure of the uncertainty set is an interval; the uncertainty set itself is that interval when $\mathcal{R}$ is compact.
\begin{lemma}
\label{lem:uncertainty-intervals}
    Let $(\mathcal{R}, r, s)$ be an input setting, let $0 \leq t \leq n$, and let $\gamma^{r,s}_{t,n} \coloneqq \frac{r^{t+1}}{1-r} + \frac{s^{n-t+1}}{1-s}$.
    Then, $U^{r,s}_t(w_{0:n}) \subseteq S^{r,s}_{t}(w_{0:n}) + \gamma^{r,s}_{t,n} [\inf \mathcal{R}, \sup \mathcal{R}]$.
    If $(\mathcal{R}, r, s)$ is proper, then $\mathrm{Closure}({U^{r,s}_t(w_{0:n})}) = S^{r,s}_{t}(w_{0:n}) + \gamma^{r,s}_{t,n} [\inf \mathcal{R}, \sup \mathcal{R}]$.
    If, additionally, $\mathcal{R}$ is compact, then equality holds without taking closures.
\end{lemma}

A natural quantity is then how long a sound monitor must wait before it can commit to a final verdict at a given time.

\begin{definition}
    Let $\mathcal{I}$ be a target interval, let $t \in \mathbb{N}$, and let $w \in \mathcal{R}^{\mathbb{Z}}$.
    The \emph{minimum required observation time} is
    \begin{equation*}
        \mu_{\mathcal{I}}(t,w)
        \coloneqq
        \min\left(
        \left\{
        n \in \mathbb{N} :
        n \geq t
        \text{ and }
        \left(
        U_t(w_{0:n}) \subseteq \mathcal{I}
        \text{ or }
        U_t(w_{0:n}) \cap \mathcal{I} = \emptyset
        \right)
        \right\}
        \cup
        \{\infty\}
        \right).
    \end{equation*}
\end{definition}

For a fixed time $t$, future observations eliminate the unobserved future contribution, but they do not reveal the values before time~$0$.
Therefore, an irreducible uncertainty of diameter $\frac{r^{t+1}}{1-r} d_{\mathcal{R}}$ may remain, and an exact verdict need not ever become possible.
Even when exact verdicts are possible, their delay has no uniform finite bound in general.

\begin{theorem}[No uniform exact verdict delay]
\label{thm:impossible-sound}
    Let $(\mathcal{R}, r, s)$ be a proper input setting such that $\mathcal{R}$ is compact and $r + s > 0$, and let $J \coloneqq \lambda^{r,s} [\inf \mathcal{R}, \sup \mathcal{R}]$.
    Let $\mathcal{I}$ be a target interval satisfying $J \cap \mathcal{I} \neq \emptyset$ and $J \setminus \mathcal{I} \neq \emptyset$.
    Then, for every $t \in \mathbb{N}$ and every $n \geq t$, there exists $w \in \mathcal{R}^{\mathbb{Z}}$ such that $\mu_{\mathcal{I}}(t,w) > n$.
\end{theorem}

The proof shows that for any proposed finite delay, one can construct a prefix whose remaining uncertainty crosses a boundary of the target interval: one completion yields a discounted sum in the target, while another yields one outside it
A sound exact monitor must therefore remain inconclusive, which motivates approximate soundness.

\subsubsection{Approximately Sound Monitors}

\begin{definition}[$\varepsilon$-approximately sound monitor]
    Let $(\mathcal{R}, r, s)$ be an input setting, let $\mathcal{I}$ be a target interval, and let $\varepsilon > 0$.
    A monitor $\mathcal{M}$ is \emph{$\varepsilon$-approximately sound} if, for every $w \in \mathcal{R}^{\mathbb{Z}}$, $n \in \mathbb{N}$, and $t \in \{0,\ldots,n\}$,
    \begin{equation}
        \label{eq:epsilonsoundmonitor}
        \moni(w_{0:n}, t) = \bot \implies S_t^{r,s}(w) \notin \mathcal{I}_{-\varepsilon}
        \quad \text{and} \quad
        \moni(w_{0:n}, t) = \top \implies S_t^{r,s}(w) \in \mathcal{I}_{+\varepsilon}.
    \end{equation}
    where, for $\mathcal{I}=(L,U)$,
    $\mathcal{I}_{+\varepsilon} \coloneqq (L-\varepsilon, U+\varepsilon)$
    and
    $\mathcal{I}_{-\varepsilon} \coloneqq (L+\varepsilon, U-\varepsilon)$.
\end{definition}

Although $\aI$, $\aI_{+\varepsilon}$, and $\aI_{-\varepsilon}$ are open, we note that the results of this section hold for arbitrary intervals.
We next give a tight bound on the minimum required observation time for $\varepsilon$-approximately sound monitors, and show how to realise them with finite resources.

\begin{theorem}
\label{thm:approximate-monitors}
    Let $(\mathcal{R}, r, s)$ be an input setting, let $\mathcal{I}$ be a target interval, let $\varepsilon > 0$, and let $T \in \mathbb{N}$.
    Define
    \begin{equation}
        \label{eq:approximate-monitors-tau}
        \tau^*(\mathcal{R}, r, s, \varepsilon, T)
        \coloneqq
        \min\left(
        \left\{
        \tau \in \mathbb{N} :
        d_{\mathcal{R}}
        \left(
        \frac{r^{T+1}}{1-r}
        +
        \frac{s^{\tau+1}}{1-s}
        \right)
        \leq
        2\varepsilon
        \right\}
        \cup
        \{\infty\}
        \right).
    \end{equation}
    Then:
    \begin{enumerate}
        \item If $\tau^* < \infty$, there exists an $\varepsilon$-approximately sound monitor that produces a verdict for every $t \geq T$ within $\tau^*$ steps.
        \item Suppose that $\mathcal{R} = [m, M]$ with $m < M$, and write $J = \lambda^{r,s} [m, M]$.
        If $J \cap \mathcal{I}_{-\varepsilon} \neq \emptyset$ and $J \setminus \mathcal{I}_{+\varepsilon} \neq \emptyset$, then, for every $\varepsilon$-approximately sound monitor and every $\tau < \tau^*$, there exists $w \in \mathcal{R}^{\mathbb{Z}}$ on which the monitor does not produce a verdict for time $T$ after $\tau$ steps.
    \end{enumerate}
\end{theorem}

For the usual case $d_{\mathcal{R}} > 0$ and $0 < s < 1$, let $B_T \coloneqq \left(\frac{2\varepsilon}{d_{\mathcal{R}}} - \frac{r^{T+1}}{1-r}\right)(1-s)$.
If $B_T > 0$, then $\tau^* = \max\{0, \left\lceil\log_s B_T\right\rceil - 1\}$.
If $B_T \leq 0$, then $\tau^* = \infty$.
The semantic definition in Eq.~\eqref{eq:approximate-monitors-tau} also covers the exceptional cases: $d_{\mathcal{R}} = 0$ gives $\tau^* = 0$, while for $s = 0$ we have $\tau^* = 0$ if $d_{\mathcal{R}} \frac{r^{T+1}}{1-r} \leq 2\varepsilon$ and $\tau^* = \infty$ otherwise.
Unless stated otherwise, we assume $0 < \tau^* < \infty$. 

For every bounded input domain, Lemma~\ref{lem:uncertainty-intervals} provides an interval enclosure whose diameter at time $t$, after $\tau$ further observations, is $d_{\mathcal{R}} \left(\frac{r^{t+1}}{1-r} + \frac{s^{\tau+1}}{1-s}\right)$.
A diameter of at most $2\varepsilon$ forces an approximately sound verdict, which proves the upper bound.
For interval-valued input domains, the finite observed contribution can be positioned continuously; under the two nontriviality conditions in point~(2), this gives matching accepting and rejecting completions for every smaller delay.

\subsection{Monitor Construction}

Algorithm~\ref{alg:monitor} constructs an $\varepsilon$-approximately sound monitor.
With $\tau=\tau^*(\mathcal{R},r,s,\varepsilon,T)$ from Eq.~\eqref{eq:approximate-monitors-tau}, it maintains:
\begin{itemize}
    \item a global variable $\mathtt{Psum}$ storing the present-inclusive discounted sum of all past observations;
    \item an array of $\tau$ tuples $(\mathtt{RunningSum}_i, \mathtt{pos}_i)$, one per active monitor;
    \item a Boolean array $\mathtt{in\_use}[1{:}\tau]$ marking which tuples are active.
\end{itemize}
Each active tuple tracks a candidate time $\mathtt{pos}_i$ whose discounted sum is being monitored.
The invariant is that after observing $x_n$, every active tuple with
$p=\mathtt{pos}_j$ satisfies
$
\mathtt{RunningSum}_j
=
S_p^{r,s}(w_{0:n})
=
\sum_{i=1}^{p} r^i x_{p-i}
+
\sum_{i=0}^{n-p} s^i x_{p+i}.
$

\begin{algorithm}[H]
\caption{$\varepsilon$-Approximately Sound Discounted-Sum Monitor}
\label{alg:monitor}
\begin{algorithmic}[1]
\Require Input setting $(\mathcal R, r, s)$, tolerance $\varepsilon$, target interval $\mathcal I$, start time $T$
\State $\tau \gets \tau^*(\mathcal{R},r,s,\varepsilon,T)$ \Comment{Eq.~\eqref{eq:approximate-monitors-tau}}
\State $\mathtt{Psum} \gets 0$
\For{$i = 1$ to $\tau$}
    \State $\mathtt{in\_use}[i] \gets \bot$
\EndFor

\For{$t = 0,1,2,\dots$}
    \State observe $x_t$; \;\;
    $\mathtt{Psum} \gets x_t + r \cdot \mathtt{Psum}$
    \Comment{now $\mathtt{Psum}=x_t+\sum_{i=1}^{t} r^i x_{t-i}$}

    \For{$j = 1$ to $\tau$} \Comment{Update all active registers}
        \If{$\mathtt{in\_use}[j] = \top$}
            \State $\mathtt{RunningSum}_j \gets \mathtt{RunningSum}_j + x_t \cdot s^{\,t-\mathtt{pos}_j}$ 
            \Comment{\footnotesize{future contribution for time $\mathtt{pos}_j$}\normalsize}
            \State compute
            $
            C_{\mathtt{pos}_j}^{r,s} =
            \mathtt{RunningSum}_j
            + \left(
            \frac{r^{\mathtt{pos}_j+1}}{1-r}
            + \frac{s^{t-\mathtt{pos}_j+1}}{1-s}
            \right)
            \cdot [\inf\mathcal R,\, \sup\mathcal R]
            $
            \If{$C_{\mathtt{pos}_j}^{r,s} \subseteq \mathcal I_{+\varepsilon}$}
                \State output verdict $\top$ for time $\mathtt{pos}_j$ ;\;\;
                $\mathtt{in\_use}[j] \gets \bot$
            \ElsIf{$C_{\mathtt{pos}_j}^{r,s} \cap \mathcal I_{-\varepsilon} = \emptyset$}
                \State output verdict $\bot$ for time $\mathtt{pos}_j$ ;\;\;
                $\mathtt{in\_use}[j] \gets \bot$
            \EndIf
        \EndIf
    \EndFor

    \If{$t \geq T$} \Comment{Instantiate new register for current $t$}
        \State select $i$ with $\mathtt{in\_use}[i] = \bot$
        \Comment{Guaranteed to exist by Thm.~\ref{thm:approximate-monitors}}
        \State $\mathtt{RunningSum}_i \gets \mathtt{Psum}$;\; \;
        \Comment{initially $S_t^{r,s}(w_{0:t})$}
        $\mathtt{pos}_i \gets t$;\;\;
        $\mathtt{in\_use}[i] \gets \top$
    \EndIf
    
\EndFor
\end{algorithmic}
\end{algorithm}

\subparagraph{Special cases.}
For future-only sums ($r = 0$), the past uncertainty vanishes.
For past-only sums ($s = 0$), later observations cannot refine the uncertainty of an earlier position.
Accordingly, the uniform bound in Eq.~\eqref{eq:approximate-monitors-tau} is either $0$ or $\infty$: each position is tested once at creation time, and if its enclosure is still inconclusive, no later observation can change that.

\subparagraph{General bounded input domains.}
Neither the upper bound in Theorem~\ref{thm:approximate-monitors}(1) nor Algorithm~\ref{alg:monitor} requires properness.
Both use only the interval enclosure from Lemma~\ref{lem:uncertainty-intervals}.
The interval-domain assumption in Theorem~\ref{thm:approximate-monitors}(2) is used only to position the observed finite contribution continuously.
For input domains with gaps, $\tau^*$ remains a valid upper bound but need not be tight for a fixed target interval.

\section{Statistical Discounted Monitor}
\label{sec:statistical-monitoring}
In statistical discounted  monitoring, we assume that the observation sequence is generated by a stochastic process and the monitors objective is to estimate the value of the expected discounted sum. The statistical monitors are similar to the monitors of Section~\ref{sec:monitoring-discounted-sums}, but for an additional statistical error term which impacts the release condition (and time $\tau^*$).

\vspace{-0.5em}\subparagraph{Setting.}
Assume that the value sequence $w$ is a realisation of the stochastic process $W=(X_t)_{t\in \ZN}$ in $\mathcal{R}$. Let $\filt\coloneqq (\filt_t)_{t\in \ZN}$ be its canonical filtration, i.e., $\filt_t$ is the sigma-algebra generated by $(X_i)_{i\leq t}$. \emph{No further assumptions} are placed on $W$.

\vspace{-0.5em}\subparagraph{Expected discounted sum.}
We are interested in the expected discounted sum, in infinite and finite form. We define the conditional expectation for the integrable random variable $Y$ as $\expe_t(Y)\coloneqq \expe(Y\mid \filt_t)$.
We define $\edsum_t^{r, s}(W) \coloneqq \lim_{n_l\to -\infty} \lim_{ n_u\to +\infty} \edsum_t^{r, s}(W_{n_l:n_u})$  
\begin{align*} 
    \edsum_t^{r, s}(W_{n_l:n_u})
    &\coloneqq
    \sum_{i=1}^{t-n_l} r^{i}\, \expe_{t-i-1}(X_{t-i})
    + \sum_{i=0}^{n_u-t} s^i\, \expe_{t+i-1}(X_{t+i}) \quad  
    \text{for $n_l\leq t\leq n_u$, $r,s\in[0,1)$}.
\end{align*}

\subsection{Monitoring}
Statistical monitors compute the observed discounted sum and bound its deviation from the expected discounted sum. This permits three natural notions of statistical soundness, which differ in how many possible verdicts are protected by the same probability guarantee.

\begin{definition}[Statistical soundness]
    For $(\mathcal{R},r,s)$, let $\aI=(L,U)$ be a target interval, $\varepsilon\geq 0$ a precision, and $\delta\in(0,1)$ an error probability.
    For $t\leq n \in \NN$ define the success event
    \begin{align*}
       E_t(n) \coloneqq
       \left\{
       \moni(W_{0:n}, t)=\bot \Rightarrow \edsum_t^{r,s}(W)\notin \aI_{-\varepsilon}
       \land
       \moni(W_{0:n}, t)=\top \Rightarrow \edsum_t^{r,s}(W)\in \aI_{+\varepsilon}
       \right\},
    \end{align*}
    where $\aI_{+\varepsilon}\coloneqq (L-\varepsilon,U+\varepsilon)$ and $\aI_{-\varepsilon}\coloneqq (L+\varepsilon,U-\varepsilon)$.
    A statistical monitor is
    \begin{align*}
         &\text{(i) pointwise $(\varepsilon,\delta)$-approximately sound if }
         &&\forall t\in\NN\,\forall n\geq t\colon \prob(E_t(n)) \geq 1-\delta;\\
        &\text{(ii) locally $(\varepsilon,\delta)$-approximately sound if }
        &&\forall t\in\NN\colon \prob(\forall n\geq t\colon E_t(n)) \geq 1-\delta;\\
        &\text{(iii) uniformly $(\varepsilon,\delta)$-approximately sound if }
        &&\prob(\forall t\in\NN\,\forall n\geq t\colon E_t(n)) \geq 1-\delta.
    \end{align*}
\end{definition}

The three notions express increasingly stronger guarantees.
\emph{Pointwise soundness} protects one fixed verdict at a fixed time $t$ and observation horizon $n$. Hence, flexible register release is \underline{not} possible and an error probability of $\delta$ must be tolerated for the verdict at each $t$.
\emph{Local soundness} protects one fixed monitored time $t$ for all observation horizons $n\geq t$. Hence, flexible register release is possible, but an error probability of $\delta$ remains for the verdict at each $t$.
\emph{Uniform soundness} protects the entire run, i.e., the invariant ``every verdict issued for every time index and at every observation horizon is correct'' holds with probability at least $1-\delta$. This is the right notion when the correctness of verdicts is imperative.

\subsection{Monitor Construction}
We modify the deterministic monitor by adding a statistical error term to the uncertainty interval (which affects the register release bound $\tau^*$
). The resulting statistical uncertainty interval for the expected discounted sum for time $t$ and observations $n$ is 
\begin{equation}
    \label{eq:stat:unc_set}
    \eset_t^{r,s}(\delta;W_{0:n})
    \coloneqq
    \einter_t^{r,s}(\delta;W_{0:n}) + \gamma_{t,n}^{r,s}\cdot \Gamma  \quad \text{where} \quad  \einter_t^{r,s}(\delta;W_{0:n})  \coloneqq \dsum_t^{r,s}(W_{0:n})
    \pm \beta_{t,n}^{r,s}(\delta) .
\end{equation}
where $\beta_{t,n}^{r,s}(\delta)$ is a statistical error term and $\gamma_{t,n}^{r,s}$ is a tail error as in Lemma~\ref{lem:uncertainty-intervals}, and $\Gamma\coloneqq [\inf \mathcal{R},\, \sup\mathcal{R}]$. 
To obtain soundness it suffices to ensure that $\einter_t^{r,s}(\delta;W_{0:n})$ covers the finite expected sum $\edsum_t^{r,s}(W_{0:n})$ with the desired probability guarantee. Then the deterministic tail term $ \gamma_{t,n}^{r,s} \cdot \Gamma$ lifts this to coverage of the infinite expected sum (see Lemma~\ref{lemma:stat:lift}).

\vspace{-0.5em}\subparagraph{Error bounds.}
For simplicity we focus on data-independent deviation bounds, but the results can be extended to variance adaptive bounds~\cite{howard2021time}. 
Hence, the width of the statistical error term is dictated by the squared discount factors, i.e., for every $n\in \NN$ and $t\leq n$ as
\begin{align*}
    \dsgn_{t,n}^{r,s}
    \coloneqq
    \sum_{i=1}^{t} r^{2i} + \sum_{i=0}^{n-t} s^{2i}
    =
    \frac{r^2(1-r^{2t})}{1-r^2} + \frac{1-s^{2(n-t+1)}}{1-s^2},
\end{align*}
and a uniform upper bound $\sgn$ on the conditional sub-Gaussian norm $\sgn_t$ of $X_t-\expe_{t-1}(X_t)$, which is trivially given by $d_\aR/2$, i.e., $\sgn_t\leq \sgn \leq d_\aR/2$~\cite{vershynin2018high}.
We define the pointwise, local, and uniform error bound, respectively, for $\delta\in (0,1)$ as
\begin{align*}
     \pse_{t,n}^{r,s}(\delta) &\coloneqq \sqrt{2\,\sgn^2\, \dsgn_{t,n}^{r,s}\, \log(2/\delta)}  \quad 
    \lse_{t,n}^{r,s}(\delta) \coloneqq k_1\,\sqrt{V_{t,n}^{r,s} \left(2\log\left(\log_2(V_{t,n}^{r,s} )+1\right) + \log\left(\tfrac{2\pi^2}{6\delta}\right)\right)} \\
     \use_{t,n}^{r,s}(\delta) &\coloneqq  \lse_{t,n}^{r,s}\left(\frac{6\delta}{\pi^2(t+1)^2}\right) 
     \quad \text{where $V_{t,n}^{r,s}\coloneqq \max(1,\sgn^2\,\dsgn_{t,n}^{r,s})$ and $k_1 \coloneqq 2^{1/4}+2^{-1/4}/\sqrt 2$.}
\end{align*}
The pointwise sound monitor leverages the fixed-sample Hoeffding--Azuma deviation bound~\cite{azuma1967weighted}. 
The locally sound monitor leverages the anytime-valid deviation bounds from Howard et al.~\cite{howard2021time}.
The uniformly sound monitor leverages the locally sound statistical error bounds with error level $6\delta/(\pi^2(t+1)^2)$ for time $t$, because a union bound over all monitored times results in a uniform guarantee, i.e., $  \sum_{t=0}^{\infty}6\delta/(\pi^2(t+1)^2)=\delta $.
Unsurprisingly, we can observe that as the guarantees become stronger, the interval width becomes wider.

\begin{restatable}{theorem}{Soundness}
    \label{lemma:stat:soundness}
    Let $(\mathcal{R},r,s)$ be an input setting, $\aI=(L,U)$ a target interval,
    $\varepsilon\geq 0$, and $\delta\in(0,1)$.
    Consider Algorithm~\ref{alg:monitor} where a register
    for time $t$ is released at a fixed ($n$) or at a flexible ( $ \eset_t^{r,s}(\delta;W_{0:n})\subseteq \aI_{+\varepsilon}$ or $\eset_t^{r,s}(\delta;W_{0:n})\cap \aI_{-\varepsilon}=\emptyset$) observation horizon. Then
    
    \noindent
    (i) if $\beta_{t,n}^{r,s}(\delta)=\pse_{t,n}^{r,s}(\delta)$ the monitor is pointwise
        $(\varepsilon,\delta)$-approx.\ sound for fixed release;

    \noindent
    (ii) if $\beta_{t,n}^{r,s}(\delta)=\lse_{t,n}^{r,s}(\delta)$, the monitor is
        locally $(\varepsilon,\delta)$-approx.\ sound
        for flexible release;

    \noindent
    (iii) if $\beta_{t,n}^{r,s}(\delta)=\use_{t,n}^{r,s}(\delta)$, the monitor is
        uniformly $(\varepsilon,\delta)$-approx.\ sound for flexible release.
\end{restatable}

\vspace{-0.5em}\subparagraph{Impossibility.}
If we consider the statistical uncertainty intervals 
as $n\to \infty$, we observe that the tail error $\gamma_{t,n}^{r,s}$ converges to $0$ from above and the statistical error $\beta_{t,n}^{r,s}(\delta)$ converges to a non-zero constant from below. In particular, for the pointwise bound we have
\begin{align*}
    \pse_{t,n}^{r,s}(\delta)\to
    \sqrt{2\sgn^2\left(\tfrac{r^2(1-r^{2t})}{1-r^2}
    +\tfrac{1}{1-s^2}\right)\log(2/\delta)}
    \quad\text{and}\quad
    d_\aR\,\gamma_{t,n}^{r,s}\to d_\aR\,\tfrac{r^{t+1}}{1-r}
    \quad\text{as $n\to\infty$.}
\end{align*}
Without additional assumptions this error is unavoidable. Specifically, the pointwise statistical term (the smallest among the soundness notions
) is minimax optimal. Intuitively, this is because the variance of the discounted average converges to a positive constant and to $0$ only as the discount factors approach $1$.

\begin{restatable}{theorem}{StatMinMax}
    \label{trm:minmax}
    Given an input setting $(\mathcal{R},r,s)$ s.t.\ $\mathcal{R}$ contains an interval of length at least $2\sigma$. Fix $n\in\NN$ and $t\leq n$.
    Consider the class $\mathcal{D}_\sgn$ of product measures on $\mathcal{R}^{n+1}$ s.t.\ $X_i-\expe_{i-1}(X_i)$ is conditionally
    $\sgn$-sub-Gaussian.
    We define: the set of all $(1-\delta)$ confidence intervals $\mathsf{CI}_\sgn(\delta)$ for
    $\edsum_t^{r,s}(W_{0:n})$, i.e., all
    $I:\mathcal{R}^{n+1}\to \mathrm{Interval}(\RN)$ satisfying 
    $  \sup_{P\in\mathcal{D}_\sgn}
        \prob_{W\sim P}\left(\edsum_t^{r,s}(W_{0:n})\notin I(W_{0:n})\right)
        \leq \delta$,
    $|I|=\sup I-\inf I$ as interval length, and
    $\eta_{t,n}^{r,s}=\sum_{i=1}^{t}r^i+\sum_{i=0}^{n-t}s^i$.
    There exist universal constants $0<c_0<C_0<\infty$ s.t.\ $\delta\in(0,1/4)$
    \begin{align*}
    c_0 \cdot
    \min\left(\sgn\eta_{t,n}^{r,s}, \pse_{t,n}^{r,s}(\delta)\right)
    &\leq
    \inf_{I\in\mathsf{CI}_\sgn(\delta)}
    \sup_{P\in\mathcal{D}_\sgn}
    \expe_P\left[|I(W_{0:n})|\right] \leq
    C_0 \cdot
    \min\left(d_\aR\eta_{t,n}^{r,s}, \pse_{t,n}^{r,s} (\delta)\right).
    \end{align*}
\end{restatable}

\section{General Discounted Properties}

Many quantitative properties involve arithmetic relations between multiple discounted aggregates.
A canonical example is (group) fairness, where one compares \emph{acceptance rates} across groups; such rates are quotients of discounted counts of accepted and total events.
This section shows how to use the monitors of Section~\ref{sec:monitoring-discounted-sums} as building blocks for monitoring multi-aggregate expressions.

\subparagraph{Setting.}
Let $\mathcal N \coloneqq [n]$, $n\in \mathbb N^{+}$, be the set of event types, and let
$(\mathcal R,r,s)$ be an input setting.
Let $\mathcal R_{\empt}\coloneqq \mathcal R\cup\{\empt\}$ extend $\mathcal R$ with a
distinguished empty symbol $\empt$.
An \emph{event} is a vector $\bm e\in \mathcal E\coloneqq (\mathcal R_{\empt})^{n}$,
where the $k$-th component $e^{(k)}\in \mathcal R_{\empt}$ is the \emph{atomic event of type $k$}.
We consider bi-infinite event streams $w=(\bm e_t)_{t\in\mathbb Z}\in \mathcal E^{\mathbb Z}$.

\subsection{Specification Language}
\label{subsec:spec-language}

We introduce a simple expression language whose atoms are (bi-directional) discounted sums over individual event types and whose connectives are the standard arithmetic operations.
We give two interpretations---\emph{synchronous} and \emph{asynchronous}---differing in how empty atomic events affect discounting.

\subparagraph{Syntax.}
Formulas are built from discounted-sum atoms by scalar multiplication, addition, multiplication, and division:
\begin{equation}
\label{eq:syntax-beyond-sums}
\varphi \Coloneqq\;
\mathcal S^{(k)}_{r,s}
\;\mid\;
c\cdot \varphi
\;\mid\;
\varphi+\varphi
\;\mid\;
\varphi\cdot \varphi
\;\mid\;
\varphi \div \varphi,
\qquad
\text{where } k\in\mathcal N,\; c\in\mathbb R.
\end{equation}

\subparagraph{Semantics: atoms.}
Consider an event stream $w=(\bm e_t)_{t\in\mathbb Z}$ and time $t\in\mathbb Z$.
We map empty atomic events to zero via $x_t^{(k)} \coloneqq \1[e_t^{(k)}\neq \empt]\cdot e_t^{(k)}$ for $t\in\mathbb Z$, $k\in\mathcal N$.
The value of $\mathcal S^{(k)}_{r,s}$ at time $t$ is a bi-directional discounted sum over the $k$-typed atomic values, with discount exponents depending on the interpretation:
\begin{equation}
\label{eq:atom-semantics}
\eval{\mathcal S^{(k)}_{r,s}}_{sync/async}(w,t)
\;\coloneqq\;
\underbrace{\textstyle\sum_{i=1}^{\infty} r^{\tau^{(k)}_{t,-i}}\; x_{t-i}^{(k)}}_{\text{past}}
\;+\;
\underbrace{x_t^{(k)}}_{\text{present}}
\;+\;
\underbrace{\textstyle\sum_{i=1}^{\infty} s^{\tau^{(k)}_{t,+i}}\; x_{t+i}^{(k)}}_{\text{future}}.
\end{equation}
The discount exponents $\tau^{(k)}_{t,\pm i}$ are defined as
$\tau_{t,\pm i}^{(k)} \coloneqq i$ in the synchronous interpretation, and $\tau_{t,\pm i}^{(k)} \coloneqq  \sum_{j=1}^{i}\1[e_{t\pm j}^{(k)}\neq \empt]$ in the asynchronous one.
Thus, in the synchronous interpretation, every time step advances discounting, even if the $k$-typed atomic event is empty.
In the asynchronous interpretation, only \emph{non-empty} atomic events advance discounting for type $k$; empty entries are ignored for the purpose of discount progression.

Note that, if $0\in \mathcal R$, any stream $w\in \mathcal E^{\mathbb Z}$ can be converted to $w'\in \mathcal E^{\mathbb Z}$ such that $\eval{\mathcal S^{(k)}_{r,s}}_{sync}(w,t) = \eval{\mathcal S^{(k)}_{r,s}}_{async}(w',t)$, and vice versa.

\subparagraph{Semantics: expressions.}
The semantics extends naturally to arithmetic operations:
for a stream $w$, time $t$, and expressions $\psi,\chi$,
\[
\eval{c\cdot \psi}(w,t) \;=\; c\cdot \eval{\psi}(w,t),
\qquad
\eval{\psi\circ\chi}(w,t)\;=\;\eval{\psi}(w,t)\circ \eval{\chi}(w,t)
\quad(\circ\in\{+,\cdot,\div\}).
\]
Note that division is undefined when the denominator evaluates to $0$.

\subsection{Uncertainty Propagation Through Interval Arithmetic}
\label{subsec:uncertainty-propagation}

Algorithm~\ref{alg:monitor} naturally extends to arbitrary formulas by instantiating a register for each atomic discounted sum and propagating uncertainty intervals via standard interval arithmetic.

Theorem~\ref{thm:approximate-monitors} bounds the observation time needed for an $\varepsilon$-approximately sound verdict on a single discounted sum.
A natural question is whether an analogous bound holds for arbitrary formulas.
The proof of Theorem~\ref{thm:approximate-monitors} relies on the fact that, for fixed $t$, the uncertainty interval of a discounted sum has length $\gamma^{r,s}_{t,t+\tau}(\sup\mathcal R-\inf\mathcal R)$ and depends only on the observation count $\tau$, allowing to identify the earliest $\tau$ at which uncertainty falls below $2\varepsilon$.

The first observation is that this type of reasoning can only be applied for the synchronous interpretation, as in the asynchronous interpretation, a stream can contain arbitrary many empty events for each atom, which make it impossible to give a minimum required observation time.
For the synchronous interpretation, we have to develop a theory for different fragments of the whole language.
In a nutshell, 
analogous results to Thm.~\ref{thm:approximate-monitors} (1) exist for expressions that contain no divisions, by considering how the uncertainty intervals of sums and products grow through interval arithmetic. This avoids the undefined divisions by 0.
When $\mathcal R$ is an interval and $\varphi$ contains only sums and scalar products, the analogous result to Thm.~\ref{thm:approximate-monitors} (2) also holds, as monitoring such an expression is equivalent to monitoring an atom on a larger input setting.

We start by defining interval arithmetic as ususal. Let $I_1 = [L_1, U_1]$, $I_2 = [L_2, U_2]$, and $c\in \mathbb R$. Then
\begin{itemize}
    \item $c\cdot I_1 = [cL_1, cU_1]$ if $c\geq 0$,  $c\cdot I_1 = [cU_1, cL_1]$ if $c<0$.
    \item $I_1 + I_2 = [L_1+L_2, U_1 + U_2]$.
    \item $I_1\cdot I_2 = [\min\{L_1\cdot L_2, L_1\cdot U_2, U_1\cdot L_2, U_1\cdot U_2 \}, 
    \max\{L_1\cdot L_2, L_1\cdot U_2, U_1\cdot L_2, U_1\cdot U_2 \}]$.
    \item If $0\notin I_2$, then $I_1/ I_2 = [\min\{L_1/ L_2, L_1/ U_2, U_1/ L_2, U_1/ U_2 \}, 
    \max\{L_1/ L_2, L_1/ U_2, U_1/ L_2, U_1/ U_2 \}]$.
\end{itemize}
The analogous definitions apply when the intervals are open on either end.

\begin{definition}[Linear and multiplicative properties]
    Let $\varphi$ be a general discounted property as defined in Eq.~\ref{eq:syntax-beyond-sums}.
    We say that $\varphi$ is \emph{linear} if it contains only atoms, additions, and scalar multiplications.
    We say that $\varphi$ is \emph{multiplicative} if it contains only atoms, additions, scalar multiplication, and products.
\end{definition}
We define the spread of a formula as the range of values it can take at each timestep.
\begin{definition}[Spread]
    Let $(\mathcal R, r,s)$ be an input setting and $\nu = \inf_{x\in \mathcal R}|x|$.
    For the spread of a formula $\varphi$ to be well defined, we assume that either $\varphi$ is multiplicative or $\nu > 0$.
    The spread of a formula $\varphi$ is defined recursively as:
    \begin{itemize}
        \item If $\varphi$ is an atom, $\mathrm{Spread}_{\mathcal R}(\varphi) = \mathcal R$.
        \item If $\varphi = c\cdot \psi$ for $c\in \mathbb R$, $\mathrm{Spread}_{\mathcal R}(\varphi) = c\cdot \mathrm{Spread}_{\mathcal R}(\psi)$.
        \item If $\varphi = \psi \circ \chi$, then $\mathrm{Spread}_{\mathcal R}(\varphi) = \mathrm{Spread}_{\mathcal R}(\psi)\circ\mathrm{Spread}_{\mathcal R}(\chi)$, for $\circ \in \circ\in\{+,\cdot,\div\}$.
    \end{itemize}
\end{definition}
\begin{remark}
    Because of the interval arithmetic, the expression of $\mathrm{Spread}(\varphi)$ is generally convoluted. 
    A special case is when $\varphi$ is linear. If $\varphi$ is linear
    it can be written as
    \[
    \varphi = \sum_{i=1}^{m^-} \alpha_i^-\cdot x_i^- +
    \sum_{i=1}^{m^+} \alpha_i^+\cdot x_i^+.
    \]
    Therefore, for $\mathcal R = [\inf R, \sup R]$, its spread can be written as
    \[
    \mathrm{Spread}_{\mathcal R}(\varphi) = \left[
    \sup R\cdot \sum_{i=1}^{m^-} \alpha_i^- + \inf R\cdot \sum_{i=1}^{m^+} \alpha_i^+,
    \inf R\cdot \sum_{i=1}^{m^-} \alpha_i^- + \sup R\cdot \sum_{i=1}^{m^+} \alpha_i^+
    \right].
    \]
\end{remark}

Using synchronous semantics, monitoring a linear expression $\varphi$ on an input setting $(\mathcal R, r,s)$ is equivalent to monitoring an atom on the input setting $(\mathcal R_\varphi, r,s)$ where $\mathcal R_\varphi = \mathrm{Spread}_{\mathcal R\cup \{0\}}(\varphi)$.
Therefore, we have the following corollary to Theorem~\ref{thm:approximate-monitors}.

\begin{corollary}
    \label{cor:approximate-monitors}
    Let $(\mathcal R, r,s)$ be an input setting, $\varepsilon>0$ a tolerance value, $T\in \mathbb N$ a point in time, $\varphi$ a linear formula.
    The following holds for $\tau^*$ defined as:
    \begin{equation}
    \label{eq:approximate-monitors-tau-cor1}
        \tau^*(\mathcal R, r,s,\varepsilon,T, \varphi) \coloneqq \left\lceil \log_s\left[\left(\frac{2\varepsilon}{\sup\mathcal R_\varphi - \inf \mathcal R_\varphi} - \frac{r^{T+1}}{1-r}\right)(1-s)\right]\right\rceil-1.
    \end{equation}
    \begin{enumerate}
        \item There exists an $\varepsilon$-approximately sound monitor $\moni^*$ that produces a verdict for all $t\geq T$ in no more than $\tau^*$ steps for the property $\varphi$ under the synchronous interpretation.
        \item If $(\mathcal R_\varphi, r,s)$, is proper and $\mathcal R$ is an interval,
        then for every $\varepsilon$-approximately sound monitor $\moni$ and every $\tau < \tau^*$, there exists an event sequence $w=(\bm e_t)_{t\in\mathbb Z}\in \mathcal E^{\mathbb Z}$ on which $\moni$ does not produce a verdict for time $T$ after $\tau$ steps.
    \end{enumerate}
\end{corollary}
To obtain an analogous result for multiplicative expressions, we need first to define the nesting depth of a multiplicative property.

\begin{definition}[Nesting depth]
Let $\varphi$ be a multiplicative property.
The \emph{nesting depth} of $\varphi$ is defined recursively as follows:
\begin{itemize}
    \item If $\varphi$ is an atom, $\mathrm{Depth}(\varphi) = 1$.
    \item If $\varphi=c\cdot \psi$ for $c\in\mathbb R$, $\mathrm{Depth}(\varphi) = \mathrm{Depth}(\psi)$.
    \item If $\varphi = \psi + \chi$, then $\mathrm{Depth}(\varphi) = \max\{ \mathrm{Depth}(\psi), \mathrm{Depth}(\chi) \}$.
    \item If $\varphi = \psi \cdot \chi$, then $\mathrm{Depth}(\varphi) = \max\{ \mathrm{Depth}(\psi), \mathrm{Depth}(\chi) \}+1$.
\end{itemize}
\end{definition}
With multiplicative properties, we loose the equivalence to monitoring atoms we had for linear properties. Therefore, while we can bound uncertainty intervals (obtaining an analogous to Thm.~\ref{thm:approximate-monitors} (1), it is no longer always possible to construct a stream of events that fills the uncertainty interval, which is required in the proof of Thm.~\ref{thm:approximate-monitors} (2).
\begin{corollary}
    \label{cor:approximate-monitors-multi}
    Let $(\mathcal R, r,s)$ be an input setting, $\varepsilon>0$ a tolerance value, $T\in \mathbb N$ a point in time, $\varphi$ a linear formula.
    The following holds for $\tau^*$ defined as:
    \begin{equation}
    \label{eq:approximate-monitors-tau-cor2}
        \tau^*(\mathcal R, r,s,\varepsilon,T, \varphi) \coloneqq \left\lceil \log_s\left[\left(\frac{\sqrt[\mathrm{Depth}(\varphi)]{2\varepsilon}}{\sup\mathcal R_\varphi - \inf \mathcal R_\varphi} - \frac{r^{T+1}}{1-r}\right)(1-s)\right]\right\rceil-1.
    \end{equation}
        There exists an $\varepsilon$-approximately sound monitor $\moni^*$ that produces a verdict for all $t\geq T$ in no more than $\tau^*$ steps for the property $\varphi$ under the synchronous interpretation.
\end{corollary}
\begin{remark}
    Corollary~\ref{cor:approximate-monitors-multi} cannot be extended to non-multiplicative expressions or expressions with the asynchronous interpretation.
    \begin{itemize}
        \item If $\varphi$ is not multiplicative, then $\mathcal R_\varphi$ is not defined, because $\mathcal R_\varphi$ is a spread on an input setting containing $0$, so divisions are not defined.
        \item If $\varphi$ is interpreted asynchronously, then it can contain arbitrary many empty events for each atom, which make it impossible to give a minimum required observation time.
    \end{itemize}
\end{remark}

\section{Register Complexity of Discounted-Sum Monitoring} \label{sec:register-complexity}

The quantity $\tau$ plays two roles in~\Cref{alg:monitor}: it bounds both the \emph{verdict delay}, i.e., the number of steps required to resolve the uncertainty interval for a position, and the \emph{register count}, i.e., the number of running sums maintained in parallel.
These roles have different origins.
The delay is inherent to discounted sums: contributions from unobserved values decay geometrically, and $\tau$ is the first horizon where this is resolved for sure.
The register count, by contrast, is an implementation cost. Since the running sums are related, it is natural to ask whether we can track the $\tau$ pending positions with less memory.

We show that no such reduction is possible in the worst case in the future-only setting.
To make the question precise, we formalize monitors as \emph{affine register machines} (ARMs): finite-state machines equipped with real-valued registers, affine updates, and strict affine guards~\cite{DBLP:journals/tcs/KaminskiF94,DBLP:conf/lics/AlurDDRY13,DBLP:conf/lics/FerrereHS18}.
By~\Cref{thm:approximate-monitors}, every monitored position admits an $\varepsilon$-approximately sound verdict by its horizon, so we encode these deadline verdicts as a safety language over $\mathcal{R}^\omega$.
We build, for every $\tau \geq 1$, a monitoring instance with horizon $\tau$ whose language cannot be recognized by any ARM with fewer than $\tau$ registers.

\subparagraph{Safety Formulation of Approximate Monitoring.} \label{sec:safety-formulation}

We first recast approximate monitoring as a safety-language recognition problem.
Throughout this section we work over the normalized input domain $\mathcal R = [0,1]$.
This is without loss of generality for interval domains: any bounded interval $[m,M]$ can be mapped to $[0,1]$ by an affine transformation, and discounted sums, target intervals, and tolerances rescale accordingly.

A \emph{monitoring instance} is a tuple $\mathcal P = (\mathcal R, r,s,\mathcal I,\varepsilon)$, where $(\mathcal R,r,s)$ is a proper input setting with $\mathcal R=[0,1]$ and $r,s\in[0,1)$, where $\mathcal I=(L,U)$ is a target interval and $\varepsilon>0$.
The \emph{horizon} of $\mathcal P$ is $\tau(\mathcal P):=\tau^*(\mathcal R,r,s,\varepsilon,0)$, as in~\Cref{thm:approximate-monitors}.
We say that $\mathcal P$ is \emph{future only} if $r=0$ and $s>0$, and \emph{past only} if $r>0$ and $s=0$.

\begin{definition}[Monitoring Language]
    The \emph{language} of a monitoring instance
    $\mathcal P=(\mathcal R,r,s,\mathcal I,\varepsilon)$ is $L(\mathcal P) := \{x\in\mathcal R^\omega \mid \forall t \geq 0 : U_t^{r,s}(x_{0:t+\tau(\mathcal P)}) \subseteq \mathcal I_{+\varepsilon}\}$.
\end{definition}

By~\Cref{lem:uncertainty-intervals,thm:approximate-monitors}, the uncertainty interval at time $t+\tau(\mathcal P)$ has diameter at most $2\varepsilon$.
Hence, if it is not contained in $\mathcal I_{+\varepsilon}$, then it is disjoint from $\mathcal I_{-\varepsilon}$.
Therefore, the rule that returns $\top$ exactly when $U_t^{r,s}(x_{0:t+\tau(\mathcal P)})\subseteq\mathcal I_{+\varepsilon}$
and returns $\bot$ otherwise is always defined and $\varepsilon$-approximately sound.

Since non-membership in $L(\mathcal P)$ is witnessed by a finite prefix, $L(\mathcal P)$ is a safety language.
Indeed, $U_t^{r,s}(x_{0:t+\tau(\mathcal P)})$ depends only on $x_{0:t+\tau(\mathcal P)}$, so any violation at time $t$ is shared by every extension of this prefix.
By~\Cref{lem:uncertainty-intervals} and since $\mathcal R=[0,1]$, the condition $U_t^{r,s}(x_{0:t+\tau(\mathcal P)})\subseteq\mathcal I_{+\varepsilon}$
is equivalent to $L - \varepsilon < S_t^{r,s}(x_{0:t+\tau(\mathcal P)}) < U + \varepsilon - \gamma_{t,t+\tau(\mathcal P)}^{r,s}$.

\subparagraph{Affine Register Machines.} \label{sec:register-machines}

Let $Y=\{y_1,\ldots,y_k\}$ be a finite set of registers.
A \emph{valuation} is a map $\nu:Y\to\mathbb R$.
An \emph{update} over $Y$ is a parallel assignment $y_i \gets \sum_{j=1}^k a_{i,j}y_j+b_i\varsigma+c_i$ for all $1\leq i\leq k$, where $\varsigma\in\mathcal R$ is the current input and all coefficients are real.
We write $\Gamma(Y)$ for the set of such updates.
A \emph{guard} over $Y$ is a finite conjunction of strict affine inequalities $\sum_{j=1}^k a_jy_j+b\varsigma \bowtie c$, where ${\bowtie}\in\{<,>\}$.
We write $(\nu,\varsigma)\models\varphi$ if the guard $\varphi$ holds under valuation $\nu$ and input $\varsigma$, and we write $\Phi(Y)$ for the set of guards over $Y$.

\begin{definition}[Affine Register Machine]
	A (deterministic) \emph{affine register machine} (ARM) is a tuple
	$\mathcal{M} = (\mathcal{R}, Y, Q, q_0, \nu_0, \Delta)$
	where $\mathcal{R} \subset \mathbb{R}$ is a bounded input set, $Y$ is a finite set of registers, $Q$ is a finite set of control locations, $q_0 \in Q$ is the initial location, $\nu_0 \colon Y \to \mathbb{R}$ is the initial valuation, and
	$\Delta \subseteq Q \times \Phi(Y) \times \Gamma(Y) \times Q$
	is a finite transition relation.
	We write transitions as $(q,\varphi,\gamma,q') \in \Delta$ and require the following determinism condition:
	for every location $q \in Q$, every valuation $\nu$, and every input $\varsigma \in \mathcal{R}$, there is at most one transition $(q,\varphi,\gamma,q') \in \Delta$ such that $(\nu,\varsigma) \models \varphi$.
\end{definition}
A configuration is a pair $(q,\nu)$.
On input $\varsigma$, the machine follows the unique enabled transition $(q,\varphi,\gamma,q')$ to $(q',\gamma(\nu,\varsigma))$.
If no transition is enabled, the run terminates.
A run on $x\in\mathcal R^\omega$ is \emph{accepting} if infinite.
The language $L(\mathcal M)$ is the set of streams with an accepting run; hence every ARM recognizes a safety language.

\subparagraph{Sufficiency: An ARM Construction.} \label{sec:upper-bound}

We show that $L(\mathcal{P})$ can be recognized by an ARM with $\tau(\mathcal{P}) + 1$ registers.
The construction follows the monitor from~\Cref{alg:monitor}.
The ARM keeps one register for each pending position and updates all pending truncated sums in parallel.
If a register $y_a$ stores the truncated sum for the pending position of age $a$, then on input $x$ the affine updates are $y_0\gets r y_0+x$ and $y_a\gets y_{a-1}+s^a x$ for all $1\leq a\leq \tau(\mathcal{P})-1$.
When a position reaches age $\tau(\mathcal{P})$, its completed truncated sum is $y_{\tau(\mathcal{P})-1}+s^{\tau(\mathcal{P})}x$, which the machine checks by a guard before applying these updates.
One additional register tracks the geometric factor needed to express the time-dependent past residual $r^{t+1}/(1-r)$.

\begin{theorem}[General Upper Bound] \label{thm:upper-bound}
	For every monitoring instance $\mathcal{P}$, there is an ARM $\mathcal{M}$ with $\tau(\mathcal{P}) + 1$ registers such that $L(\mathcal{M}) = L(\mathcal{P})$.
\end{theorem}
In the future-only case the past residual is zero.
The auxiliary register is therefore unnecessary, and the same construction uses exactly one register per pending future window.
\begin{corollary}[Future-Only Upper Bound] \label{cor:future-upper-bound}
	For every future-only monitoring instance $\mathcal{P}$, there is an ARM $\mathcal{M}$ with $\tau(\mathcal{P})$ registers such that $L(\mathcal{M}) = L(\mathcal{P})$.
\end{corollary}

\subparagraph{Necessity: A Register Lower Bound.} \label{sec:lower-bound}

We now show that the future-only upper bound is optimal in the worst case.
For each $k\geq 1$, consider the future-only monitoring instance
$\mathcal P_k := ([0,1],\,0,\,\tfrac12,\,(0,1),\,2^{-(k+1)})$.
For $x\in[0,1]^\omega$ and $t\geq 0$, write $V_t^k(x) := \sum_{i=0}^{k}2^{-i}x_{t+i}$ and $T_k := 1-2^{-(k+1)}$.
For this family, the safety condition has a particularly simple form.
The horizon is $k$, the residual future uncertainty is $2^{-k}$, and the lower bound is vacuous because all inputs are nonnegative.
Thus, the monitoring language is exactly the set of streams whose discounted windows of length $k+1$ remain below the threshold $T_k$.

\begin{lemma}[Characterization of the hard instances]
\label{lem:Pk-characterization}
    For every $k\geq 1$, the future-only monitoring instance $\mathcal P_k$ has horizon $\tau(\mathcal P_k)=k$, and its monitoring language is $L(\mathcal P_k) = \{x\in[0,1]^\omega \mid \forall t\geq 0: \sum_{i=0}^{k}2^{-i}x_{t+i} < 1-2^{-(k+1)} \}$.
\end{lemma}

The lower-bound proof is a finite-dimensional linear-algebra argument.
Suppose an ARM with only $k-1$ registers recognizes the language above.
Consider the first $k$ transitions taken on a suitable accepted anchor word.
On all length-$k$ prefixes inducing this same transition sequence, the register valuation after $k$ steps is an affine map from a $k$-dimensional prefix space to $\mathbb R^{k-1}$.
Therefore, this map has a nonzero kernel direction: a perturbation that is invisible to the registers after the first $k$ inputs.

The separating prefixes are chosen from an explicit one-parameter family.
For $\theta$ just below $T_k$, the word $u(\theta) \cdot \alpha \cdot 0^\omega$ is constructed so that its first $k$ discounted windows all have value $\theta$.
Moving slightly in the invisible kernel direction preserves the first $k$ transitions, the register valuation after $k$ steps, and the entire suffix.
However, by choosing the last nonzero coordinate of the kernel direction, one discounted window is pushed above $T_k$.
Determinism then forces the ARM to treat two identical configurations with the same suffix identically, although one stream satisfies the monitoring condition and the other violates it.
Note that $\mathcal{R}$ being an interval is crucial here, as it ensures that the kernel direction (an arbitrary vector in $\mathbb{R}^k$) can be realized as a small perturbation inside $\mathcal{R}^k$.

\begin{theorem}[Future-Only Lower Bound] \label{thm:future-lower-bound}
	For every $k \geq 1$, there is a future-only monitoring instance $\mathcal{P}_k$ such that for every ARM $\mathcal{M}$ with $\tau(\mathcal{P}_k) - 1$ registers we have $L(\mathcal{M}) \neq L(\mathcal{P}_k)$.
\end{theorem}

\subparagraph{Special Cases.} \label{sec:special-cases}

For the past-only case ($s=0$ and $r>0$), it suffices to use just two registers: one for the observed past, which follows the rule $p_t = x_t+rp_{t-1}$, and one for the uncertainty term $r^{t+1}/(1-r)$. 
The register complexity for a general $s>0$ and $r>0$ remains open. 

Over finite alphabets, the future-only lower bound disappears.
For $r=0$, monitoring reduces to checking a sliding window $x_t,\ldots,x_{t+\tau}$, so a finite-state monitor suffices, at the cost of exponential state space in $\tau$.
With past discounting, however, this is no longer possible: although a finite alphabet collapses the $\tau$ parallel future registers into finite control, the recurrence $p_t = x_t + r p_{t-1}$ still can take infinitely many values, which may require an infinite-state monitor.

\section{Experiments}
We evaluate our monitors on a collection of real-world and synthetic scenarios 
with the objective of demonstrating the practical effectiveness of the proposed monitors and to examine how tight the theoretical bounds on resource usage and accuracy are in practice. In the main paper our evaluation is guided by the research questions: \textbf{(RQ1)} How closely does the observed register usage of Algorithm~\ref{alg:monitor} match the theoretical bound of Thm.~\ref{thm:approximate-monitors}?, and 
\textbf{(RQ2)} How does register usage vary under synchronous and asynchronous semantics? 
We address research questions beyond (RQ1) and (RQ2) in the Appendix~\ref{sec:more_experiments}.

We evaluate our monitors on traces from two benchmarks: \texttt{PowerData}, a dataset of power usage traces from Google data centers~\cite{google-traces}; and \texttt{Adult}, demographic parity traces produced by a fairness-aware neural network trained on the Adult dataset~\cite{adult_2}.

\subparagraph{Register Usage on \texttt{PowerData}.}

To address RQ1, we study register usage on traces from \texttt{PowerData}.
We report \emph{normalized register usage}, defined as the number of active registers divided by the theoretical upper bound $\tau^*$ from Thm.~\ref{thm:approximate-monitors}.
The input range is $\mathcal R = [0,1]$ and we fix the discount factors to $r=s=0.9$.
In Fig.~\ref{fig:RQ1-google} (a), we vary the tolerance $\varepsilon$ logarithmically between $0.5$ and $5\cdot 10^{-4}$, while fixing the target interval to $\mathcal I = [0.703,0.728]$, centered at the empirical mean with width equal to one standard deviation of the dataset.
In Fig.~\ref{fig:RQ1-google} (b), we fix $\varepsilon=5\cdot 10^{-4}$ and vary the width of $\mathcal I$ between $10$ and $0.01$ times the standard deviation, again centered at the mean.
In all experiments, the initial monitoring time $T$ is chosen such that $r^{T+1}/(1-r)$ is closest to $\varepsilon/(\sup \mathcal R - \inf \mathcal R)$.
In both Fig.~\ref{fig:RQ1-google} (a) and (b), the left graph represents the cumulative distribution of normalized active monitors, with one line representing each parameter setting  -- $\varepsilon$ in (a), $|\mathcal I|$ in (b) --, while the right graph represents each distribution as a boxplot. Each boxplot corresponds to a parameter setting, color matched with the corresponding left graph.

Across all configurations, most verdicts are produced with fewer than $60\%$ of the theoretically available registers active, and register usage never exceeds $80\%$.
This result indicates that we could actually implement a monitor with, for example, $70\%$ of the theoretically required registers, and still output a sound verdict most of the time.

\subparagraph{Monitoring Demographic Parity.}

To investigate RQ2, we monitor demographic parity on the \texttt{Adult} dataset.
The input set is $\mathcal R=\{0,1\}$ and we again fix $r=s=0.95$.
We consider decreasing tolerance values with a fixed target interval $\mathcal I = [-0.1,0.1]$.
We report the results in Fig.~\ref{fig:RQ2-adult}, in the same format as the previous figure, for the demographic parity property (Ex.~\ref{ex:dem-parity}) with the synchronous (blue) and asynchronous (green) interpretation.
In contrast to the previous experiments, we report absolute register counts rather than normalized usage.
This is because Thm.~\ref{thm:approximate-monitors} does not provide a uniform bound in this setting,  as the denominators can be arbitrarily close to zero.
As expected, synchronous interpretations are significantly more efficient than asynchronous ones, since discounting is applied uniformly across streams rather than independently.

\begin{figure}[t]
    \centering

    \begin{subfigure}{0.48\textwidth}
        \centering
        \includegraphics[width=0.49\textwidth]{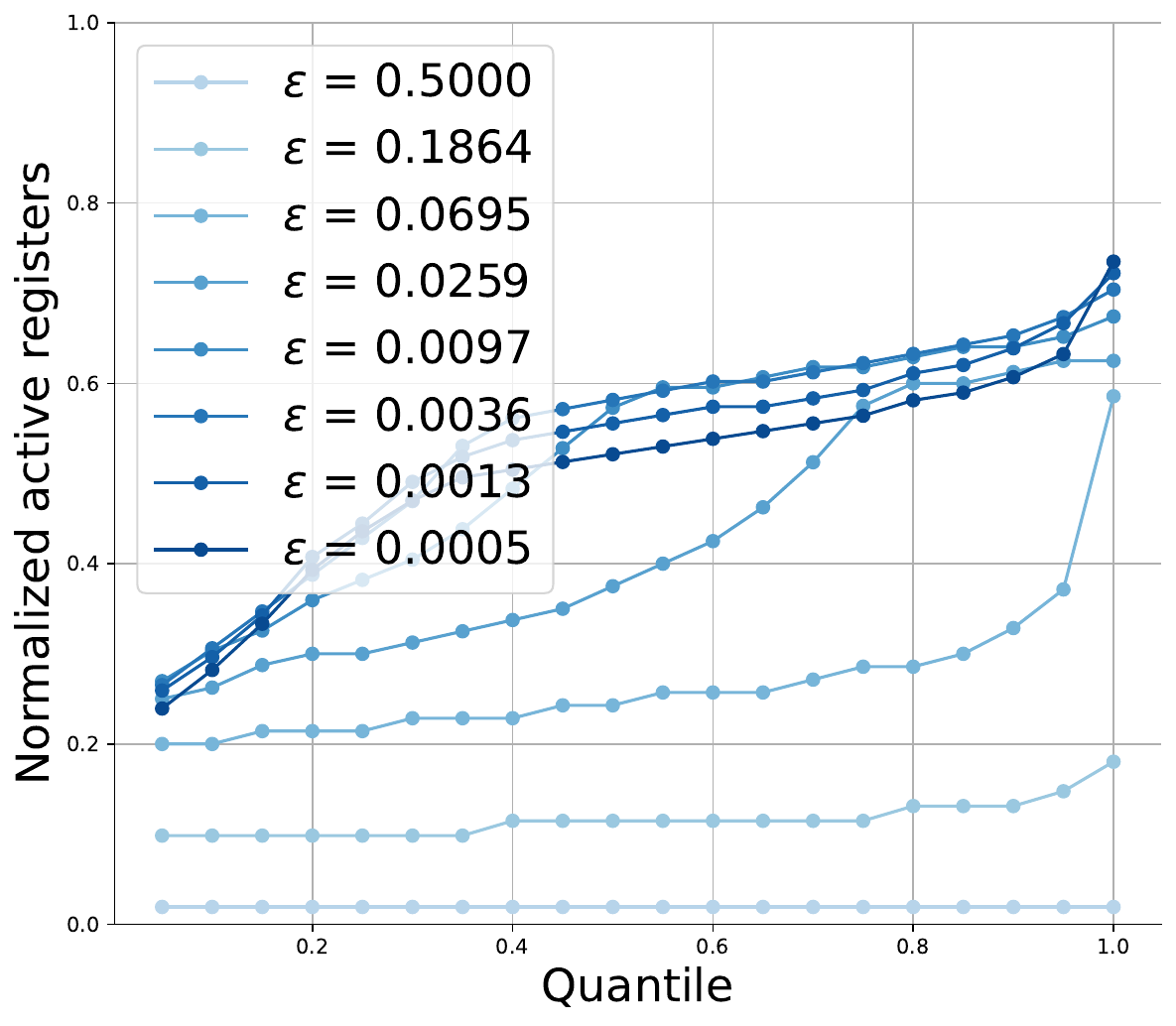}
        \includegraphics[width=0.49\textwidth]{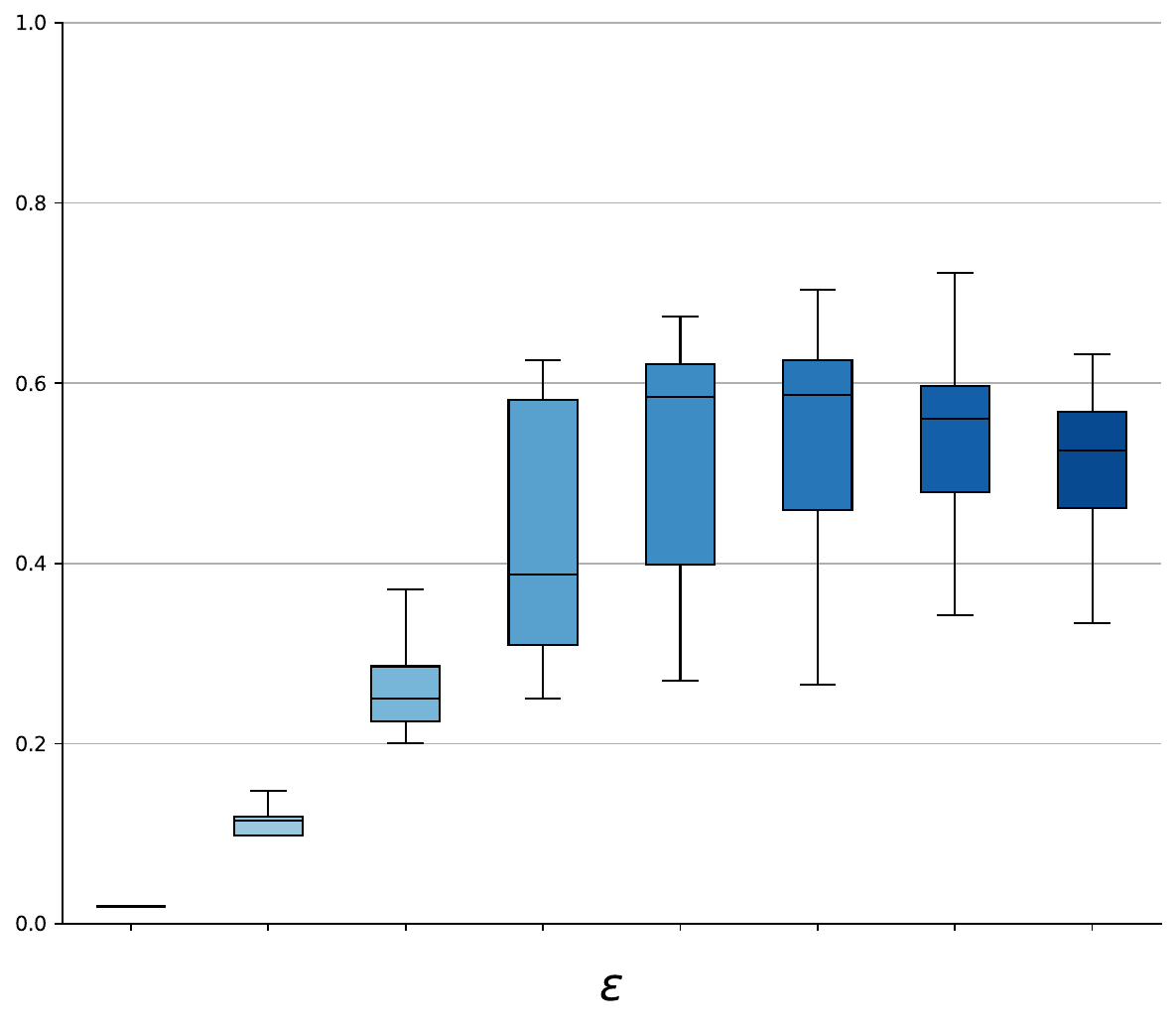}
        \caption{Variable monitoring tolerance factor $\varepsilon$.}
    \end{subfigure}
    \hfill
    % Group (b)
    \begin{subfigure}{0.48\textwidth}
        \centering
        \includegraphics[width=0.49\textwidth]{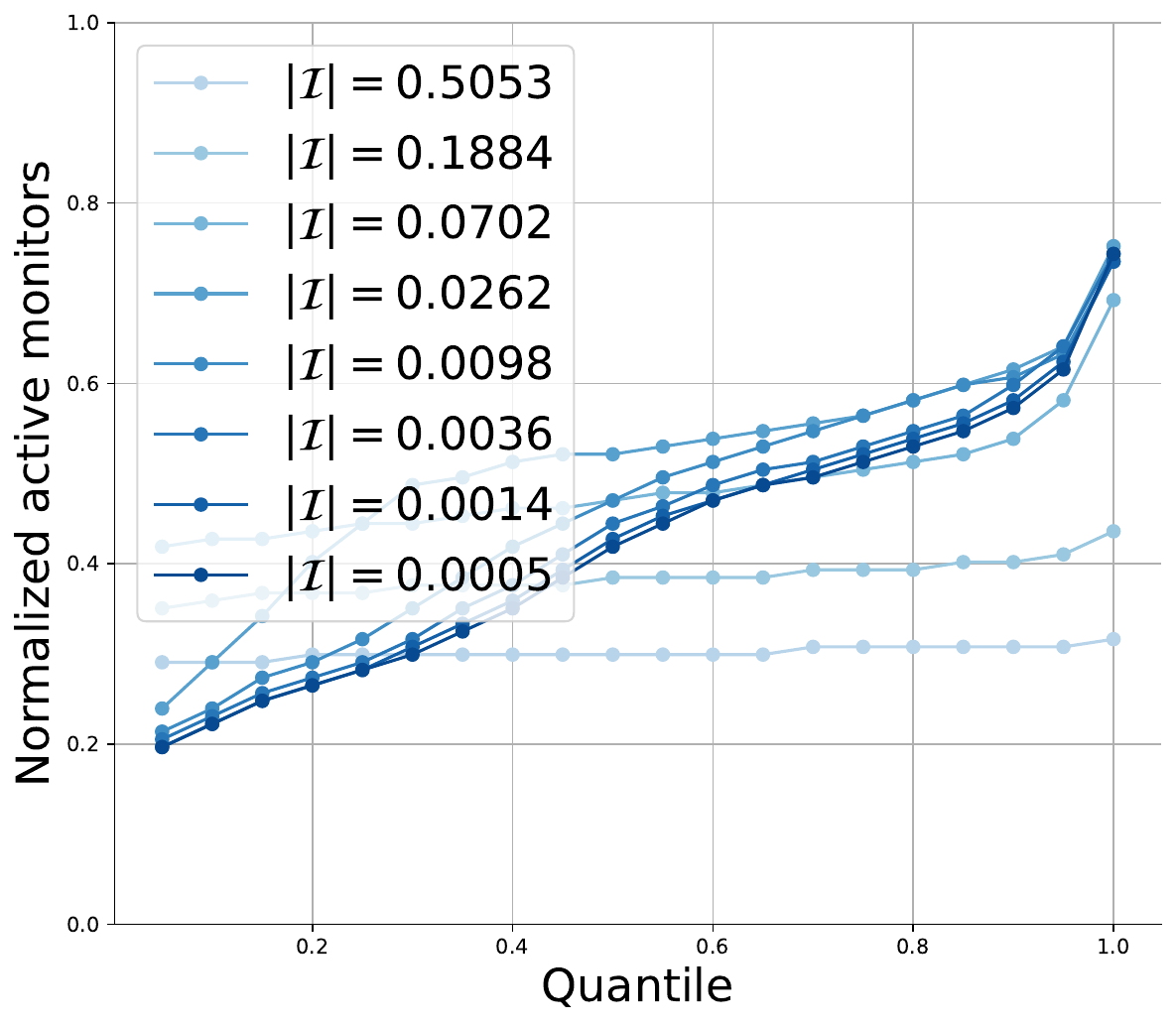}
        \includegraphics[width=0.49\textwidth]{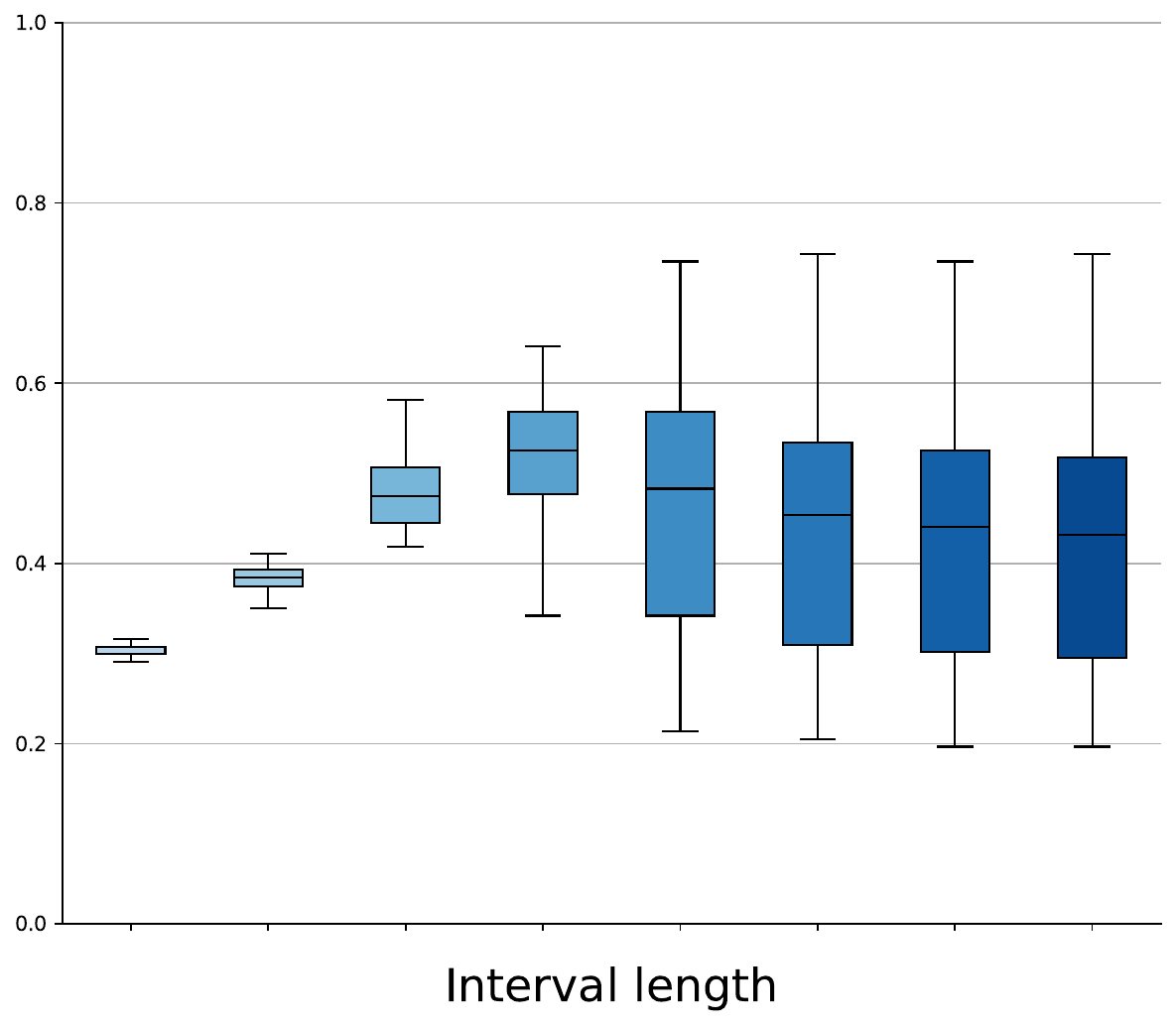}
        \caption{Variable width of the  target interval $|\mathcal I|$.}
    \end{subfigure}

    \caption{Register usage in terms of normalized active registers, \texttt{PowerData} dataset.}
    \label{fig:RQ1-google}
\end{figure}

\begin{figure}[t]
     \centering
     \begin{subfigure}{0.4\textwidth}
         \centering
         \includegraphics[width=\textwidth]{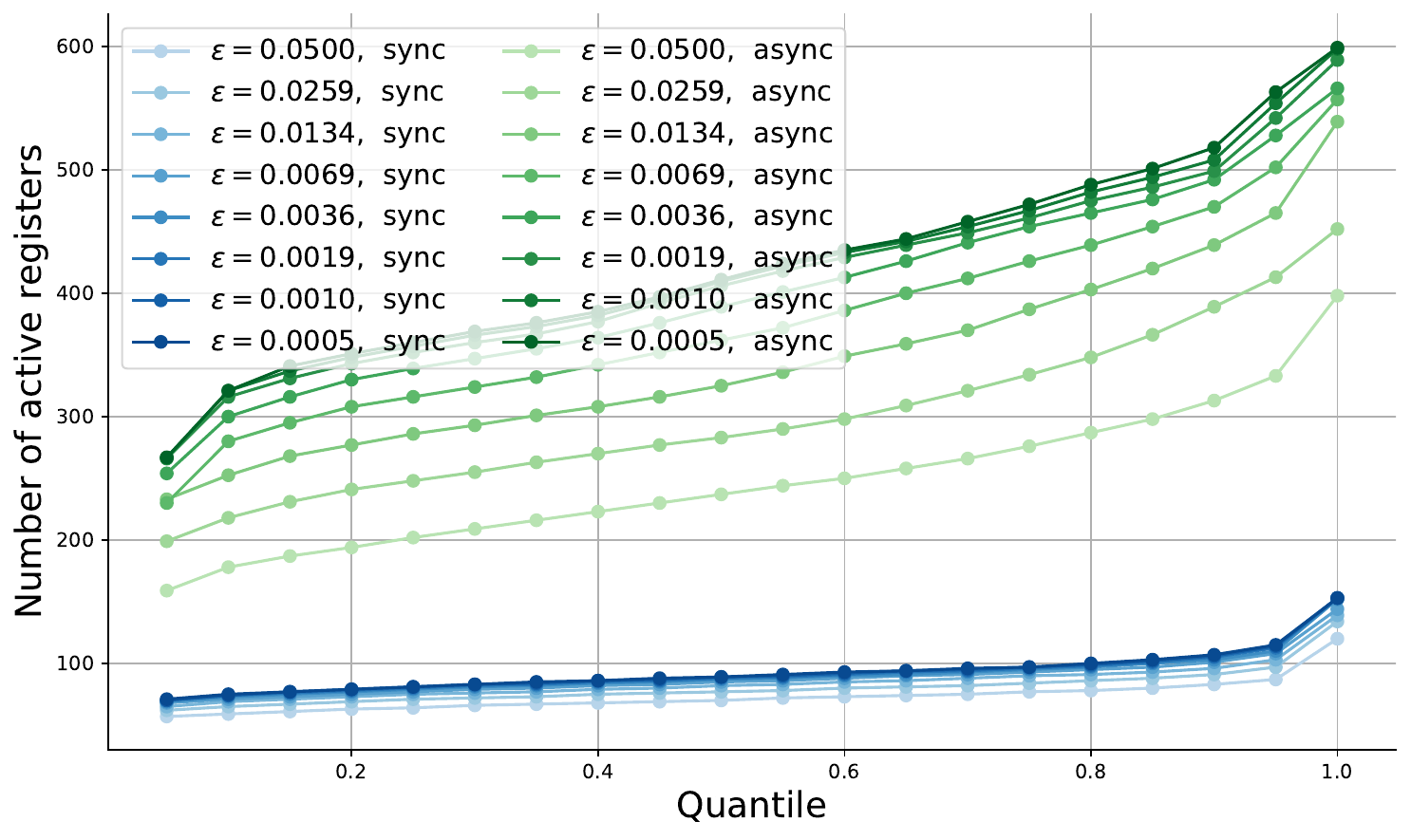}
         % \caption{}
     \end{subfigure}\hspace{2em}
     \begin{subfigure}{0.4\textwidth}
         \centering
         \includegraphics[width=\textwidth]{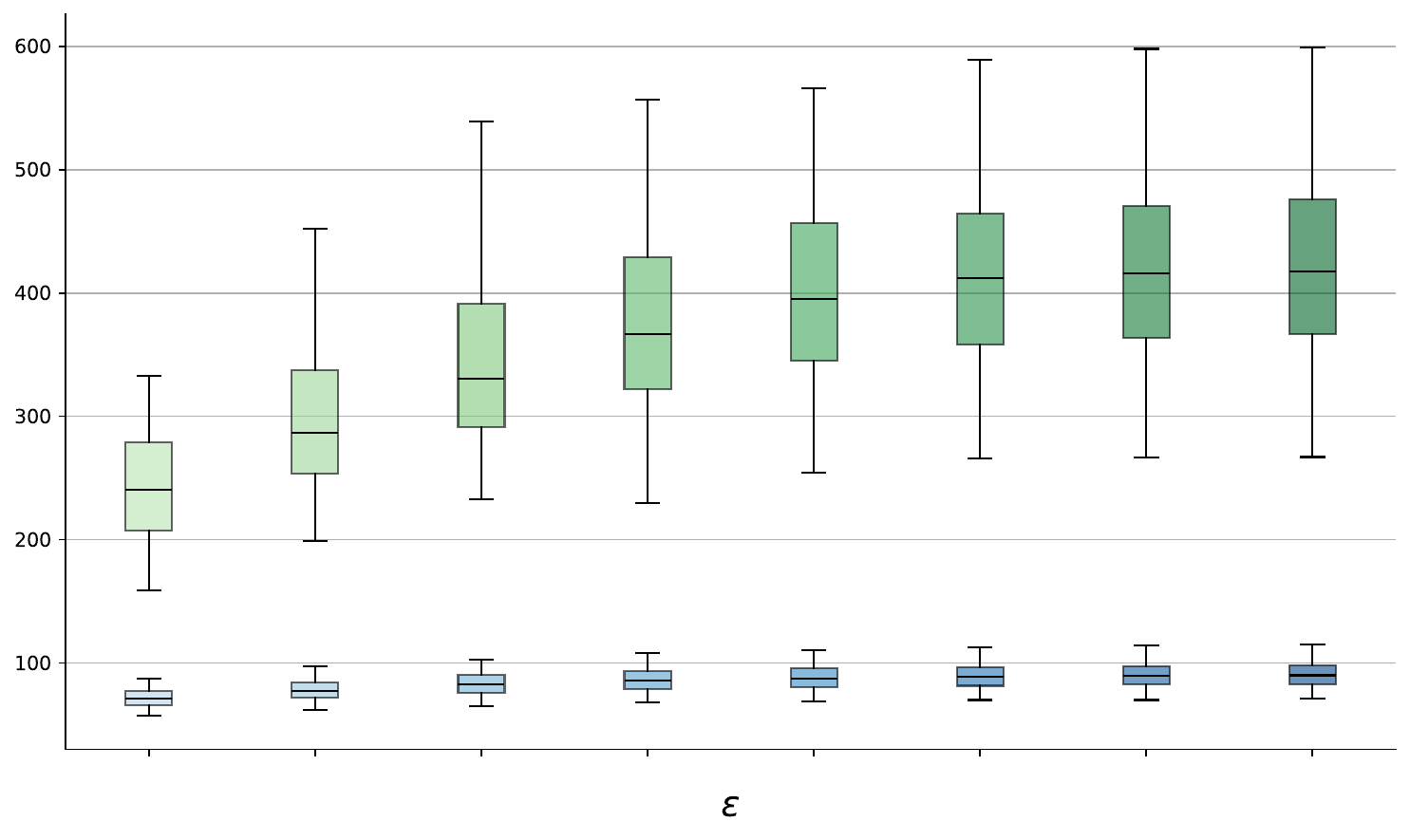}
         % \caption{}
     \end{subfigure}
        \caption{Active registers when   demographic parity, variable $\varepsilon$. \texttt{Adult} dataset.}
        \label{fig:RQ2-adult}
\end{figure}

\section{Related Work}

Discounted-sum objectives are well studied in verification, automata, games, and logic as quantitative specifications where recent events outweigh distant ones~\cite{chatterjee2010quantitative,de2003discounting,boker2012approximate,shapley1953stochastic,filar2012competitive,almagor2014discounting}.
The nontrivial structure of achievable values underlies the difficulty of the target discounted-sum problem~\cite{boker2015target}, motivating approximation:
nondeterministic discounted-sum automata cannot be determinized exactly, but approximate determinization is always possible~\cite{boker2012approximate} and approximate inclusion is decidable~\cite{bansal2022comparator}.
We study online monitoring under bounded memory and show a parallel dichotomy: exact sound monitoring requires unbounded memory, whereas approximate monitoring admits tight bounds.

Runtime verification provides lightweight formal guarantees by checking specifications against executions online~\cite{bartocci2018introduction}.
Quantitative monitoring replaces boolean verdicts with real-valued measures; STL robustness semantics, for instance, quantify how strongly a signal satisfies a specification~\cite{maler2004monitoring,fainekos2009robustness,donze2010robust,donze2013efficient,deshmukh2017robust}.
Recent work~\cite{henzinger2021quantitative} establishes precision-cost tradeoffs for quantitative monitoring.
We focus on discounted-sum monitoring and extend the framework to stochastic processes via confidence sequences.

Discounted sums can be viewed as convolution with an exponential kernel: the exponentially weighted moving average (EWMA) computes a normalized discounted sum~\cite{roberts2000control}, equivalent to a first-order IIR filter~\cite{schafer1989discrete}.
EWMA control charts detect mean shifts under i.i.d. inputs~\cite{lucas1990exponentially,montgomery2020introduction}, and exponential smoothing in forecasting balances responsiveness and noise reduction~\cite{gardner2006exponential,hyndman2008automatic}.
We complement these by studying memory requirements and approximation guarantees for arbitrary input streams.

Standard fairness metrics such as demographic parity and equalized odds~\cite{hardt2016equality,dwork2012fairness} are static population statistics.
Existing fairness monitors use only cumulative statistics~\cite{henzinger2023monitoring} or stream-based specifications~\cite{baumeister2025stream}.
Recent work highlights temporal aspects of fairness~\cite{gohar2024long}:
discounting appears in fair reinforcement learning~\cite{jabbari2017fairness}, decision making~\cite{torres2024temporal}, and resource allocation~\cite{kumar2025remember}, but targets policy design rather than monitoring with formal guarantees.

\bibliography{references}

\appendix

\section{Discussion: Alternative Local Measures}

In this paper we study discounted sums as a method to obtain a local measure for global property. 
Discounted sums are a principled method of aggregation, where instances in the future and in the past contribute to the local measure according to their value and the distance from one event to the other.

In this section, we discuss alternative localization measures, and how do they compare to discounted sums. We mainly discuss window averages and kernel-weighted averages.

\subsection{Window average}
A natural localization measure, alternative to the discounted average, is the rolling window average of length $2l+1$, defined as
\[
\mathrm{WA}_t^{l}(w) = \frac{1}{2l+1}\sum_{i=-l}^{l} x_{t+i}.
\]
The parameter of window length serves a similar purpose as the pair of discount factors $(r,s)$.
While there is no canonical correspondance, a natural one is to consider symmetric discounting $(r = s)$ and that the weight of the central element in both averages is the same. This yields
\[
\frac{1}{2l+1} = \frac{1}{1+\frac{2r}{1-r}} \iff 
l = \frac{r}{1-r} \iff r = \frac{l}{1+l}.
\]
Under these conditions, we have a similar bounded difference as in Lemma~\ref{lem:convergingsums}, that is $|\mathrm{WA}_{t+1}^{l}(w) -  \mathrm{WA}_t^{l}(w)| = d_{\mathcal R}$.
We can also compute uncertainty sets in a similar fashion as Lemma~\ref{lem:uncertainty-intervals}, although the uncertainty sets may not be intervals (i.e., they may have gaps) if $\mathcal R$ has gaps of any size, due to the window averages expressing a finite sum.
Sound monitors that give a verdict in at most $l$ observations trivially exist, and they can be trivially made to use only $2l+1$ registers.

In essence, window averages are computationally easier to manage, but have two main undesired effects:
\begin{itemize}
    \item There is a sharp jump between the importance of a single observation for one step and its consecutive. In contrast, in discounted averages, the importance associated with one step falls gradually as it becomes further away from the point of interest. 
    For many applications, like monitoring fairness, this is a more principled approach than setting a hard cutoff distance in time.
    \item For properties expressed as quotients, even if $0\in\mathcal R$, as long as it does not contain both positive and negative elements at the same time, the discounted sum of the denominator will be well defined from the first instance of a non-zero value. 
    In contrast, when localizing with a rolling window average, as soon as we have a sequence of $2l+1$ instances of zero, the value of the property is not well defined. 
    This is particularly important for synchronous properties, as non-appearances are in effect substituted by zero.
\end{itemize}

\subsection{Kernel-Weighted Average}
Both rolling window and discounted averages are a particular case of kernel weighted averages.
A kernel is a weight function $\kappa\colon \mathbb Z \to \mathbb R_{\geq 0}$ that has a finite sum, i.e., 
$\sum_{i\in\mathbb Z}\kappa(i) < \infty$.
To have the effect of a local measure, we ask for the kernel to be non-decreasing in $\mathbb Z_{\leq 0}$ and non-increasing for $\mathbb Z_{\geq 0}$.
Kernels generalize both discounted  and rolling window averages:
\begin{itemize}
    \item The window average corresponds to the kernel $\kappa(i) = \frac{1}{2l+1}\1[-l \leq i \leq l]$.
    \item The discounted sum corresponds to the kernel
    $\kappa(i) = r^{\max\{0,-i\}}\cdot s^{\max\{0,i\}}$.
\end{itemize}
While kernel-weighted averages are conceptually appealing, it is not clear how we could efficiently monitor a property defined by an arbitrary kernel. 
For discounted sums, we take advantage of the fact that the value can be computer recursively, so we can efficiently monitor into the future while forgetting past values (see the update, for example, in line 6 of Alg.~\ref{alg:monitor}).
For many popular kernels, like Gaussian kernels, this is not possible.

\section{Monitoring Discounted Sums}

\noindent\textbf{Lemma~\ref{lem:convergingsums}.}
\emph{
    Let $(\aR, r, s)$ be an input setting and $w \in \aR^{\mathbb{Z}}$ be a sequence of values, and $t\in\ZN$ a point in time.
    Then $\dsum_t$ exists and 
    $\big| \dsum^{r,s}_t(w) - \dsum^{r,s}_{t-1}(w)\big|\leq d_{\mathcal R}$.
}
\begin{proof}
    Since $\dsum_t$ is defined as an infinite sum, we first check convergence. Let $a=\inf\aR$ and $b=\sup\aR$.
    Then $|x_{t\pm i}|\le \max\{|a|,|b|\}$ for all $i\ge 0$, and hence the past and future parts are dominated by geometric series with ratios $\dfp$ and $\dff$, respectively. Since $\dfp, \dff \in [0,1)$, both series converge absolutely, so $\dsum_t$ exists.

    Next we show the Lipschitz property. Write
    \begin{align*}
        \dsum_t
        &= \sum_{i=1}^\infty \dfp^i x_{t-i} \;+\; x_t \;+\; \sum_{i=1}^\infty \dff^i x_{t+i},\\
        \dsum_{t-1}
        &= \sum_{i=1}^\infty \dfp^i x_{t-1-i} \;+\; x_{t-1} \;+\; \sum_{i=1}^\infty \dff^i x_{t-1+i}.
    \end{align*}
    By re-indexing,
    \begin{align*}
         \sum_{i=1}^\infty \dfp^i x_{t-1-i} + x_{t-1}
        = \sum_{j=0}^\infty \dfp^{j} x_{t-1-j},
        \qquad
        x_t + \sum_{i=1}^\infty \dff^i x_{t+i}
        = \sum_{j=0}^\infty \dff^{j} x_{t+j},
    \end{align*}
    and similarly
    \begin{align}
          \sum_{i=1}^\infty \dfp^i x_{t-i}
        = \dfp \sum_{j=0}^\infty \dfp^{j} x_{t-1-j},
        \qquad
        \sum_{i=1}^\infty \dff^i x_{t-1+i}
        = \dff \sum_{j=0}^\infty \dff^{j} x_{t+j}.
    \end{align}
    Therefore,
    \begin{align*}
        \dsum_{t-1}-\dsum_t
        &= \Big(\sum_{j=0}^\infty \dfp^{j} x_{t-1-j} - \dfp \sum_{j=0}^\infty \dfp^{j} x_{t-1-j}\Big)
        \;+\;
        \Big(\dff \sum_{j=0}^\infty \dff^{j} x_{t+j} - \sum_{j=0}^\infty \dff^{j} x_{t+j}\Big)\\
        &= (1-\dfp)\sum_{j=0}^\infty \dfp^{j} x_{t-1-j} \;-\; (1-\dff)\sum_{j=0}^\infty \dff^{j} x_{t+j}.
    \end{align*}
    Since $a\leq x_s\leq b$ for all $s$, we have
    \begin{align*}
           \frac{a}{1-\dfp}\leq \sum_{j=0}^\infty \dfp^{j} x_{t-1-j}\leq \frac{b}{1-\dfp},
        \qquad
        \frac{a}{1-\dff}\leq \sum_{j=0}^\infty \dff^{j} x_{t+j}\leq \frac{b}{1-\dff}.
    \end{align*}
    Multiplying by $(1-\dfp)$ and $(1-\dff)$, respectively, yields
    \begin{align*}
         a \leq (1-\dfp)\sum_{j=0}^\infty \dfp^{j} x_{t-1-j}\leq b,
        \qquad
        a \leq (1-\dff)\sum_{j=0}^\infty \dff^{j} x_{t+j}\leq b.
    \end{align*}
    Hence
    \begin{align*}
        a-b \leq \dsum_{t-1}-\dsum_t \leq b-a, \quad \text{which implies} \quad  \big|\dsum_{t-1}-\dsum_t\big|\leq b-a = d_{\aR}.
    \end{align*}
   
\end{proof}

\noindent\textbf{Lemma~\ref{lem:possible-sums}.}
\emph{
    Let $(\aR, r,s)$ be an input setting.
    Then the following hold.
    \begin{enumerate}
    \item[(1)] $\mathrm{Sums}_{\mathcal R}^{r,s}\subseteq [\lambda m,\lambda M]$.
    \item[(2)] If $(\aR, r,s)$ is proper, then 
    $
  Closure(\mathrm{Sums}_{\mathcal R}^{r,s})\;=\;[\lambda m,\lambda M]
    $.
    \item[(3)] If $(\aR, r,s)$ is proper and compact, then $\mathrm{Sums}_{\mathcal R}^{r,s}\;=\;[\lambda m,\lambda M]$.
    Where $m=\inf \mathcal R$ and $M=\sup\mathcal R$.
    \end{enumerate}
}
\begin{proof}

\noindent\textbf{\emph{(1)}}
    Since $m\le x_i\le M$ for all $i\in\mathbb Z$ and all coefficients $r^i,s^i$ are nonnegative, we obtain
    \[
    S_{r,s}(x)
    = x_0 + \sum_{i=1}^\infty x_{-i}r^i + \sum_{i=1}^\infty x_is^i
    \le M\left(1+\sum_{i=1}^\infty r^i+\sum_{i=1}^\infty s^i\right).
    \]
    Similarly,
    \[
    S_{r,s}(x)\ge m\left(1+\sum_{i=1}^\infty r^i+\sum_{i=1}^\infty s^i\right).
    \]
    Using the geometric series identities $\sum_{i=1}^\infty r^i=\frac{r}{1-r}$ and $\sum_{i=1}^\infty s^i=\frac{s}{1-s}$ yields
    \[
    \lambda m \le S_{r,s}(x)\le \lambda M.
    \]
    Since $x$ was arbitrary, this proves $\mathrm{Sums}^{\mathcal R}_{r,s}\subseteq [\lambda m,\lambda M]$.

\noindent\textbf{\emph{(3)}}
First, define the one-sided discounted sets
    \[
    A_r(\mathcal R)\;:=\;\Big\{\sum_{i=1}^\infty x_i r^i:\ x_i\in\mathcal R\Big\},
    \qquad
    A_s(\mathcal R)\;:=\;\Big\{\sum_{i=1}^\infty x_i s^i:\ x_i\in\mathcal R\Big\}.
    \]
    Then
    \[
    \mathrm{Sums}^{\mathcal R}_{r,s} \;=\; \mathcal R + A_r(\mathcal R) + A_s(\mathcal R)
    \]
    as a Minkowski sum.
    We split the proof into three steps.
    
\medskip\noindent
\textbf{Step 1: each one-sided set is an interval.}
Fix $r\in[0,1)$ and consider the space $\mathcal K$ of nonempty compact subsets of $\mathbb R$ equipped with the Hausdorff metric $d_H$, defined as usual as
\[
d_H(X,Y) = \max\left\{
\sup_{x\in X}\inf_{y\in Y}|x-y|, \sup_{y\in Y}\inf_{x\in X}|x-y| 
\right\}.
\]
Define the operator
\[
F_r:\mathcal K\to 2^{\mathbb R},\qquad
F_r(K)\;:=\;\bigcup_{a\in\mathcal R} (ra+rK).
\]
\begin{enumerate}
    \item \emph{$F_r$ maps $\mathcal K$ to $\mathcal K$.}
    $F_r(K) = r\cdot (\mathcal R + K)$ (as a Minkowski sum). 
    Since the sum is a continuous function in the product topology, the image of a compact set is a compact set.

    \item \emph{$F_r$ is a contraction.}
    For $K, L \in \mathcal{K}$, let $x = r(a + k) \in F_r(K)$, where $a \in \mathcal{R}$ and $k \in K$.
    Choose $\ell \in L$ such that $|k - \ell| \leq d_H(K, L)$.
    Then $r(a + \ell) \in F_r(L)$ and $|r(a + k) - r(a + \ell)| \leq r d_H(K, L)$.
    The symmetric argument, exchanging $K$ and $L$, gives $d_H(F_r(K), F_r(L)) \leq r d_H(K, L)$.
    Thus, $F_r$ is a contraction with constant $r$.

    \item \emph{$A_r(\mathcal R)$ is the unique fixed point of $F_r$.}
    Since $F_r$ is a contraction, it has a unique fixed point (Banach's fixed-point theorem).
    We first need to show that $A_r(\mathcal R)$ is indeed in the domain of $F_r$.
    First note the set equation
    \begin{equation}\label{eq:Ar-fixed}
    A_r(\mathcal R)\;=\;\bigcup_{a\in\mathcal R}\bigl(ra+rA_r(\mathcal R)\bigr)
    \;=\;F_r\bigl(A_r(\mathcal R)\bigr),
    \end{equation}
    which follows by separating the first digit of the series
    $\sum_{i=1}^\infty a_i r^i = ra_1 + r\sum_{i=2}^\infty a_i r^{i-1}$.
    Because $\mathcal R$ is compact, $\mathcal R^{\mathbb N}$ is compact (Thychonoff's theorem). The map $\eta\colon \mathcal{R}^{\mathbb N}\to \mathbb R$, defined as
    $(x_i)_{i\in\mathbb N}\mapsto \sum_{i=1}^\infty x_i r^i$ is continuous, so $A_r(\mathcal R)$ is the image of a compact set through a continuous map, therefore  compact and hence an element of $\mathcal K$.
    By the Banach fixed point theorem, $F_r$ has a unique fixed point in $\mathcal K$; by \eqref{eq:Ar-fixed}, that fixed point is $A_r(\mathcal R)$.
    \item \emph{A particular interval is also a fixed point under the gap condition.}
    Define
    \[
    I_r \;:=\;\Big[\frac{rm}{1-r},\frac{rM}{1-r}\Big].
    \]
    For each $a\in\mathcal R$,
    \[
    ra+rI_r \;=\;
    \Big[ra+\frac{r^2m}{1-r},\;ra+\frac{r^2M}{1-r}\Big],
    \]
    so
    \[
    F_r(I_r)=\bigcup_{a\in\mathcal R}
    \Big[ra+\frac{r^2m}{1-r},\;ra+\frac{r^2M}{1-r}\Big].
    \]
    The leftmost endpoint occurs at $a=m$ and equals $\frac{rm}{1-r}$, and the rightmost endpoint occurs at $a=M$ and equals $\frac{rM}{1-r}$.
    Moreover, if $a<b$ are such that there is no point of $\mathcal R$ in $(a,b)$, then the corresponding two intervals overlap provided
    \[
    ra+\frac{r^2M}{1-r}\;\ge\;rb+\frac{r^2m}{1-r}
    \quad\Longleftrightarrow\quad
    r(b-a)\;\le\;\frac{r^2}{1-r}(M-m).
    \]
    Since the largest such gap is $\Delta$, the condition
    \begin{equation}
    \label{eq:gap-condition}
    \Delta\;\le\;
    \min\!\left\{
    \frac{r}{1-r},
    \frac{s}{1-s}
    \right\}(M-m)
    \end{equation}
    implies all adjacent pieces overlap, hence $F_r(I_r)$ is a connected compact set with the same endpoints as $I_r$, i.e. $F_r(I_r)=I_r$.
    Therefore $I_r$ is a fixed point of $F_r$.
    \item \emph{Conclude $A_r(\mathcal R)=I_r$.}
    Finally, by uniqueness of the fixed point of $F_r$ in $\mathcal K$ and since both $A_r(\mathcal R)$ and $I_r$ are fixed points, we get $A_r(\mathcal R)=I_r$.
\end{enumerate}

Applying the same argument with $s$ in place of $r$, and using \eqref{eq:gap-condition} (which implies $\Delta\le \frac{r}{1-r}(M-m)$ and $\Delta\le \frac{s}{1-s}(M-m)$), we obtain
\begin{equation}
\label{eq:aux1}
A_r(\mathcal R)=\Big[\frac{rm}{1-r},\frac{rM}{1-r}\Big],
\qquad
A_s(\mathcal R)=\Big[\frac{sm}{1-s},\frac{sM}{1-s}\Big].
\end{equation}

\medskip\noindent
\textbf{Step 2: the two-sided tail set is an interval.}

\[
I_{\mathrm{tail}}
:=A_r(\mathcal R)+A_s(\mathcal R).
\]
Since the Minkowski sum of intervals is an interval, we have
\[
I_{\mathrm{tail}}
=
\Big[\frac{rm}{1-r}+\frac{sm}{1-s},\;
\frac{rM}{1-r}+\frac{sM}{1-s}\Big].
\]

\medskip\noindent
\textbf{Step 3: adding the center digit yields $[\lambda m,\lambda M]$.}
Using $\mathrm{Sums}^{\mathcal R}_{r,s}=\mathcal R + I_{\mathrm{tail}}$, we can write
\[
\mathrm{Sums}^{\mathcal R}_{r,s}
=
\bigcup_{a\in\mathcal R} (a+I_{\mathrm{tail}}).
\]
Each set $a+I_{\mathrm{tail}}$ is an interval of length
\[
|I_{\mathrm{tail}}|
=
\Big(\frac{r}{1-r}+\frac{s}{1-s}\Big)(M-m),
\]
so adjacent pieces overlap whenever the gap between adjacent points of $\mathcal R$ is at most $|I_{\mathrm{tail}}|$.
But by \eqref{eq:gap-condition},
\[
\Delta
\le \min\Big\{\frac{r}{1-r},\frac{s}{1-s}\Big\}(M-m)
\le \Big(\frac{r}{1-r}+\frac{s}{1-s}\Big)(M-m)
=|I_{\mathrm{tail}}|,
\]
hence all adjacent intervals overlap and the union is a single interval.

The left endpoint is obtained by taking $a=m$ and the left endpoint of $I_{\mathrm{tail}}$:
\[
m+\frac{rm}{1-r}+\frac{sm}{1-s}
=\lambda m,
\]
and similarly the right endpoint is
\[
M+\frac{rM}{1-r}+\frac{sM}{1-s}
=\lambda M.
\]
Therefore $\mathrm{Sums}^{\mathcal R}_{r,s}=[\lambda m,\lambda M]$.

\noindent\textbf{\emph{(2)}}
The set $\mathrm{Closure}(\mathcal{R})$ is compact, has the same infimum and supremum as $\mathcal{R}$, and satisfies
\[
\Delta_{\mathrm{Closure}(\mathcal{R})}
\leq
\Delta_{\mathcal{R}}.
\]
Hence, $(\mathrm{Closure}(\mathcal{R}), r, s)$ is proper, and part~(3) gives
\[
\mathrm{Sums}_{\mathrm{Closure}(\mathcal{R})}^{r,s}
=
[\lambda m, \lambda M].
\]
Moreover, $\mathcal{R}^{\mathbb{Z}}$ is dense in $\mathrm{Closure}(\mathcal{R})^{\mathbb{Z}}$, and the discounted-sum function is continuous because its geometric tails converge uniformly on the bounded input domain.
Therefore,
\[
\mathrm{Closure}(\mathrm{Sums}_{\mathcal{R}}^{r,s})
=
\mathrm{Sums}_{\mathrm{Closure}(\mathcal{R})}^{r,s}
=
[\lambda m, \lambda M].
\]

\end{proof}

\noindent\textbf{Lemma~\ref{lem:uncertainty-intervals}.}
\emph{
    Let $(\mathcal{R}, r, s)$ be an input setting, let $0 \leq t \leq n$, and let $\gamma^{r,s}_{t,n} \coloneqq \frac{r^{t+1}}{1-r} + \frac{s^{n-t+1}}{1-s}$.
    Then, $U^{r,s}_t(w_{0:n}) \subseteq S^{r,s}_{t}(w_{0:n}) + \gamma^{r,s}_{t,n} [\inf \mathcal{R}, \sup \mathcal{R}]$.
    If $(\mathcal{R}, r, s)$ is proper, then $\mathrm{Closure}({U^{r,s}_t(w_{0:n})}) = S^{r,s}_{t}(w_{0:n}) + \gamma^{r,s}_{t,n} [\inf \mathcal{R}, \sup \mathcal{R}]$.
    If, additionally, $\mathcal{R}$ is compact, then equality holds without taking closures.
}

\begin{proof}
Let $m = \inf \mathcal{R}$ and $M = \sup \mathcal{R}$.
For every extension of $w_{0:n}$, the unobserved contribution to the discounted sum at time $t$ is
\[
\sum_{i = t + 1}^{\infty} r^i x_{t-i}
+
\sum_{i = n - t + 1}^{\infty} s^i x_{t+i}.
\]
The sums of the corresponding coefficients are
\[
\sum_{i = t + 1}^{\infty} r^i
=
\frac{r^{t+1}}{1-r}
\]
and
\[
\sum_{i = n - t + 1}^{\infty} s^i
=
\frac{s^{n-t+1}}{1-s}.
\]
Since every unobserved value lies in $[m, M]$, this proves
\[
U_t^{r,s}(w_{0:n})
\subseteq
S_t^{r,s}(w_{0:n})
+
\gamma_{t,n}^{r,s} [m, M].
\]
Under properness, the closures of the two one-sided tail sets are, by the argument used in Lemma~\ref{lem:possible-sums},
\[
\frac{r^{t+1}}{1-r} [m, M]
\]
and
\[
\frac{s^{n-t+1}}{1-s} [m, M].
\]
Taking their Minkowski sum proves the stated closure equality.
If $\mathcal{R}$ is compact, the two tail sets are compact and already equal to these intervals, so equality holds without closure.
\end{proof}

\noindent\textbf{Theorem~\ref{thm:impossible-sound}.}
\emph{
    Let $(\mathcal{R}, r, s)$ be a proper input setting such that $\mathcal{R}$ is compact and $r + s > 0$, and let $J \coloneqq \lambda^{r,s} [\inf \mathcal{R}, \sup \mathcal{R}]$.
    Let $\mathcal{I}$ be a target interval satisfying $J \cap \mathcal{I} \neq \emptyset$ and $J \setminus \mathcal{I} \neq \emptyset$.
    Then, for every $t \in \mathbb{N}$ and every $n \geq t$, there exists $w \in \mathcal{R}^{\mathbb{Z}}$ such that $\mu_{\mathcal{I}}(t,w) > n$.
}

\begin{proof}
Let $m = \inf \mathcal{R}$, $M = \sup \mathcal{R}$, and fix $t \in \mathbb{N}$ and $n \geq t$.
Set
\[
\Gamma
\coloneqq
\frac{r^{t+1}}{1-r}
+
\frac{s^{n-t+1}}{1-s}.
\]
The assumptions on the target imply $m < M$, and $r + s > 0$ implies $\Gamma > 0$.
For each prefix $u \in \mathcal{R}^{n+1}$, Lemma~\ref{lem:uncertainty-intervals} and compactness give
\[
U_t^{r,s}(u)
=
S_t^{r,s}(u)
+
\Gamma [m, M].
\]
These are closed intervals of the same positive length $\Gamma(M-m)$.
As $u$ ranges over all prefixes, their union is exactly
\[
\mathrm{Sums}_{\mathcal{R}}^{r,s}
=
J.
\]
Choose a boundary point $c \in J$ of $\mathcal{I}$ at which membership in $\mathcal{I}$ changes relative to $J$.
Choose points $y_k \in J$ converging to $c$ from the side whose membership differs from that of $c$.
For each $k$, choose a prefix $u_k$ such that $y_k \in U_t^{r,s}(u_k)$.
If some $U_t^{r,s}(u_k)$ also contains $c$, it contains one accepted and one rejected feasible value, and we are done.
Otherwise these fixed-length intervals lie strictly on the side of $y_k$.
Their relevant endpoints converge to $c$.
Because $\mathcal{R}^{n+1}$ is compact and the interval endpoints depend continuously on the prefix, a limiting prefix has an uncertainty interval with endpoint $c$ and with positive length on the side of the $y_k$.
This interval again contains both an accepted and a rejected feasible sum.
Hence the resulting prefix admits two completions with opposite exact verdicts.
A sound monitor must therefore remain inconclusive after observing through time $n$, so $\mu_{\mathcal{I}}(t,w) > n$ for any completion $w$ of this prefix.
\end{proof}

\noindent\textbf{Theorem~\ref{thm:approximate-monitors}.}
\emph{
    Let $(\mathcal{R}, r, s)$ be an input setting, let $\mathcal{I}$ be a target interval, let $\varepsilon > 0$, and let $T \in \mathbb{N}$.
    Define
    \begin{equation*}
        \tau^*(\mathcal{R}, r, s, \varepsilon, T)
        \coloneqq
        \min\left(
        \left\{
        \tau \in \mathbb{N} :
        d_{\mathcal{R}}
        \left(
        \frac{r^{T+1}}{1-r}
        +
        \frac{s^{\tau+1}}{1-s}
        \right)
        \leq
        2\varepsilon
        \right\}
        \cup
        \{\infty\}
        \right).
    \end{equation*}
    Then:
    \begin{enumerate}
        \item If $\tau^* < \infty$, there exists an $\varepsilon$-approximately sound monitor that produces a verdict for every $t \geq T$ within $\tau^*$ steps.
        \item Suppose that $\mathcal{R} = [m, M]$ with $m < M$, and write $J = \lambda^{r,s} [m, M]$.
        If $J \cap \mathcal{I}_{-\varepsilon} \neq \emptyset$ and $J \setminus \mathcal{I}_{+\varepsilon} \neq \emptyset$, then, for every $\varepsilon$-approximately sound monitor and every $\tau < \tau^*$, there exists $w \in \mathcal{R}^{\mathbb{Z}}$ on which the monitor does not produce a verdict for time $T$ after $\tau$ steps.
    \end{enumerate}
}

\begin{proof}
Let $m = \inf \mathcal{R}$, $M = \sup \mathcal{R}$, and $d_{\mathcal{R}} = M - m$.

\noindent
\textbf{(1) Upper bound.}
For $0 \leq t \leq n$, define the interval enclosure
\[
C_t(w_{0:n})
\coloneqq
S_t^{r,s}(w_{0:n})
+
\left(
\frac{r^{t+1}}{1-r}
+
\frac{s^{n-t+1}}{1-s}
\right) [m, M].
\]
By Lemma~\ref{lem:uncertainty-intervals},
\[
U_t^{r,s}(w_{0:n})
\subseteq
C_t(w_{0:n}).
\]
Consider the monitor that returns $\top$ if $C_t(w_{0:n}) \subseteq \mathcal{I}_{+\varepsilon}$, returns $\bot$ if $C_t(w_{0:n}) \cap \mathcal{I}_{-\varepsilon} = \emptyset$, and returns $?$ otherwise.
The inclusion above immediately gives approximate soundness.
Suppose neither decisive condition holds.
Then there are
\[
y \in C_t(w_{0:n}) \cap \mathcal{I}_{-\varepsilon}
\]
and
\[
z \in C_t(w_{0:n}) \setminus \mathcal{I}_{+\varepsilon}.
\]
Since $\mathcal{I}_{-\varepsilon} = (L + \varepsilon, U - \varepsilon)$ and $\mathcal{I}_{+\varepsilon} = (L - \varepsilon, U + \varepsilon)$, we have $|y - z| > 2\varepsilon$.
Consequently, every enclosing interval of diameter at most $2\varepsilon$ is decisive.
At observation time $n = t + \tau$, the diameter of the enclosure is
\[
d_{\mathcal{R}}
\left(
\frac{r^{t+1}}{1-r}
+
\frac{s^{\tau+1}}{1-s}
\right)
\leq
d_{\mathcal{R}}
\left(
\frac{r^{T+1}}{1-r}
+
\frac{s^{\tau+1}}{1-s}
\right)
\]
for every $t \geq T$.
The definition of $\tau^*$ therefore proves point~(1).
When $d_{\mathcal{R}} > 0$ and $0 < s < 1$, the defining inequality is equivalent to
\[
s^{\tau+1}
\leq
B_T
\coloneqq
\left(
\frac{2\varepsilon}{d_{\mathcal{R}}}
-
\frac{r^{T+1}}{1-r}
\right)(1-s).
\]
If $B_T \leq 0$, no finite $\tau$ satisfies it.
If $B_T > 0$, the least nonnegative integer solution is
\[
\max\left\{
0,
\left\lceil\log_s B_T\right\rceil - 1
\right\}.
\]
The cases $d_{\mathcal{R}} = 0$ and $s = 0$ follow directly from the defining inequality.

\noindent
\textbf{(2) Lower bound.}
By affine rescaling, it suffices to prove the result for $\mathcal{R} = [0, 1]$; the target interval and $\varepsilon$ are rescaled by the same affine map.
Write $\lambda = \lambda^{r,s}$ and fix $\tau < \tau^*$.
Set
\[
\Gamma
\coloneqq
\frac{r^{T+1}}{1-r}
+
\frac{s^{\tau+1}}{1-s}.
\]
Set
\[
A
\coloneqq
\lambda - \Gamma.
\]
By the definition of $\tau^*$,
\[
\Gamma
>
2\varepsilon.
\]
As $w_{0:T+\tau}$ ranges over $[0, 1]^{T+\tau+1}$, its observed discounted contribution at time $T$ ranges over the full interval $[0, A]$.
For a prefix having observed contribution $p \in [0, A]$, compactness and Lemma~\ref{lem:uncertainty-intervals} give
\[
U_T(w_{0:T+\tau})
=
[p, p + \Gamma].
\]
Let $\mathcal{I} = (L, U)$.
Since
\[
[0, \lambda] \setminus \mathcal{I}_{+\varepsilon}
\neq
\emptyset,
\]
there is a feasible definitely rejecting value either to the left of the target or to its right.
Suppose first that there is one on the left.
Then $L - \varepsilon \geq 0$.
The condition $[0, \lambda] \cap \mathcal{I}_{-\varepsilon} \neq \emptyset$ also gives $L + \varepsilon < \lambda$.
Since $\Gamma > 2\varepsilon$, the interval
\[
(L + \varepsilon - \Gamma, L - \varepsilon] \cap [0, A]
\]
is nonempty.
Choose $p$ in this intersection.
Then $[p, p + \Gamma]$ contains a point outside $\mathcal{I}_{+\varepsilon}$ and a point in $\mathcal{I}_{-\varepsilon}$.
If the definitely rejecting value lies to the right, then $U + \varepsilon \leq \lambda$ and $U - \varepsilon > 0$.
In this case
\[
[U + \varepsilon - \Gamma, U - \varepsilon) \cap [0, A]
\]
is nonempty.
Choosing $p$ in this intersection gives the same conclusion.
Thus in either case there is a prefix with two completions, one whose discounted sum belongs to $\mathcal{I}_{-\varepsilon}$ and one whose discounted sum lies outside $\mathcal{I}_{+\varepsilon}$.
Approximate soundness rules out both $\top$ and $\bot$ on this prefix.
Hence no approximately sound monitor can guarantee a verdict after $\tau < \tau^*$ steps.
\end{proof}

\section{Statistical Discounted Monitor}

\subsection{Remarks}
\begin{remark}
\label{remark:cond_exp}
We emphasise that $\expe_{t-1}(X_t)$ is a \emph{random variable} (it depends on the realised history up to time $t-1$), whereas $\expe(X_t)=\expe(\expe_{t-1}(X_t))$ is deterministic. 
This is the relevant notion in monitoring: at time $t$ we evaluate a decision maker relative to the realised stream (via $\filt_{t-1}$), rather than relative to a hypothetical ``future'' distribution.
\end{remark}

\begin{remark}
    \label{remark:stat:release_time}
  The deterministic register bound $\tau^*$ no longer applies automatically, because the statistical interval contains the additional term $\beta_{t,n}^{r,s}(\delta)$.
  A register can be safely released at time $n$ whenever
$\eset_t^{r,s}(\delta;W_{0:n})\subseteq \aI_{+\varepsilon}$ or
$\eset_t^{r,s}(\delta;W_{0:n})\cap \aI_{-\varepsilon}=\emptyset$.
A deterministic sufficient condition for release is $  2\beta_{t,n}^{r,s}(\delta)+d_\aR\gamma_{t,n}^{r,s}\leq 2\varepsilon .$
For pointwise bounds this condition may only be used at deterministic horizons; for
data-dependent release one must use the local or uniform bounds.
\end{remark}

\begin{remark}
    \label{remark:stat:lipschitz}
    The uniform bounds can be tightened up to constant factors by leveraging the Lipschitz constant of discounted sum properties Lemma~\ref{lem:convergingsums} to construct a grid of time points on which to perform the union bounds.
\end{remark}

\begin{remark}
    \label{remark:stat:limit}
    For the other soundness notions this pointwise optimality definition does not apply
    directly, because local and uniform soundness require validity over stretches of time.
    This allows one to trade tightness at one horizon against additional slack at another.
    Hence, Theorem~\ref{trm:minmax} is restricted to pointwise bounds. We only remark that the additional
    $\log\log$ term in local and uniform bounds is known to be necessary in classical
    undiscounted settings where $V_n\to\infty$, due to the law of the iterated
    logarithm~\cite{howard2021time}.
\end{remark}

\begin{remark}
    \label{remark:stat:average}
    The implications of the impossibility result are particularly visible for the \emph{expected discounted average}, where according to Theorem~\ref{trm:minmax} the statistical error for the two-sided limit 
\begin{align*}
    \frac{\edsum_t^{r,s}}{ \frac{1-rs}{(1-r)(1-s)}  } \quad \text{is (up to constants)} \quad O\left(\sqrt{\tfrac{1}{2}\,\sgn^2\,
\frac{(1-r)(1-s)(1+rs)}{(1+r)(1+s)(1-rs)}\,
\log(2/\delta)}\right).
\end{align*}
And as $r,s\uparrow 1$, the normalized weights spread over an increasingly long effective window approaching the limit average, and thus approaching $0$.
\end{remark}

\subsection{Proofs}

\subparagraph{Basics.}
Here we define the distribution class, bounded range, and summable weight sequence used throughout the appendix.
Let $\mathcal{D}$ be the set of all probability distributions over $\aR^\ZN$, where $\aR\subset \RN$ is bounded and closed with $a=\inf \aR$, $b=\sup \aR$, and $d_{\aR}=b-a$.
Let $\bm{w}\in \RN^\ZN$ be a weight sequence such that $\|\bm{w}\|_1\coloneqq \sum_{i\in \ZN}|w_i|<\infty$; hence also $\|\bm{w}\|_2^2\coloneqq \sum_{i\in \ZN} w_i^2<\infty$.
Fix $P\in\mathcal D$ and let $\bm{X}=(X_i)_{i\in\ZN}$ be the process with law $P$, adapted to the canonical past filtration $(\filt_i)_{i\in\ZN}$.

\subparagraph{Expectations.}
We define the predictable target and the corresponding plug-in estimator.
We write $\expe_t(\cdot)=\expe(\cdot\mid\filt_t)$ for the expectation conditioned on the past and define the predictable target and its plug-in estimator as
\begin{align*}
    \mu(P) \coloneqq \sum_{i\in\ZN} w_i\expe_{i-1}(X_i)
    \qquad\text{and}\qquad
    \hat{\mu}(\bm{X}) \coloneqq \sum_{i\in\ZN} w_iX_i .
\end{align*}
Both sums are well-defined since $X_i\in[a,b]$ and $\|\bm{w}\|_1<\infty$.

\subparagraph{Confidence intervals.}
We define the admissible confidence intervals and the minimax length criterion.
A confidence interval is a function $I\colon \aR^\ZN\to \mathrm{Interval}(\RN)$.
For $\conf\in(0,1)$ define
\begin{align*}
    \mathsf{CI}_\sgn(\conf)
    \coloneqq
    \left\{
    I \middle|
    \sup_{P\in\mathcal D_\sgn}
    \prob_P\left(\mu(P)\notin I(\bm X)\right)\leq \conf
    \right\}.
\end{align*}
For $\bm x\in\aR^\ZN$, let $|I(\bm x)|\coloneqq \sup I(\bm x)-\inf I(\bm x)$ and define the worst-case length $\ell(I)\coloneqq \sup_{\bm x\in\aR^\ZN}|I(\bm x)|$.
Our objective is to bound the minimax length $\inf_{I\in\mathsf{CI}_\sgn(\conf)}\ell(I)$.

\subparagraph{KL-Divergence bound.}
This lemma gives a quadratic upper bound on the KL divergence between two symmetric Bernoulli parameters.
\begin{lemma}
\label{lemma:kl}
Let $x\in (-1/2,1/2)$. Then
\begin{align*}
\kl\left(\tfrac12+x\ \middle\|\ \tfrac12-x\right)
\leq \frac{8x^2}{1-4x^2}.
\end{align*}
In particular, if $|x|\leq 1/4$, then
\begin{align*}
\kl\left(\tfrac12+x\ \middle\|\ \tfrac12-x\right)
\leq \frac{32}{3}x^2.
\end{align*}
\end{lemma}

\begin{proof}
A direct calculation gives
\begin{align*}
\kl\left(\tfrac12+x\ \middle\|\ \tfrac12-x\right)
&=
\left(\tfrac12+x\right)\log\frac{1+2x}{1-2x}
+
\left(\tfrac12-x\right)\log\frac{1-2x}{1+2x} =
2x\log\frac{1+2x}{1-2x}.
\end{align*}
Because we have for $x\geq 0$,
\begin{align*}
\log\frac{1+2x}{1-2x}
=
\int_{-2x}^{2x}\frac{1}{1+u}\,du
\leq
\frac{4x}{1-2|x|}
\leq
\frac{4x}{1-4x^2}
\end{align*}
we can bound the KL divergence as follows
\begin{align*}
\kl\left(\tfrac12+x\ \middle\|\ \tfrac12-x\right)
\leq
\frac{8x^2}{1-4x^2}.
\end{align*}
If $|x|\leq 1/4$, then $1-4x^2\geq 3/4$, which gives the second claim.
\end{proof}

\subparagraph{Weighted sub-Gaussian lower bound.}
This lemma proves that every valid confidence interval must have length at least the minimum of the range and sub-Gaussian scales.
\begin{lemma}[Bounded sub-Gaussian two-point lower bound]
\label{lemma:lower}
Assume that there exist $x_-,x_+\in\mathcal R$ with $x_+-x_-=2\sgn$.
Then, for every $\conf\in(0,1/4)$ and an universal constant $c>0$.
\begin{align*}
    \inf_{I\in\mathsf{CI}_\sgn(\conf)}\ell(I)
    \ \ge\
    c\,\min\left(\sgn\|\bm w\|_1,
    \sgn\|\bm w\|_2\sqrt{\log(2/\conf)}\right)
\end{align*}
\end{lemma}

\begin{proof}
If $\|\bm w\|_2=0$, the claim is trivial. Hence assume $\|\bm w\|_2>0$.
It suffices to consider product measures supported on $\{x_-,x_+\}$. Let $L\coloneqq \log(1/(4\conf))$, $B\coloneqq 1\vee L$, $a\coloneqq \|\bm w\|_2/\sqrt B$, and $\kappa\coloneqq 1/8$.
Define
\begin{align*}
    u_i\coloneqq
    \kappa\,\mathrm{sgn}(w_i)\,
    \min\left(\frac{|w_i|}{a},1\right),
    \qquad
    p_i\coloneqq \frac12+u_i,
    \qquad
    q_i\coloneqq \frac12-u_i .
\end{align*}
Then $|u_i|\leq \kappa\leq 1/4$ and
\begin{align*}
    \sum_i u_i^2
    \leq
    \kappa^2\sum_i \frac{w_i^2}{a^2}
    =
    \kappa^2 B .
\end{align*}
Let $P_p$ and $P_q$ be the product measures under which $X_i=x_+$ with probability $p_i$ and $q_i$, respectively, and $X_i=x_-$ otherwise.
For every coordinate, the centred variable is supported in an interval of length $2\sgn$, and hence is $\sgn$-sub-Gaussian by Hoeffding's lemma. Thus $P_p,P_q\in\mathcal D_\sgn$.
Writing $\mu_p\coloneqq \mu(P_p)$ and $\mu_q\coloneqq \mu(P_q)$, we get
\begin{align*}
    |\mu_p-\mu_q|
    =
    2\sgn\left|\sum_i w_i(p_i-q_i)\right|  =
    4\sgn\sum_i w_i u_i =
    4\sgn\kappa
    \sum_i |w_i|\min\left(\frac{|w_i|}{a},1\right).
\end{align*}
Set $A\coloneqq \|\bm w\|_1$ and $C\coloneqq \|\bm w\|_2\sqrt B=\|\bm w\|_2^2/a$.
Using the elementary inequality
\begin{align*}
    \sum_i \min\left(|w_i|,\frac{w_i^2}{a}\right)
    \geq
    \frac{AC}{A+C}
    \geq
    \frac12\min(A,C),
\end{align*}
we obtain
\begin{align}
    \Delta\coloneqq |\mu_p-\mu_q|
    \geq
    2\sgn\kappa\,
    \min\left(\|\bm w\|_1,\|\bm w\|_2\sqrt B\right).
    \label{eq:lower:separation}
\end{align}
Next, by Lemma~\ref{lemma:kl},
\begin{align*}
\kl(P_p\|P_q)
&=
\sum_i
\kl\left(\tfrac12+u_i\ \middle\|\ \tfrac12-u_i\right)  \leq
\frac{32}{3}\sum_i u_i^2
\leq
\frac{32}{3}\kappa^2 B
=
\frac16 B .
\end{align*}

Suppose, for contradiction, that $\ell(I)<\Delta$ for some $I\in\mathsf{CI}_\sgn(\conf)$.
Define the test $\varphi(\bm X)\coloneqq \indi[\mu_p\in I(\bm X)]$.
Since an interval of length less than $\Delta$ cannot contain both $\mu_p$ and $\mu_q$, coverage of $I$ under $P_p$ and $P_q$ implies $ \prob_{P_p}(\varphi=0)\leq \conf$ and $  \prob_{P_q}(\varphi=1)\leq \conf $.
If $L\geq 1$, then $B=L$ and the Bretagnolle--Huber inequality gives
\begin{align*}
    2\conf
    &\geq
    \prob_{P_p}(\varphi=0)+\prob_{P_q}(\varphi=1) \geq
    \frac12\exp(-\kl(P_p\|P_q))
    >
    \frac12 e^{-L}
    =
    2\conf,
\end{align*}
a contradiction. If $L<1$, then $B=1$ and Pinsker's inequality gives
\begin{align*}
    \mathrm{TV}(P_p,P_q)
    \leq
    \sqrt{\frac{\kl(P_p\|P_q)}{2}}
    \leq
    \sqrt{\frac{1}{12}}
    <
    \frac12 .
\end{align*}
Thus every test satisfies
\begin{align*}
    \prob_{P_p}(\varphi=0)+\prob_{P_q}(\varphi=1)
    &\geq
    1-\mathrm{TV}(P_p,P_q)
    >
    \frac12
    >
    2\conf,
\end{align*}
again a contradiction.
Hence $\ell(I)\geq \Delta$ for every $I\in\mathsf{CI}_\sgn(\conf)$.
Finally, since $B=1\vee\log(1/(4\conf))$ is comparable to $\log(2/\conf)$ on $\conf\in(0,1/4)$, Eq.~\eqref{eq:lower:separation} yields the stated bound after absorbing constants.
\end{proof}

\subparagraph{Sub-Gaussian MGF bound.}
This lemma shows that the weighted estimation error is itself sub-Gaussian with variance proxy $\sgn^2\|\bm w\|_2^2$.

\begin{lemma}[Sub-Gaussian MGF bound]
\label{lemma:mgf}
For every $P\in\mathcal D_\sgn$ and every $\lambda\in\RN$,
\begin{align*}
    \expe_P\left[\exp\left(\lambda(\hat{\mu}(\bm X)-\mu(P))\right)\right]
    \leq
    \exp\left(\frac{\lambda^2}{2}\sgn^2\|\bm w\|_2^2\right).
\end{align*}
\end{lemma}

\begin{proof}
For $m\in\NN$ define the truncations
\begin{align*}
\mu_m(P)\coloneqq \sum_{i=-m}^m w_i\expe_{i-1}(X_i),
\qquad
\hat{\mu}_m(\bm X)\coloneqq \sum_{i=-m}^m w_iX_i,
\qquad
S_m\coloneqq \hat{\mu}_m(\bm X)-\mu_m(P).
\end{align*}
Then $S_m=\sum_{i=-m}^m w_i\left(X_i-\expe_{i-1}(X_i)\right)$.
Since $X_i-\expe_{i-1}(X_i)$ is conditionally $\sgn$-sub-Gaussian,
\begin{align*}
\expe\left[
\exp\left(\lambda w_i(X_i-\expe_{i-1}(X_i))\right)
\middle|\filt_{i-1}\right]
\leq
\exp\left(\frac{\lambda^2 w_i^2\sgn^2}{2}\right)
\quad\text{a.s.}
\end{align*}
Iterating the tower property gives
\begin{align*}
\expe[e^{\lambda S_m}]
\leq
\exp\left(\frac{\lambda^2\sgn^2}{2}\sum_{i=-m}^m w_i^2\right).
\end{align*}
Since $\|\bm w\|_1<\infty$, we have $S_m\to \hat{\mu}(\bm X)-\mu(P)$ almost surely. By applying Fatou's lemma yields
\begin{align*}
\expe\left[e^{\lambda(\hat{\mu}(\bm X)-\mu(P))}\right]
\leq
\exp\left(\frac{\lambda^2\sgn^2}{2}\|\bm w\|_2^2\right).
\end{align*}
\end{proof}

\subparagraph{Sub-Gaussian variance bound.}
This lemma converts the sub-Gaussian MGF bound into a variance bound.

\begin{lemma}[Variance bound]
\label{lemma:var}
For every $P\in\mathcal D$,
\begin{align*}
\Var_P(\hat{\mu}(\bm{X})-\mu(P)) \ \le\ \frac{d_{\aR}^2}{4}\|\bm{w}\|_2^2.
\end{align*}
\end{lemma}

\begin{proof}
By Lemma~\ref{lemma:mgf}, the random variable $Y\coloneqq \hat{\mu}(\bm{X})-\mu(P)$ is sub-Gaussian with proxy variance $(d_{\aR}\|\bm{w}\|_2/2)^2$.
For any centred sub-Gaussian $Y$ with $\expe[e^{\lambda Y}]\leq \exp(\lambda^2\sigma^2/2)$, differentiating at $\lambda=0$ gives $\Var(Y)\leq \sigma^2$.
Applying this with $\sigma^2=(d_{\aR}\|\bm{w}\|_2/2)^2$ proves the claim.
\end{proof}

\subparagraph{Sub-Gaussian upper bound.}
This lemma constructs confidence intervals whose lengths match the lower bound up to constants.

\begin{lemma}[Upper bound]
\label{lemma:upper}
For every $\conf\in(0,1)$,
\begin{align*}
      \inf_{I \in \mathsf{CI}_\sgn(\conf)}\ell(I)
      \ \le\
      C\,\min\left(
      d_{\aR}\|\bm w\|_1,
      \sgn\|\bm w\|_2\sqrt{\log(2/\conf)}
      \right)
\end{align*}
for a universal constant $C>0$.
\end{lemma}

\begin{proof}
First, the deterministic range bound gives the interval
\begin{align*}
I_{\mathrm{rng}}(\bm x)\coloneqq
\left[
\sum_{i:w_i\geq 0}w_i a+\sum_{i:w_i<0}w_i b,\,
\sum_{i:w_i\geq 0}w_i b+\sum_{i:w_i<0}w_i a
\right].
\end{align*}
Since $\expe_{i-1}(X_i)\in[a,b]$ almost surely, this interval covers $\mu(P)$ for every $P\in\mathcal D_\sgn$. Its length is $d_\aR\|\bm w\|_1$.
Second, by Chernoff's method and Lemma~\ref{lemma:mgf},
\begin{align*}
\sup_{P\in\mathcal D_\sgn}
\prob_P\left(|\hat{\mu}(\bm X)-\mu(P)|\geq \varepsilon\right)
\leq
2\exp\left(-\frac{\varepsilon^2}{2\sgn^2\|\bm w\|_2^2}\right).
\end{align*}
Choosing $\varepsilon=\sqrt{2\sgn^2\|\bm w\|_2^2\log(2/\conf)}$ yields the interval
\begin{align*}
I_H(\bm x)\coloneqq
\hat{\mu}(\bm x)\pm
\sqrt{2\sgn^2\|\bm w\|_2^2\log\frac{2}{\conf}}
\in\mathsf{CI}_\sgn(\conf),
\end{align*}
whose length is $2\sqrt{2\sgn^2\|\bm w\|_2^2\log(2/\conf)}$.
Taking the better of $I_{\mathrm{rng}}$ and $I_H$ proves the claim.
\end{proof}

\subparagraph{Generic minimax theorem.}
This theorem combines the upper and lower bounds into a minimax characterisation for weighted sums.

\begin{theorem}[Minimax confidence interval length]
\label{trm:app_minmax}
For every $\conf\in(0,1/4)$,
\begin{align*}
c\min\left(
\sgn\|\bm w\|_1,
\sgn\|\bm w\|_2\sqrt{\log(2/\conf)}
\right)
\leq
\inf_{I\in\mathsf{CI}_\sgn(\conf)}\ell(I)
\leq
C\min\left(
d_{\aR}\|\bm w\|_1,
\sgn\|\bm w\|_2\sqrt{\log(2/\conf)}
\right)
\end{align*}
for universal constants $0<c<C<\infty$.
\end{theorem}

\begin{proof}
Combine Lemma~\ref{lemma:lower} and Lemma~\ref{lemma:upper}.
\end{proof}

\subparagraph{Discounted sum minimax theorem.}
This theorem specialises the generic weighted-sum minimax bound to finite discounted-sum weights.

\begin{theorem}[Specialisation to discounted sums]
\label{trm:minmax:discounted}
Fix $r,s\in[0,1)$, a center time $t\in\{0,\dots,n\}$, and a horizon $n\geq t$.
Consider the finite discounted-sum weights
\begin{align*}
w_{t-i}=r^i\ (i=1,\dots,t),
\qquad
w_t=1,
\qquad
w_{t+i}=s^i\ (i=1,\dots,n-t),
\qquad
w_k=0\ \text{otherwise}.
\end{align*}
Then
\begin{align*}
\|\bm w\|_2^2
=
\sum_{i=1}^{t} r^{2i}+\sum_{i=0}^{n-t}s^{2i}
=
\dsgn_{t,n}^{r,s},
\qquad
\|\bm w\|_1
=
\sum_{i=1}^{t}r^i+\sum_{i=0}^{n-t}s^i
=
\eta_{t,n}^{r,s}.
\end{align*}
Consequently,
\begin{align*}
c_0 \cdot
\min\left(\sgn\eta_{t,n}^{r,s}, \pse_{t,n}^{r,s}(\delta)\right)
&\leq
\inf_{I\in\mathsf{CI}_\sgn(\conf)}
\sup_{P\in\mathcal D_\sgn}
\expe_P\left[|I(W_{0:n})|\right]  \leq
C_0 \cdot
\min\left(d_\aR\eta_{t,n}^{r,s}, \pse_{t,n}^{r,s}(\delta)\right)
\end{align*}
for universal constants $0<c_0<C_0<\infty$.
\end{theorem}

\begin{proof}
The identities for $\|\bm w\|_1$ and $\|\bm w\|_2^2$ follow by geometric-series summation.
For these weights, $\hat\mu(\bm X)=\dsum_t^{r,s}(W_{0:n})$ and $\mu(P)=\edsum_t^{r,s}(W_{0:n})$.
Moreover,
\begin{align*}
\sgn\|\bm w\|_2\sqrt{\log(2/\delta)}
\asymp
\sqrt{2\sgn^2\dsgn_{t,n}^{r,s}\log(2/\delta)}
=
\pse_{t,n}^{r,s}(\delta),
\end{align*}
where the constants are universal.
The proof concludes by applying Theorem~\ref{trm:app_minmax} and then Lemma~\ref{lemma:exp-vs-wc}.
\end{proof}

\subparagraph{Paper minimax theorem.}
This paragraph derives the paper's minimax theorem from the generic weighted-sum result.

\StatMinMax*

\begin{proof}[Proof of Theorem~\ref{trm:minmax}]
This is a direct instantiation of Theorem~\ref{trm:app_minmax} (the weighted-sum minimax theorem), with the finite-horizon discounted-sum weights
\begin{align*}
w_{t-i}=r^i,\quad w_t=1,\quad w_{t+i}=s^i,\quad (i\geq 1),\qquad w_k=0\ \text{otherwise on }\{0,\dots,n\}.
\end{align*}
For these weights,
\begin{align*}
\|\bm{w}\|_2^2=\sum_{i=1}^{t} r^{2i}+\sum_{i=0}^{n-t}s^{2i}=\dsgn_{t,n}^{r,s},
\qquad
\|\bm{w}\|_1=\sum_{i=1}^{t} r^{i}+\sum_{i=0}^{n-t}s^{i}=\eta_{t,n}^{r,s},
\end{align*}
and $\hat\mu(\bm{X})=\dsum_t^{r,s}(W_{0:n})$, $\mu(P)=\edsum_t^{r,s}(W_{0:n})$.
Substituting these identities into Theorem~\ref{trm:app_minmax} yields the claimed bound.
Finally, Lemma~\ref{lemma:exp-vs-wc} identifies worst-case length with worst-case expected length over the stated class, so the minimax statements match.
\end{proof}

\subparagraph{Expected and worst-case length.}
This lemma shows that worst-case expected length and worst-case pointwise length coincide when the class contains all Dirac measures.

\begin{lemma}[Expected vs.\ worst-case length]
\label{lemma:exp-vs-wc}
Fix $n\in\NN$ and a distribution class $\mathcal D$ on $\mathcal R^{\{0,\dots,n\}}$ that contains, for every $\bm x\in\mathcal R^{\{0,\dots,n\}}$, the degenerate (Dirac) product measure $P_{\bm x}$ with $P_{\bm x}(W_{0:n}=\bm x)=1$.
For any interval-valued map $I:\mathcal R^{\{0,\dots,n\}}\to \mathrm{Interval}(\RN)$ define
\begin{align*}
\ell(I)\coloneqq \sup_{\bm x\in\mathcal R^{\{0,\dots,n\}}}|I(\bm x)| \quad \text{then} \quad \sup_{P\in\mathcal D}\ \expe_P\left[|I(W_{0:n})|\right]=\ell(I).
\end{align*}
In particular, for any constraint set $\mathsf{CI}_\sgn(\conf)$,
\begin{align*}
\inf_{I\in\mathsf{CI}_\sgn(\conf)}\ \sup_{P\in\mathcal D}\ \expe_P\left[|I(W_{0:n})|\right]
=
\inf_{I\in\mathsf{CI}_\sgn(\conf)}\ \ell(I).
\end{align*}
\end{lemma}

\begin{proof}
Fix $I$.

\emph{(Upper bound)} For any $P\in\mathcal D$ we have $|I(W_{0:n})|\leq \ell(I)$ almost surely, hence $\expe_P[|I(W_{0:n})|]\leq \ell(I)$.
Taking the supremum over $P$ gives
\begin{align*}
\sup_{P\in\mathcal D}\expe_P[|I(W_{0:n})|]\leq \ell(I).
\end{align*}

\emph{(Lower bound)} For every $\varepsilon>0$ choose $\bm x_\varepsilon$ with $|I(\bm x_\varepsilon)|\geq \ell(I)-\varepsilon$ by definition of the supremum.
Since $P_{\bm x_\varepsilon}\in\mathcal D$ and $W_{0:n}=\bm x_\varepsilon$ a.s.\ under $P_{\bm x_\varepsilon}$,
\begin{align*}
\expe_{P_{\bm x_\varepsilon}}[|I(W_{0:n})|]=|I(\bm x_\varepsilon)|\geq \ell(I)-\varepsilon.
\end{align*}
Thus $\sup_{P\in\mathcal D}\expe_P[|I(W_{0:n})|]\geq \ell(I)-\varepsilon$ for all $\varepsilon>0$, so letting $\varepsilon\downarrow 0$ yields
\begin{align*}
\sup_{P\in\mathcal D}\expe_P[|I(W_{0:n})|]\geq \ell(I).
\end{align*}
Combining both bounds gives equality. The minimax identity follows by taking $\inf_{I\in\mathsf{CI}_\sgn(\conf)}$.
\end{proof}

\subparagraph{Soundness of tail completion.}
This lemma shows that finite-sum coverage lifts to infinite-sum coverage after adding the deterministic discounted tail.

\begin{restatable}{lemma}{StatLift}
    \label{lemma:stat:lift}
    Let $(\mathcal R,r,s)$ be an input setting and let $\Gamma\coloneqq[\inf\mathcal R,\sup\mathcal R]$.
    For every $t\in\NN$ and $n\geq t$ $  \left\{
        \edsum_t^{r,s}(W_{0:n})\in
        \einter_t^{r,s}(\delta;W_{0:n})
        \right\}
        \subseteq
        \left\{
        \edsum_t^{r,s}(W)\in
        \eset_t^{r,s}(\delta;W_{0:n})
        \right\}$.
    Consequently,
    \begin{align*}
        &\prob\left(
        \edsum_t^{r,s}(W)\in \eset_t^{r,s}(\delta;W_{0:n})
        \right)
        \geq
        \prob\left(
        \edsum_t^{r,s}(W_{0:n})\in \einter_t^{r,s}(\delta;W_{0:n})
        \right),\\
        &\prob\left(
        \forall n\geq t:\ \edsum_t^{r,s}(W)\in \eset_t^{r,s}(\delta;W_{0:n})
        \right)
        \geq
        \prob\left(
        \forall n\geq t:\ \edsum_t^{r,s}(W_{0:n})\in
        \einter_t^{r,s}(\delta;W_{0:n})
        \right),\\
        &\prob\left(
        \forall t\in\NN\,\forall n\geq t:\ \edsum_t^{r,s}(W)\in
        \eset_t^{r,s}(\delta;W_{0:n})
        \right)
        \geq
        \prob\left(
        \forall t\in\NN\,\forall n\geq t:\ \edsum_t^{r,s}(W_{0:n})\in
        \einter_t^{r,s}(\delta;W_{0:n})
        \right).
    \end{align*}
\end{restatable}

\begin{proof}[Proof of Lemma~\ref{lemma:stat:lift}]
Fix $t\in\NN$ and $n\geq t$.
Write the infinite expected discounted sum as the finite observed part plus the unobserved tails:
\begin{align*}
\edsum_t^{r,s}(W)-\edsum_t^{r,s}(W_{0:n})
=
\sum_{i=t+1}^{\infty} r^i\expe_{t-i-1}(X_{t-i})
+
\sum_{i=n-t+1}^{\infty} s^i\expe_{t+i-1}(X_{t+i}) .
\end{align*}
Since $X_k\in\mathcal R$ almost surely, also $\expe_{k-1}(X_k)\in\Gamma$ almost surely.
Hence
\begin{align*}
\edsum_t^{r,s}(W)-\edsum_t^{r,s}(W_{0:n})
\in
\left(
\sum_{i=t+1}^{\infty} r^i+\sum_{i=n-t+1}^{\infty}s^i
\right)\Gamma
=
\gamma_{t,n}^{r,s}\cdot\Gamma .
\end{align*}
Therefore, whenever $\edsum_t^{r,s}(W_{0:n})\in\einter_t^{r,s}(\delta;W_{0:n})$, we have
\begin{align*}
\edsum_t^{r,s}(W)
&\in
\einter_t^{r,s}(\delta;W_{0:n})+\gamma_{t,n}^{r,s}\cdot\Gamma=\eset_t^{r,s}(\delta;W_{0:n}).
\end{align*}
This proves the event inclusion.
The pointwise, local, and uniform inequalities follow by applying the same inclusion respectively for fixed $(t,n)$, for fixed $t$ and all $n\geq t$, and for all $t\in\NN$ and all $n\geq t$.
\end{proof}

\subparagraph{Pointwise soundness.}
This lemma gives a fixed-time confidence interval for the finite expected discounted sum.

\begin{restatable}{lemma}{StatPoint}
    \label{lemma:stat:point}
    Let $(\mathcal R, r,s)$ be a input setting.
    Then for every $\delta\in (0,1)$ and all $n\in \NN$, $t\in\{0,\dots,n\}$,
    \begin{align*}
         \prob\left(\edsum_t^{r, s}(W_{0:n}) \in \dsum_t^{r, s}(W_{0:n}) \pm \beta_{t,n}^{r,s}(\delta)\right)\geq  1-\delta,
         \qquad
         \beta_{t,n}^{r,s}(\delta) \coloneqq \sqrt{2\sgn^2 \dsgn_{t,n}^{r,s} \log(2/\delta)}.
    \end{align*}
\end{restatable}

\begin{proof}[Proof of Lemma~\ref{lemma:stat:point}]
Fix $n\in\NN$ and $t\in\{0,\dots,n\}$.
Define weights $(w_k)_{k=0}^n$ by
\begin{align*}
w_{t-i}=r^i\ (i=1,\dots,t),\qquad w_t=1,\qquad w_{t+i}=s^i\ (i=1,\dots,n-t).
\end{align*}
Then
\begin{align*}
\dsum_t^{r,s}(W_{0:n})-\edsum_t^{r,s}(W_{0:n})
=\sum_{k=0}^n w_k\left(X_k-\expe_{k-1}(X_k)\right).
\end{align*}
By Lemma~\ref{lemma:mgf} applied with sub-Gaussian proxy $\sgn$ in place of $d_\aR$, for all $\lambda\in\RN$,
\begin{align*}
\expe\left( \exp \left(\lambda\sum_{k=0}^n w_k(X_k-\expe_{k-1}(X_k))\right)\right)
\le
\exp\left(\frac{\lambda^2}{2}\sgn^2\sum_{k=0}^n w_k^2\right).
\end{align*}
Chernoff’s method gives
\begin{align*}
\prob\left(\dsum_t^{r,s}(W_{0:n})-\edsum_t^{r,s}(W_{0:n})\geq \varepsilon\right)
\leq \exp\left(-\frac{\varepsilon^2}{2\sgn^2\sum_{k=0}^n w_k^2}\right),
\end{align*}
and the same for the left tail.
Using $\sum_{k=0}^n w_k^2=\dsgn_{t,n}^{r,s}$ and setting $\varepsilon=\sqrt{2\sgn^2\dsgn_{t,n}^{r,s}\log(2/\delta)}=\beta_{t,n}^{r,s}(\delta)$ results in $\prob(|\dsum-\edsum|\geq \beta_{t,n}^{r,s}(\delta))\leq \delta$.
\end{proof}

\subparagraph{Local soundness.}
This lemma gives an anytime-valid confidence interval over all horizons $n\geq t$ for a fixed center time $t$.

\begin{restatable}{lemma}{StatLocal}
\label{lemma:stat:local}
Let $(\mathcal R,r,s)$ be a input setting, $t\in\mathbb N$ and $\delta\in (0,1)$. Then:
\begin{align*}
&\prob\left(\forall n\geq t:\ \edsum_t^{r,s}(W_{0:n}) \in \dsum_t^{r,s}(W_{0:n}) \pm \beta_{t,n}^{r,s}(\delta)\right)\ \ge\ 1-\delta, \\
&\beta_{t,n}^{r,s}(\delta)
\coloneqq
k_1\sqrt{V_{t,n}^{r,s} \left(2\log\left(\log_2(V_{t,n}^{r,s} )+1\right) + \log\left(\tfrac{2\pi^2}{6\delta}\right)\right)}
\end{align*}
where $V_{t,n}^{r,s}\coloneqq \max(1,\sgn^2\dsgn_{t,n}^{r,s})$ and $k_1 \coloneqq 2^{1/4}+2^{-1/4}/\sqrt 2$.
\end{restatable}

\begin{proof}[Proof of Lemma~\ref{lemma:stat:local}]
Let $S_n\coloneqq \dsum_t^{r,s}(W_{0:n})-\edsum_t^{r,s}(W_{0:n})$.
For fixed $t$, define weights $w_k^{(t)}$ by $w_{t-i}^{(t)}=r^i$ for $i=1,\dots,t$ and $w_{t+i}^{(t)}=s^i$ for $i\geq 0$.
For $m\geq 0$, let
\begin{align*}
S_m\coloneqq \sum_{k=0}^m w_k^{(t)}
\left(X_k-\expe_{k-1}(X_k)\right).
\end{align*}
Then $(S_m)_{m\geq 0}$ is a martingale with conditionally $\sgn$-sub-Gaussian increments, i.e.,
\begin{align*}
V_{t,n}^{r,s}\coloneqq \max\left(1,\ \sgn^2\sum_{k=0}^n w_k^2\right)=\max(1,\sgn^2\dsgn_{t,n}^{r,s}).
\end{align*}
Then for every $n\geq t$, $S_n=\dsum_t^{r,s}(W_{0:n})-\edsum_t^{r,s}(W_{0:n})$, we apply the stitched sub-Gaussian uniform boundary of~\cite[Thm.~1]{howard2021time} with $\eta=2$, $m=1$, and $h(j)=a(j+1)^2$ where $a=\pi^2/6$, and union bound the two one-sided boundaries, splitting $\delta$ into $\delta/2$, to obtain
\begin{align*}
\prob\left(\forall n\geq t:\ |S_n|\leq \beta_{t,n}^{r,s}(\delta)\right)\geq 1-\delta,
\end{align*}
with $\beta_{t,n}^{r,s}(\delta)$ exactly as stated after substituting $V_{t,n}^{r,s}$.
Finally, $|S_n|\leq \beta_{t,n}^{r,s}(\delta)$ is equivalent to $\edsum_t^{r,s}(W_{0:n})\in \dsum_t^{r,s}(W_{0:n})\pm \beta_{t,n}^{r,s}(\delta)$.
\end{proof}

\subparagraph{Uniform soundness.}
This lemma obtains simultaneous coverage over all center times $t$ and all horizons $n\geq t$ by a union bound over the local guarantees.

\begin{restatable}{lemma}{StatUniform}
\label{lemma:stat:uniform}
Let $(\mathcal R,r,s)$ be a input setting and $\delta\in (0,1)$. Then:
\begin{align*}
&\prob\left(\forall t\in \NN \forall n\geq t:\ \edsum_t^{r,s}(W_{0:n}) \in \dsum_t^{r,s}(W_{0:n}) \pm \beta_{t,n}^{r,s}(\delta)\right)\ \ge\ 1-\delta, \\
&\beta_{t,n}^{r,s}(\delta)
\coloneqq
k_1\sqrt{V_{t,n}^{r,s} \left(2\log\left(\log_2(V_{t,n}^{r,s} )+1\right) + \log\left(\tfrac{2\pi^2}{6\delta_t}\right)\right)}
\quad \text{where} \quad
\delta_t \coloneqq \frac{6\delta}{\pi^2(t+1)^2}.
\end{align*}
where $V_{t,n}^{r,s}\coloneqq \max(1,\sgn^2\dsgn_{t,n}^{r,s})$ and $k_1 \coloneqq 2^{1/4}+2^{-1/4}/\sqrt 2$.
\end{restatable}

\begin{proof}[Proof of Lemma~\ref{lemma:stat:uniform}]
This follows directly from a simple union bound.
\begin{align*}
    &\prob\left(\forall t\in \NN \forall n\geq t \colon \edsum_t^{r,s}(W_{0:n}) \in \dsum_t^{r,s}(W_{0:n}) \pm \beta_{t,n}^{r,s}(\delta)\right) \\
    &= 1- \prob\left(\exists t\in \NN \exists n\geq t \colon \edsum_t^{r,s}(W_{0:n}) \not\in \dsum_t^{r,s}(W_{0:n}) \pm \beta_{t,n}^{r,s}(\delta)\right)  \\
    &\geq 1- \sum_{t=0}^\infty \prob\left(\exists n\geq t \colon \edsum_t^{r,s}(W_{0:n}) \not\in \dsum_t^{r,s}(W_{0:n}) \pm \beta_{t,n}^{r,s}(\delta)\right)  \\
     &\geq 1- \sum_{t=0}^\infty  \delta_t
    = 1- \sum_{t=0}^\infty 
      \frac{6\delta}{\pi^2(t+1)^2} = 1-\delta.
\end{align*}
\end{proof}

\subparagraph{Paper soundness.}
This paragraph derives the paper's statistical soundness theorem from finite-sum coverage and tail completion.

\Soundness*

\begin{proof}
By Lemmas~\ref{lemma:stat:point},~\ref{lemma:stat:local}, and~\ref{lemma:stat:uniform}, respectively, the intervals $\einter_t^{r,s}(\delta;W_{0:n})$ cover the finite expected discounted sum $\edsum_t^{r,s}(W_{0:n})$ with the pointwise, local, and uniform guarantees.
By Lemma~\ref{lemma:stat:lift}, the corresponding tail-completed intervals $\eset_t^{r,s}(\delta;W_{0:n})$ cover the infinite expected discounted sum $\edsum_t^{r,s}(W)$ with the same type of guarantee.
On this coverage event, if the monitor outputs $\top$, then $  \edsum_t^{r,s}(W)
    \in
    \eset_t^{r,s}(\delta;W_{0:n})
    \subseteq
    \aI_{+\varepsilon}.$
If the monitor outputs $\bot$, then
\begin{align*}
      \edsum_t^{r,s}(W)
    \in
    \eset_t^{r,s}(\delta;W_{0:n})
    \quad\text{and}\quad
    \eset_t^{r,s}(\delta;W_{0:n})\cap\aI_{-\varepsilon}=\emptyset,
\end{align*}
and therefore $\edsum_t^{r,s}(W)\notin\aI_{-\varepsilon}$.
This is exactly the required statistical approximate soundness.
The pointwise case applies only when the verdict horizon is fixed in advance.
\end{proof}

\section{General Discounted Properties}

\subsection{Synchronous vs Asynchronous example}

In Table~\ref{tab:example1}, we show an example of a sequence of events with two types,  to illustrate the evolution of the values and the corresponding asynchronous discounting.
\begin{table}[h]
\centering
\begin{tabular}{ll|ccccccccc}
& $i$       & 0       & 1                  & 2                  & 3 & 4       & 5       & 6 & 7       & $8\:\dots$  \\ \hline
& $e^{(1)}_i$ & 1       & 1                  & $\empt$            & 0 & 1       & 0       & 0 & $\empt$ & $1 \:\dots$ \\
& $e^{(2)}_i$ & $\empt$ & 0                  & 1                  & 1 & $\empt$ & $\empt$ & 0 & $\empt$ & $1\:\dots$  \\ \hline
& $x^{(1)}_i$ & 1       & $1$ & $0$ & $0$   & $1$         & $0$        & $0$  &  $0$       & $1\dots$          \\
& $x^{(2)}_i$ & 0       & $0$    & $1$  & $1$  & $0$        &  $0$   & $0$   & $0$         & $1\:\dots$ \\ \hline
$\eval{\cdot}_S $& $\tau_{0,+i}^{(k)}$       & 0       & 1                  & 2                  & 3 & 4       & 5       & 6 & 7       & $8\:\dots$  \\ 
 \hline
$\eval{\cdot}_A$& $\tau_{0,+i}^{(1)}$ & 0       & $1$ & $1$ & $2$   & $3$         & $4$        & $5$  &  $5$       & $6\:\dots$          \\
$\eval{\cdot}_A$& $\tau_{0,+i}^{(2)}$ & 0       & $0$    & $1$  & $2$  & $2$        &  $2$   & $3$   & $3$         & $4\:\dots$ \\
\end{tabular}
\caption{Excerpt of an event sequence, with the corresponding values and discount factors.}
\label{tab:example1}
\end{table}

\section{Register Complexity of Discounted-Sum Monitoring}

\noindent\textbf{Theorem~\ref{thm:upper-bound}.}
\emph{
	For every monitoring instance $\mathcal{P}$, there is an ARM $\mathcal{M}$ with $\tau(\mathcal{P}) + 1$ registers such that $L(\mathcal{M}) = L(\mathcal{P})$.
}
\begin{proof}
	Let $\tau = \tau(\mathcal{P})$.
	We construct an ARM $\mathcal{M} = (\mathcal{R}, Y, Q, q_0, \nu_0, \Delta)$
	Let $Y = \{y_0, \ldots, y_{\tau-1}, \rho\}$, where $y_0, \ldots, y_{\tau-1}$ store running sums for pending positions, and $\rho$ tracks $r^{n+1}$ before reading the $n$th input.
	Let $Q = \{q_c : 0 \leq c \leq \tau\}$, where $c$ counts how many positions are currently pending, capped at $\tau$.
	The initial location is $q_0$ with $\nu_0(y_i)=0$ for all $i$ and $\nu_0(\rho)=r$.

	From location $q_c$ on input $x\in\mathcal{R}$, $\mathcal{M}$ applies the parallel update $\rho \gets r \rho$, $y_0 \gets r y_0 + x$, and $y_a \gets y_{a-1} + s^a x$ for $a \in \{ 1,\ldots,\min\{c,\tau-1\} \}$, and moves to $q_{\min\{c+1,\tau\}}$.
	In location $q_\tau$, before applying the update above, $\mathcal{M}$ tests whether the oldest pending position $t$ can be resolved using the current input $x$.
	At this step $t=n-\tau$, and by the maintained invariant $\rho=r^{n+1}$ and $y_{\tau-1}=S_t^{r,s}(x_{0:t+\tau-1})$.
	Thus the completed truncated value is $V := y_{\tau-1} + s^\tau x \;=\; S_t^{r,s}(x_{0:t+\tau})$
    Writing $\delta_f := \frac{s^{\tau+1}}{1-s}$, we define $\kappa_{r,\tau} := \dfrac{r^{-\tau}}{1-r}$ if $r > 0$, and $\kappa_{r,\tau} := 0$ otherwise.
    Since $t = n-\tau$ and $\rho=r^{n+1}$, we have $\delta_p(t) = \frac{r^{t+1}}{1-r} = \kappa_{r,\tau}\rho$.
    Therefore, $\mathcal{M}$ includes in $q_\tau$ a self-loop guarded by $V > L-\varepsilon \land V + \kappa_{r,\tau}\rho < U + \varepsilon - \delta_f$, and blocks if the guard fails.

	By induction on $n$, just before reading $x_n$ we have $\rho=r^{n+1}$ and for every age $a<\min\{n,\tau\}$,
	$y_a = S^{r,s}_{n-1-a}(x_{0:n-1})$.
	In particular, in $q_\tau$ the guard is satisfied iff
	$U_t^{r,s}(x_{0:t+\tau})\subseteq \mathcal{I}_{+\varepsilon}$ for $t=n-\tau$.
	Hence $\mathcal{M}$ blocks on an input stream exactly when some position violates the monitoring condition, and therefore $L(\mathcal{M}) = L(\mathcal{P})$.
\end{proof}

\noindent\textbf{Lemma~\ref{lem:Pk-characterization}.}
\emph{
	For every $k \geq 1$, the monitoring instance $\mathcal P_k$ satisfies $\tau(\mathcal P_k)=k$ and $L({\mathcal P_k}) = \{x\in [0,1]^\omega \mid \forall t\geq 0:\; V_t^k(x) < T_k\}$.
}
\begin{proof}
	By \Cref{thm:approximate-monitors}, $\tau(\mathcal P_k)=\tau^*([0,1],\,0,\,\tfrac12,\,2^{-(k+1)},\,0)=\left\lceil \log_{1/2}(2^{-(k+1)})\right\rceil - 1 = k$.

	Since $r=0$, for every $t\geq 0$ we have $S_t^{0,1/2}(x_{0:t+k}) = \sum_{i=0}^{k}2^{-i}x_{t+i} = V_t^k(x)$.
    The residual uncertainty after observing $k$ future symbols is $\gamma_{t,t+k}^{0,1/2} = \frac{(1/2)^{k+1}}{1-1/2} = 2^{-k}$.
	Hence the condition $U_t^{0,1/2}(x_{0:t+k}) \subseteq (0,1)_{+\varepsilon}$ with $\varepsilon=2^{-(k+1)}$ is equivalent to $-\varepsilon < V_t^k(x) < 1+\varepsilon-2^{-k}$.
	The left inequality is trivial because $V_t^k(x)\geq 0$, and the right inequality simplifies to $V_t^k(x) < 1-2^{-(k+1)} = T_k$.
	Therefore $x\in L({\mathcal P_k})$ iff $V_t^k(x)<T_k$ for all $t\geq 0$.
\end{proof}

\noindent\textbf{Theorem~\ref{thm:future-lower-bound}.}
\emph{
	For every $k \geq 1$, there is a future-only monitoring instance $\mathcal{P}_k$ such that for every ARM $\mathcal{M}$ with $\tau(\mathcal{P}_k) - 1$ registers we have $L(\mathcal{M}) \neq L(\mathcal{P}_k)$.
}
\begin{proof}
	Let $k \geq 1$ and $\mathcal{P}_k$ be the monitoring instance defined in~\Cref{sec:register-complexity}.
    For brevity, write $V_t := V_t^k$ and $T := T_k$.
	Suppose toward contradiction that there exists an ARM $\mathcal{M}$ with $m:=k-1$ registers such that $L(\mathcal{M}) = L(\mathcal{P}_k)$.

	Let us record two identities that follow from te definition of the monitoring instance $\mathcal{P}_k$.
	For every $t \geq 0$,
	\begin{align}
		V_t(x)
		&=
		x_t+\tfrac12 V_{t+1}(x)-2^{-(k+1)}x_{t+k+1},
		\label{eq:Pk-rec}\\
		V_j(x')-V_j(x)
		&=
		\sum_{i=j}^{k-1}2^{-(i-j)}(x'_i-x_i)
		\qquad (0\leq j\leq k-1),
		\label{eq:Pk-diff}
	\end{align}
	whenever $x_i=x'_i$ for all $i\geq k$.

	\subparagraph{Fixed control paths.}
	For a length-$k$ transition sequence $\pi$ of $\mathcal M$, let $C_\pi \subseteq [0,1]^k$ be the set of input prefixes $u=(u_0,\dots,u_{k-1})$ such that, during the first $k$ steps while reading $u$, the machine executes exactly the transitions of $\pi$.
	Unrolling $\mathcal{M}$ along $\pi$ shows that $C_\pi$ is the intersection of $[0,1]^k$ with finitely many strict affine half-spaces, hence is relatively open in $[0,1]^k$.
	It also yields for each $\pi$ a matrix $A_\pi \in \mathbb R^{m \times k}$ and a vector $b_\pi \in \mathbb R^m$ such that for every $u \in C_\pi$, the register valuation after $k$ steps is $A_\pi u + b_\pi$.
	Since $m = k-1 < k$, every matrix $A_\pi$ has a nontrivial kernel.

	\subparagraph{Near-boundary prefixes and an anchor word.}
	Let $\alpha := T/2$.
	For $\theta \in (3T/4,T]$, define $u(\theta) \in \mathbb R^k$ by setting $u_i(\theta) = \theta/2$ for $0 \leq i \leq k-2$ and $u_{k-1}(\theta) = \theta - \tfrac{T}{4}$.
	All coordinates of $u(\theta)$ lie in $(0,1)$, so $u(\theta) \in (0,1)^k$.

	We first note that for every $\theta < T$ the word $w_\theta := u(\theta) \cdot \alpha \cdot 0^\omega$ belongs to $L(\mathcal{P}_k)$.
	Indeed, for every $j \in \{0,\dots,k-1\}$ we have $V_j(w_\theta) = \theta$.
	For $j = k-1$, the window contains $u_{k-1}(\theta)$ with weight $1$ and $\alpha$ with weight $1/2$, so $V_{k-1}(w_\theta) = (\theta - \frac{T}{4}) + \frac{1}{2} \cdot \frac{T}{2} = \theta$.
	For $j \leq k-2$, the symbol at position $j+k+1$ is $0$, so~\Cref{eq:Pk-rec} gives $V_j(w_\theta) = u_j(\theta) + \frac{1}{2} V_{j+1}(w_\theta) = \frac{\theta}{2} + \frac{1}{2} \theta = \theta$, by backward induction from the case $j = k-1$.
	Moreover, for $t\geq k$, the only possible nonzero symbol still visible is the single letter $\alpha$ at position $k$, hence $V_t(w_\theta) \leq \alpha = T/2 < T$.
	Therefore $w_\theta \in L(\mathcal{P}_k)$ whenever $\theta < T$.

	Next, we use the prefix $u(T)$ to fix a control path.
	Let $\hat{w} := u(T) \cdot 0^\omega$.
	Again, $\hat{w} \in L(\mathcal{P}_k)$:
	we have $V_{k-1}(\hat{w}) = T - \frac{T}{4} = \frac{3T}{4} < T$, and for $j \leq k-2$, we have $V_j(\hat{w}) = u_j(T) + \frac{1}{2} V_{j+1}(\hat{w}) < \frac{T}{2} + \frac{1}{2} T = T$ by backward induction, while $V_t(\hat{w}) = 0$ for all $t \geq k$.

	Since $\hat{w} \in L(\mathcal{P}_k)$ and $\mathcal{M}$ recognizes $L(\mathcal{P}_k)$, the run of $\mathcal{M}$ on $\hat{w}$ is infinite.
	In particular, while reading the prefix $u(T)$ during the first $k$ steps, the machine follows some length-$k$ transition sequence $\pi^\star$.

	\subparagraph{An invisible perturbation inside one control cell.}
	Consider the transition sequence $\pi^\star$ above.
	Let $O := C_{\pi^\star}\cap(0,1)^k$ be the set of length-$k$ prefixes that force the first $k$ transitions of $\mathcal{M}$ to be exactly $\pi^\star$.
	Since $O$ is open in $\mathbb{R}^k$ and $u(T) \in O$, there is a neighborhood of $u(T)$ in which the control path is stable:
	Let $\eta>0$ be such that the open ball $B_\infty(u(T),\eta) := \{v\in\mathbb R^k : \|v-u(T)\|_\infty < \eta\}$	satisfies $B_\infty(u(T),\eta)\subseteq O$.

	Choose any nonzero vector $d \in \ker(A_{\pi^\star})$.
	Intuitively, $d$ encodes an ``invisible'' direction for the register valuations after $k$ following $\pi^\star$.
	Let $j$ be the largest index such that $d_j \neq 0$, and replace $d$ by $-d$ if necessary so that $d_j > 0$.
	Let $\lambda_{\min} := \frac{\eta}{2\|d\|_\infty}$ and $\delta := \min\{\eta/2,\,T/4\}$.
	We pick $\lambda$ such that $0 < \lambda < \min\{\lambda_{\min},\,2\delta/d_j\}$, and define $\theta := T-\frac{\lambda d_j}{2}$.
	Then, $\theta \in (T-\delta,T) \subseteq (3T/4,T)$, so $u(\theta)$ is well-defined.

	By the explicit form of $u(\cdot)$ we have $\|u(\theta)-u(T)\|_\infty = |T-\theta| = \frac{\lambda d_j}{2} < \delta \leq \eta/2$, and therefore $u(\theta) \in B_\infty(u(T),\eta/2)$.
	Also, $\|\lambda d\|_\infty < \lambda_{\min}\|d\|_\infty = \eta/2$, so $\|u(\theta)+\lambda d-u(T)\|_\infty \leq \|u(\theta)-u(T)\|_\infty + \|\lambda d\|_\infty < \eta$, which implies $u(\theta)+\lambda d \in B_\infty(u(T),\eta) \subseteq O$.

	Let $u := u(\theta)$ and $u' := u+\lambda d$ and define $w := u \cdot \alpha \cdot 0^\omega$ and $w' := u' \cdot \alpha \cdot 0^\omega$.
	Because $u,u'\in C_{\pi^\star}$, the runs of $\mathcal{M}$ on $w$ and $w'$ take the same first $k$ transitions. Moreover, $u'-u = \lambda d$ and $d \in \ker(A_{\pi^\star})$, so the register valuations after $k$ steps coincide.
	Since the suffixes from position $k$ onward are identical, determinism yields $w \in L(\mathcal{M})$ iff $w'\in L(\mathcal{M})$.

	\subparagraph{Contradiction.}
	Since $\theta<T$, the construction above gives $w \in L(\mathcal{P}_k)$.
	On the other hand, $w$ and $w'$ agree at all positions $i\geq k$, so~\Cref{eq:Pk-diff} yields $V_j(w')-V_j(w) = \lambda\sum_{i=j}^{k-1}2^{-(i-j)}d_i$.
	By the choice of $j$ as the largest index with $d_j \neq 0$, we have $d_i = 0$ for all $i>j$, and therefore $V_j(w')-V_j(w) = \lambda d_j$.
	Since $V_j(w)=\theta$, it follows that $V_j(w') = \theta+\lambda d_j > T$ by the definition of $\theta$.
	Hence $w' \notin L(\mathcal{P}_k)$.
	However, since we assumed $L(\mathcal{M}) = L(\mathcal{P}_k)$ and showed that $w \in L(\mathcal{M})$ iff $w'\in L(\mathcal{M})$, we obtain a contradiction.
\end{proof}

\section{Further Experimental Evaluation}
\label{sec:more_experiments}
Because of space restrictions, we include extra experiments supporting our evaluation in this appendix.

\subsection{RQ1: Resource Usage}
In addition to the experiments with \texttt{PowerData} traces, we also test our monitors for resource monitoring classification accuracy of a neural network trained on the MNIST digit dataset. 
In particular, we are interested in how the monitoring evolves under a shift in the underlying distribution.

\begin{figure}[t]
     \centering
     \begin{subfigure}{0.32\textwidth}
         \centering
         \includegraphics[width=\textwidth]{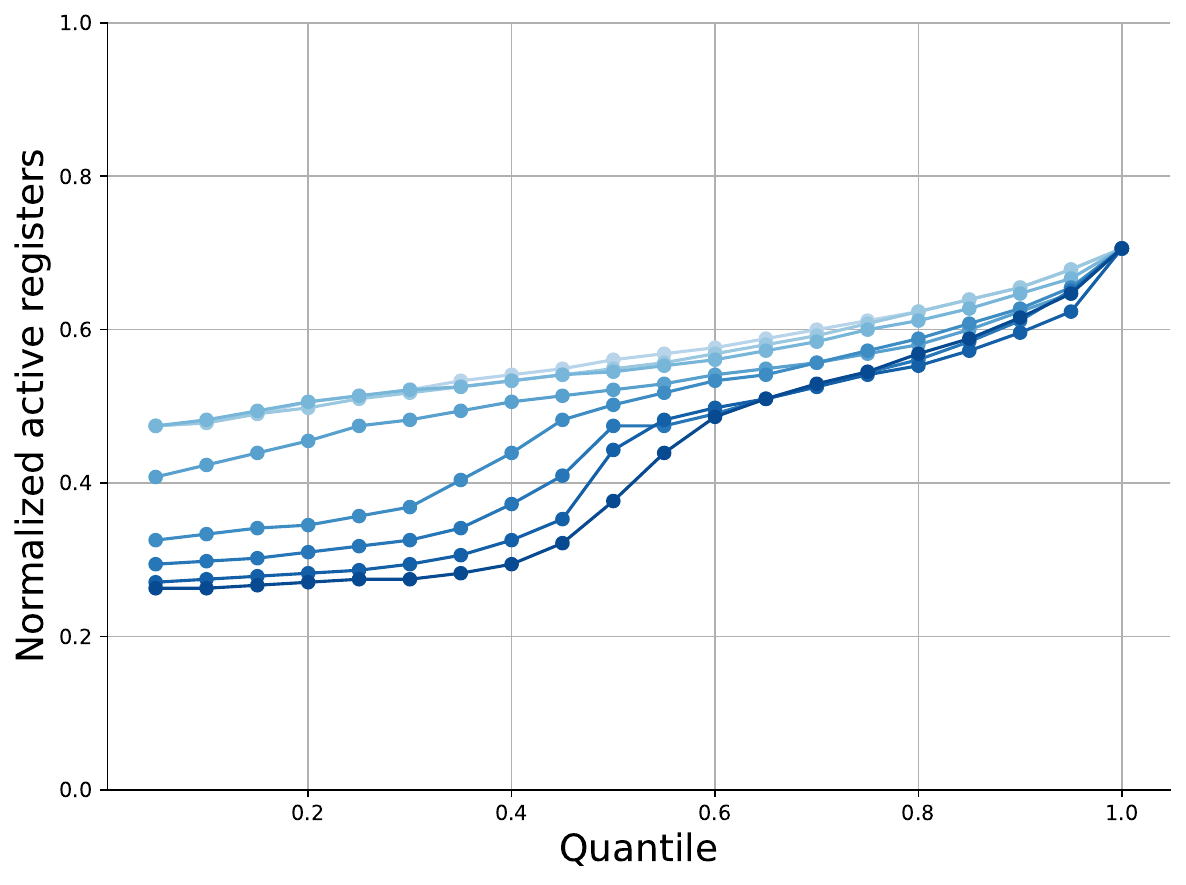}
         \caption{}
         \label{fig:mnist-quantiles}
     \end{subfigure}
     \begin{subfigure}{0.32\textwidth}
         \centering
         \includegraphics[width=\textwidth]{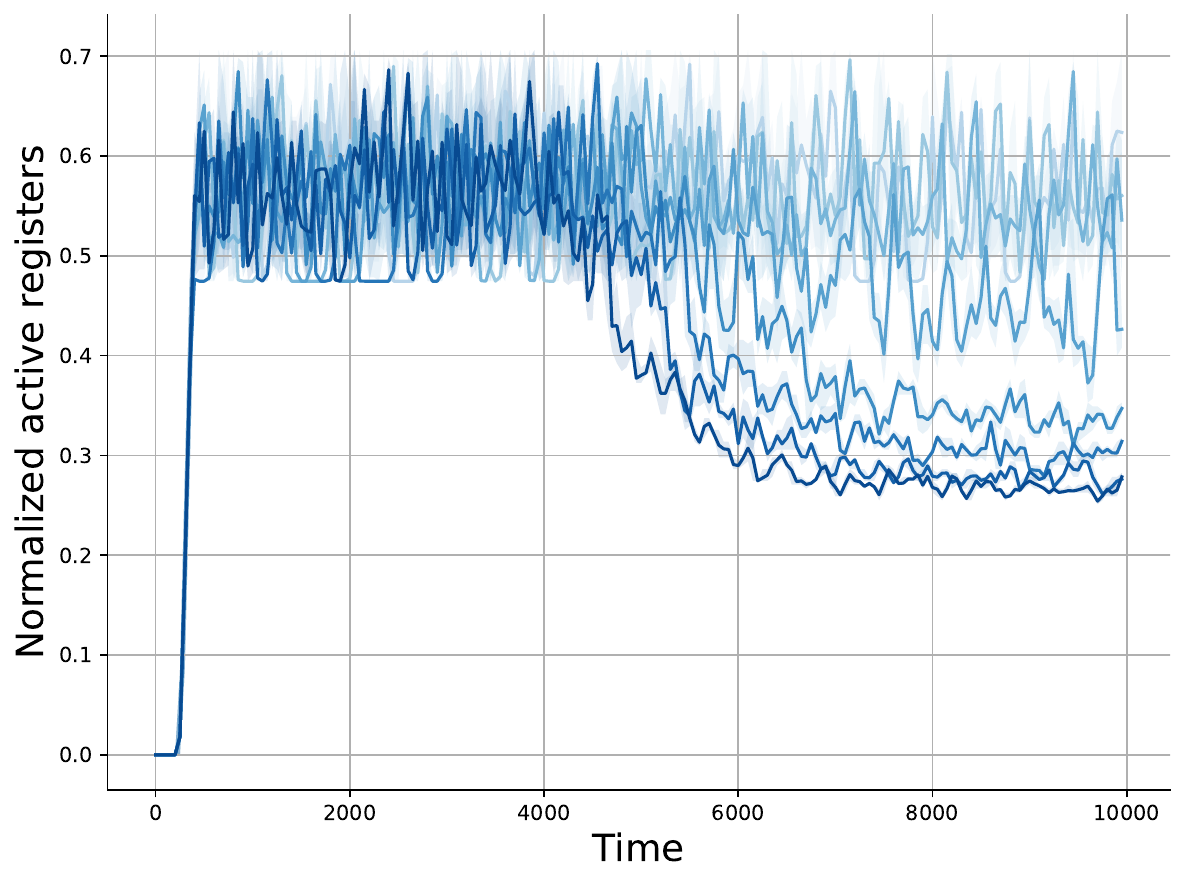}
         \caption{}
         \label{fig:mnist-whole-execution}
     \end{subfigure}
     \begin{subfigure}{0.32\textwidth}
         \centering
         \includegraphics[width=\textwidth]{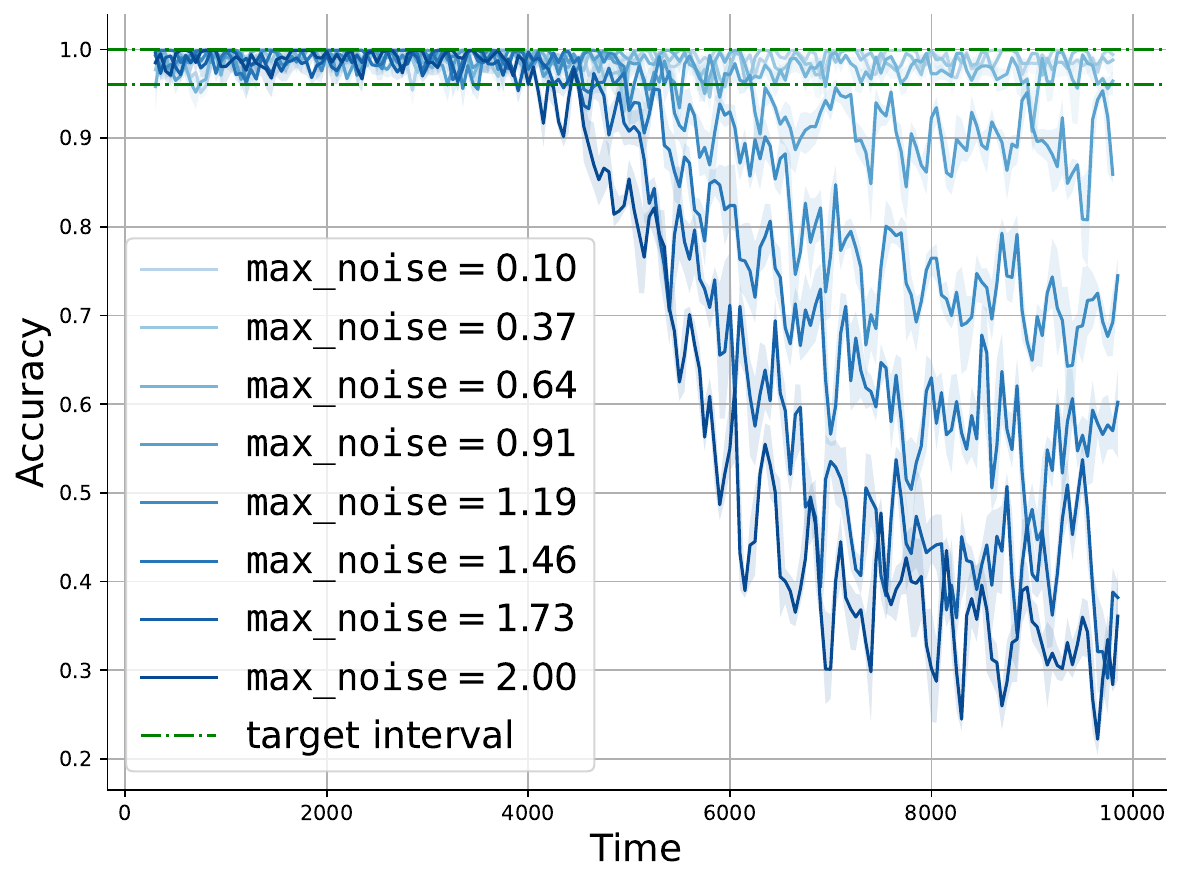}
         \caption{}
         \label{fig:mnist-whole-execution-accuracy}
     \end{subfigure}
        \caption{Normalized active registers and discounted values, \texttt{MNIST} dataset.}
        \label{fig:RQ1-mnist}
\end{figure}

For MNIST, we train a neural network on the training set and generate a trace of length $10000$ by sequentially sampling test instances, recording whether each prediction is correct.
Thus, $\mathcal R=\{0,1\}$, with discount factors fixed to $r=s=0.95$, target interval $\mathcal I=[0.96,1]$, and tolerance $\varepsilon=0.001$.

We induce distribution shift by adding Gaussian noise with time-increasing variance $\sigma(t)$ to the test inputs; the parameter \texttt{max\_noise} controls the final corruption level.
As shown in Fig.~\ref{fig:RQ1-mnist}(a), increasing noise leads to reduced register usage.
Figures~\ref{fig:RQ1-mnist}(b) and (c) illustrate a representative execution, plotting the number of active registers and the monitored discounted accuracy, respectively.
As accuracy degrades over time, the monitor can more quickly conclude violation, requiring fewer active registers.
Consistent with the \texttt{PowerData} results, using only $70\%$ of the theoretical register budget would still be mostly sound.

\subsection{RQ2: Monitoring Demographic Parity}

To further investigate RQ2, we monitor demographic parity on the \texttt{Adult} dataset with a fixed tolerance $\varepsilon = 5\cdot 10^{-5}$ and decreasing interval width. As before, the input set is $\mathcal R=\{0,1\}$ and we again fix $r=s=0.95$.
We report the results in Fig.~\ref{fig:RQ2-adult-appendix}, with the same format Fig.~\ref{fig:RQ2-adult}, for the demographic parity property (Ex.~\ref{ex:dem-parity}) with the synchronous (blue) and asynchronous (green) interpretation.
As expected, and validating the results on Fig.~\ref{fig:RQ2-adult}, synchronous interpretations are significantly more efficient than asynchronous ones, since discounting is applied uniformly across streams rather than independently.

\begin{figure}[t]
     \begin{subfigure}{0.45\textwidth}
         \centering
         \includegraphics[width=\textwidth]{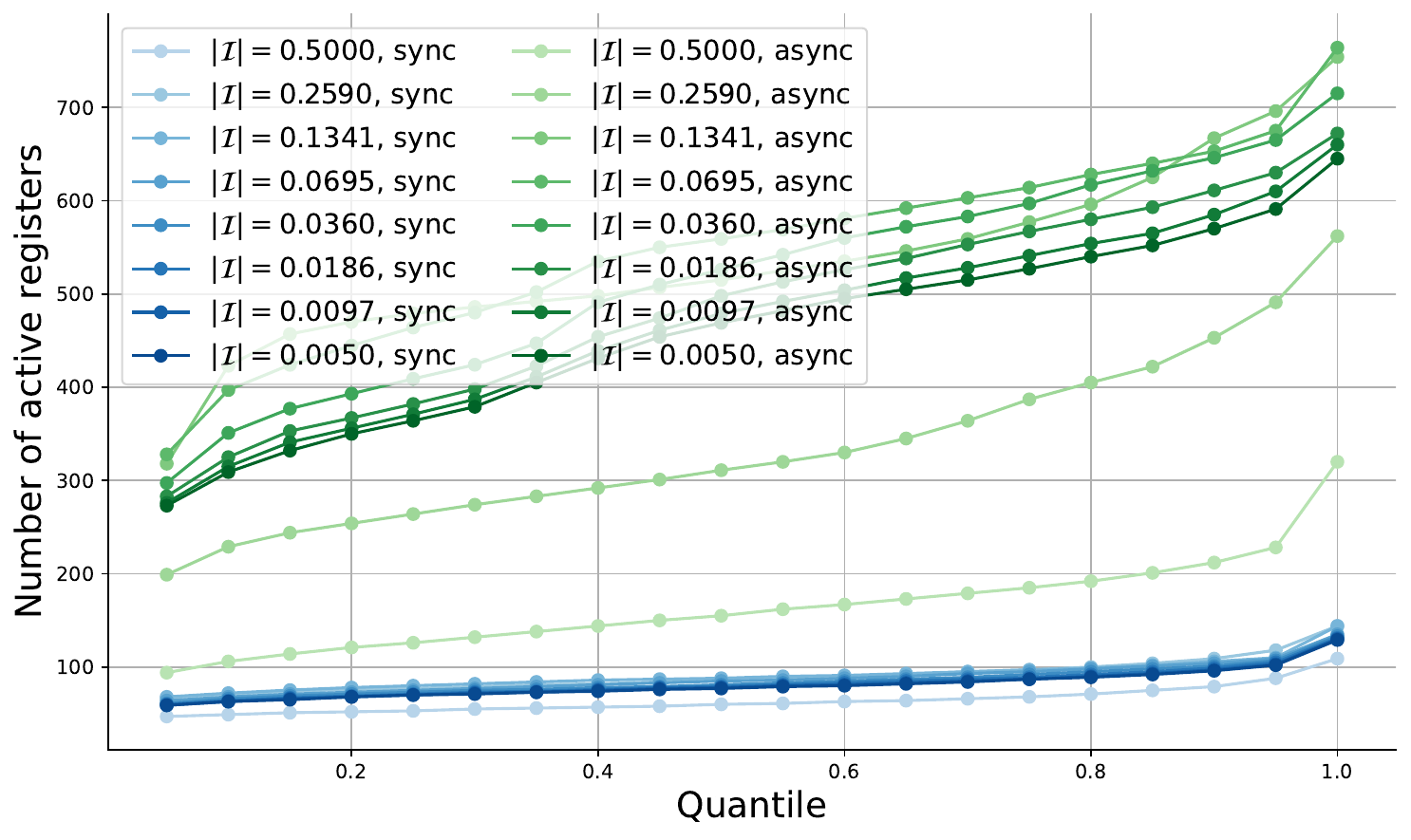}
         % \caption{}
     \end{subfigure}
     \begin{subfigure}{0.45\textwidth}
         \centering
         \includegraphics[width=\textwidth]{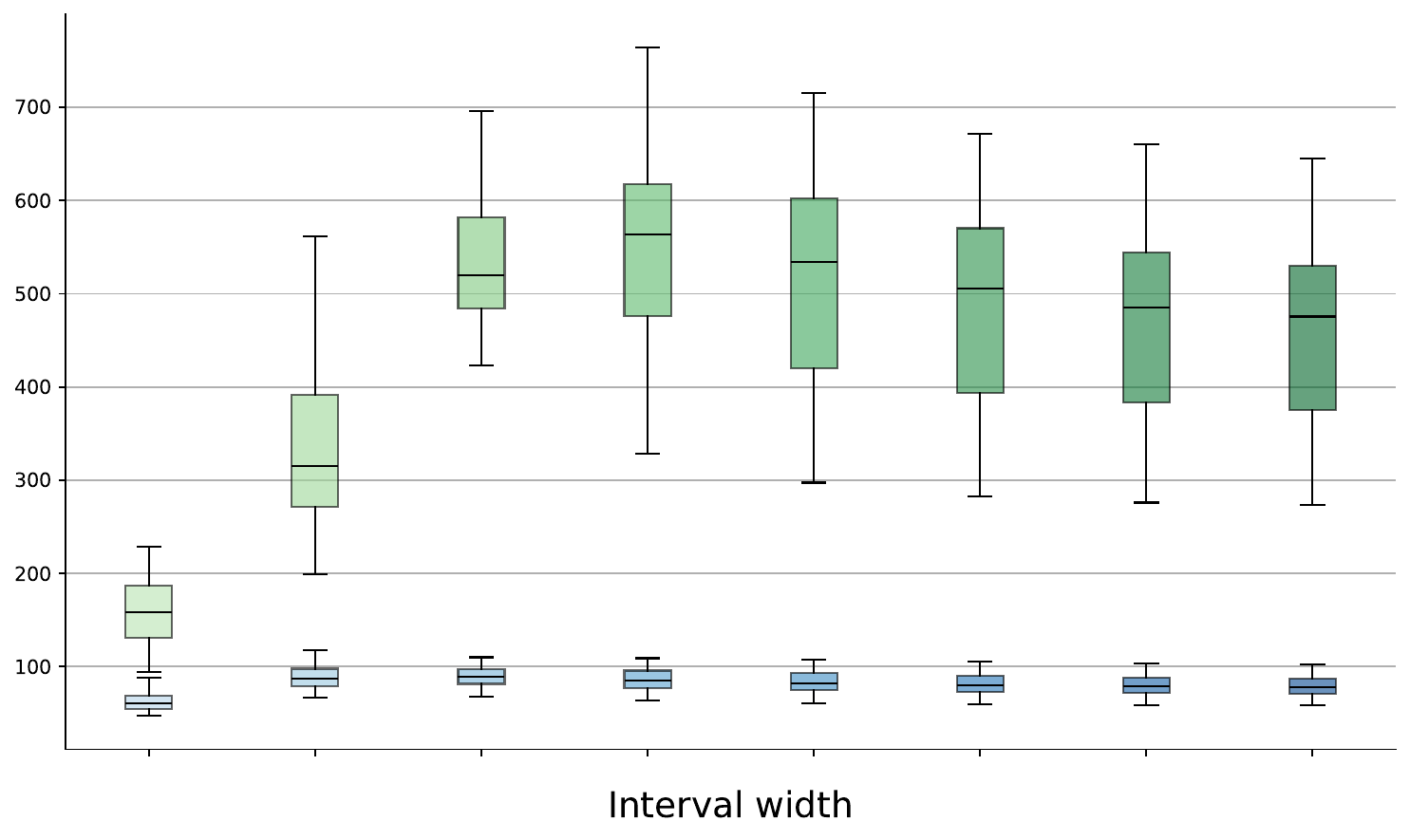}
         % \caption{}
     \end{subfigure}
        \caption{Active registers when monitoring demographic parity with variable target interval width, \texttt{Adult} dataset.}
        \label{fig:RQ2-adult-appendix}
\end{figure}

\subsection{RQ3: Statistical Monitoring Behaviour}

Our objective is to assess the behaviour of the statistical monitor during deployment.
Since the statistical uncertainty intervals do not converge to width zero, the main question is power, i.e., whether the monitor eventually produces a decisive verdict.
To evaluate this question, we use a synthetic piecewise-stationary Beta process.
We use synthetic data because the evaluation requires access to the ground-truth expected discounted average.

\subparagraph{Setup.}
The Beta process consists of four phases, each spanning $150$ time steps.
At each time index, the observation $X_t$ is sampled independently from a Beta distribution parametrised by $(a,b)$, where $(a,b)= (1,9),\; (5,5),\; (8,2),\; (4,6)$.
The corresponding phase means are approximately $0.10$, $0.50$, $0.80$, and $0.40$.
We use a uniform upper bound on the conditional sub-Gaussian norm of approximately $\sgn=0.15$.
The monitor observes only the realised samples $X_t$, while the latent process
$P_t=\expe_{t-1}(X_t)$ is used only for evaluation.

We monitor the expected discounted average $\edsum_t^{r,s}(W)/\lambda^{r,s}$ where $\lambda^{r,s}=1+\frac{r}{1-r}+\frac{s}{1-s}$. In the experiment, $r=s=0.95$, so $\lambda^{r,s}\approx39$.
The target interval is $I=[0.4,0.6]$, the tolerance is $\varepsilon=0.05$, and the error probability is $\delta=0.01$.
Thus, positive verdicts are sound with respect to
$I_{+\varepsilon}=[0.35,0.65]$, while negative verdicts are sound with respect to the complement of
$I_{-\varepsilon}=[0.45,0.55]$.

We use the process depicted in Fig.~\ref{fig:stat:process} to evaluate the pointwise, local, and uniform statistical error bounds.
Since $r=s=0.95$, the discounted average is strongly smoothed: near phase boundaries, it blends past and future phases instead of following the latent mean instantaneously.
The figure therefore illustrates the main purpose of discounted monitoring: it suppresses single-sample noise while still reacting to changes in the local behaviour.

\begin{quote}
    \emph{Important:} In the experiments below, we also use flexible release with the pointwise bound.
    This is for demonstration purposes only: the goal is to visualise how the statistical uncertainty intervals evolve over time, independently of the choice of a fixed release horizon.
\end{quote}

\begin{figure}
    \centering
    \includegraphics[width=1.0\linewidth]{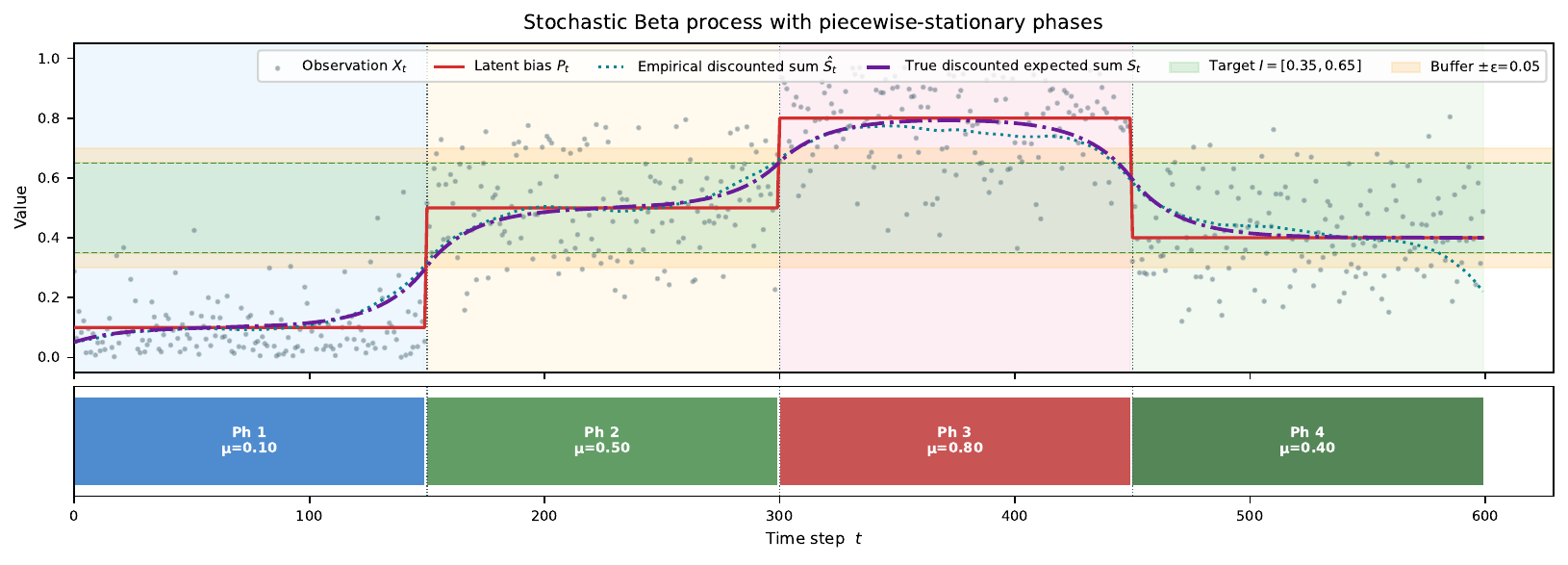}
    \caption{Four-phase Beta process $\mathrm{Beta}(a,b)$.
    The figure shows the realised observations $X_t$, the latent mean process
    $P_t$, the observed discounted average $\dsum_t^{r,s}(W)/\lambda^{r,s}$, and the oracle expected discounted average $\edsum_t^{r,s}(W)/\lambda^{r,s}$.
    The shaded background indicates the four Beta phases.}
    \label{fig:stat:process}
\end{figure}

\subparagraph{Monitor verdicts.}
We first assess how many observations are required before each monitored time index $t$ receives a decisive verdict.
Fig.~\ref{fig:stat:heat} shows the resulting verdicts for the pointwise, local, and uniform statistical error bounds.
The observed behaviour matches the phase structure of the process.
Time indices centred in the low and high phases tend to receive negative verdicts, while time indices centred in the middle phases tend to receive positive verdicts.
Around phase boundaries, verdicts are delayed or remain inconclusive because the discounted average mixes information from neighbouring phases.

The number of inconclusive verdicts increases with the strength of the soundness guarantee.
This is also visible in Fig.~\ref{fig:stat:evolution}, which overlays the statistical uncertainty intervals with the resulting verdicts.
Pointwise intervals are the narrowest, local intervals are wider because they permit flexible release over all horizons $n\geq t$, and uniform intervals are widest because they protect all monitored time indices simultaneously.
Consequently, pointwise verdicts tend to appear earliest, while uniform verdicts are the most conservative.

\begin{figure}
    \centering
    \includegraphics[width=1.0\linewidth]{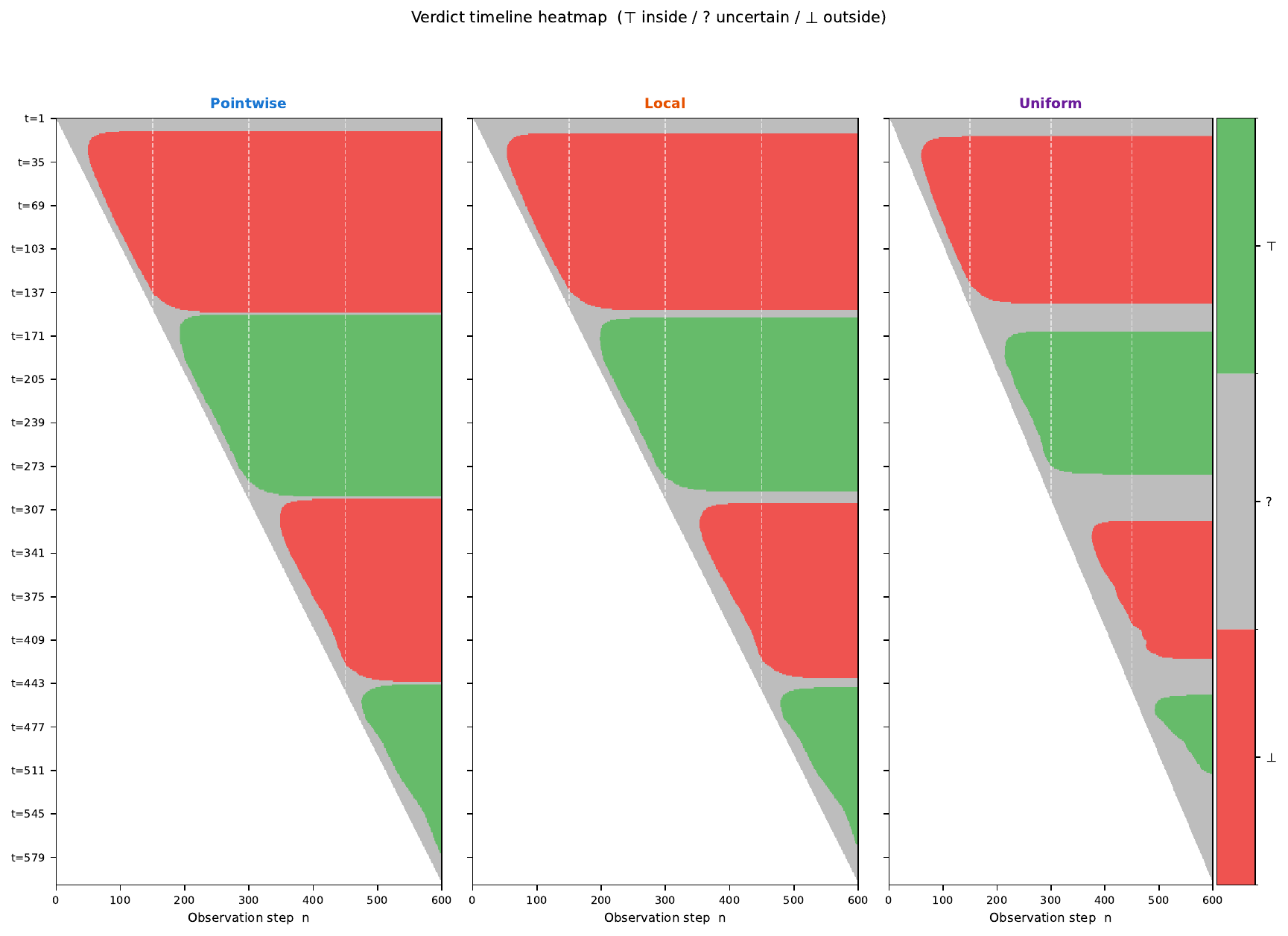}
    \caption{Monitor verdicts.
    Rows correspond to monitored time indices $t$, columns to observation horizons $n$, and colours encode the verdict:
    $\top$ for inside, $\bot$ for outside, and $?$ for inconclusive.}
    \label{fig:stat:heat}
\end{figure}

\begin{figure}
    \centering
    \includegraphics[width=1.0\linewidth]{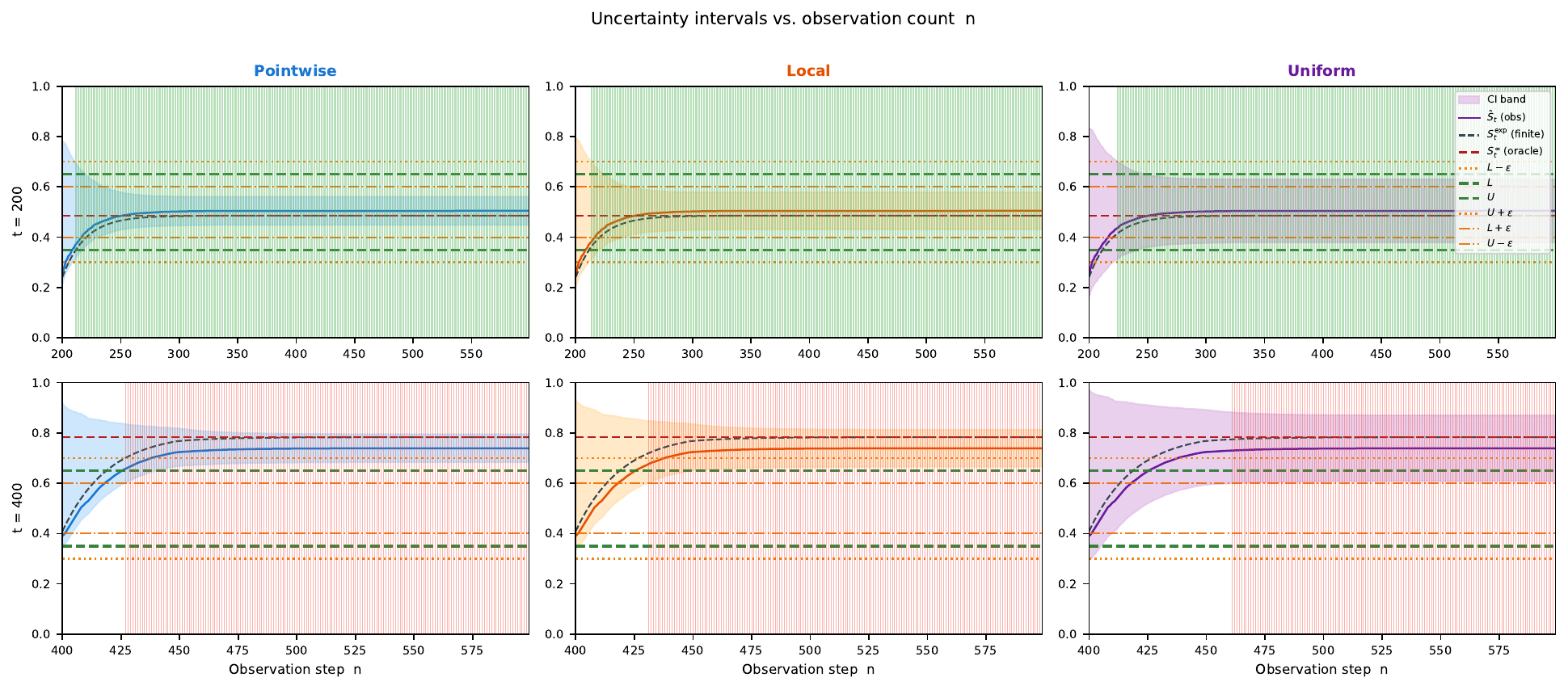}
    \caption{Evolution of pointwise, local, and uniform statistical uncertainty intervals.
    The figure fixes several monitored time indices $t$ and shows how the uncertainty interval evolves as more observations $n\geq t$ become available.
    The solid line is the observed discounted average, the dashed line is the finite expected discounted average, and the horizontal red line is the oracle expected discounted average.}
    \label{fig:stat:evolution}
\end{figure}

\subparagraph{Statistical uncertainty interval evolution.}
We next decompose the half-width of the statistical uncertainty interval.
Fig.~\ref{fig:stat:decomp} shows that the deterministic tail error decreases as the observation horizon $n$ grows, because more future observations have been revealed.
The statistical error term, in contrast, increases toward its limiting value, since more noisy observations enter the discounted estimate.
Thus, the total uncertainty is governed by a tradeoff: waiting reduces deterministic tail error, but exposes the estimate to its limiting statistical error.

Fig.~\ref{fig:stat:process_ci} shows the statistical uncertainty interval at the time of the first decisive verdict and at the end of the run, assuming that the register is never released.
This highlights the tension between statistical certainty and resource consumption: delaying release may reduce the deterministic tail error, but it also keeps registers active for longer and cannot remove the limiting statistical error.

\begin{figure}
    \centering
    \includegraphics[width=1.0\linewidth]{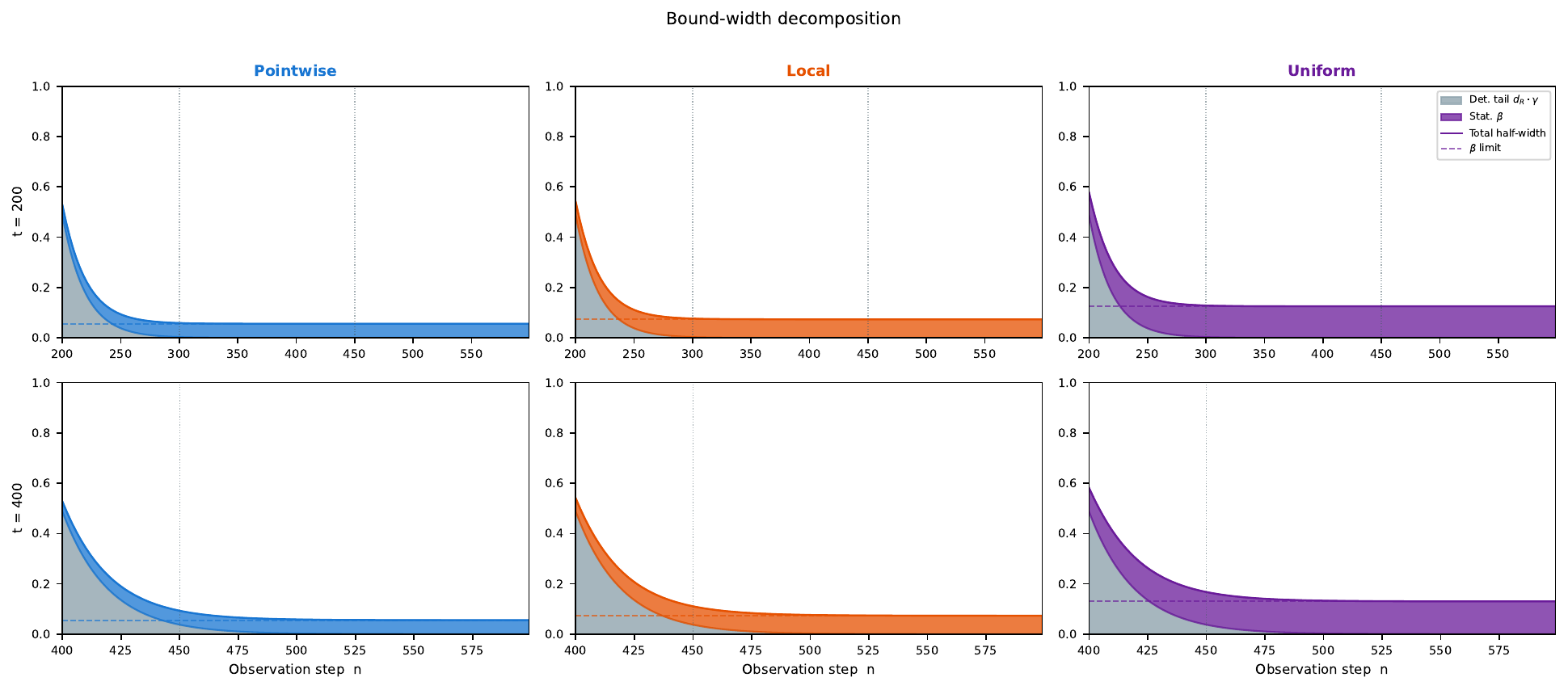}
    \caption{Uncertainty half-width decomposition.
    The figure decomposes the half-width of the statistical uncertainty interval into the deterministic tail term $d_{\mathcal R}\gamma_{t,n}^{r,s}$ and the statistical error term $\beta_{t,n}^{r,s}(\delta)$.}
    \label{fig:stat:decomp}
\end{figure}

\begin{figure}[p]
    \centering

    \includegraphics[width=0.9\linewidth]{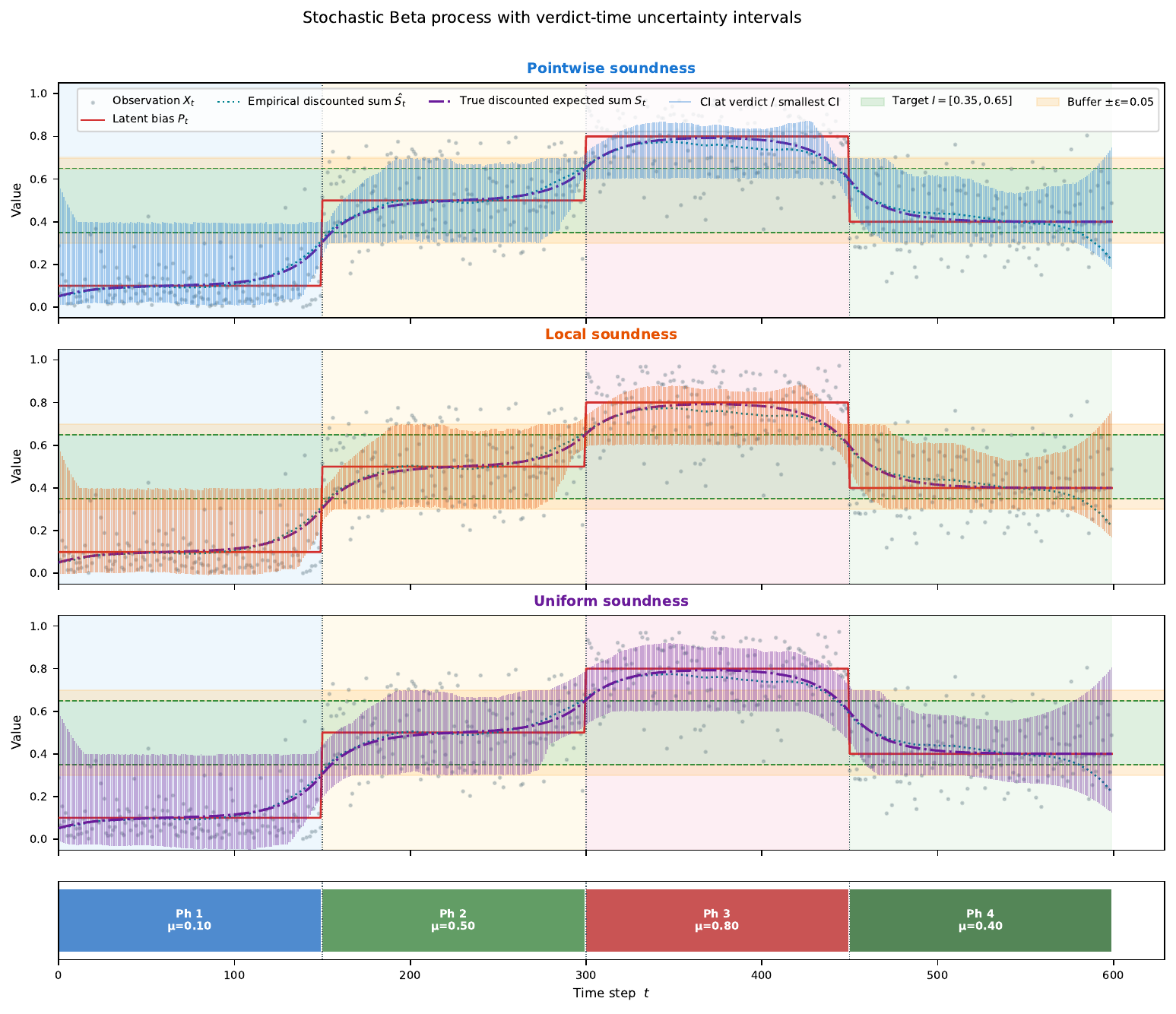}
    \vspace{0.5em}

    \includegraphics[width=0.9\linewidth]{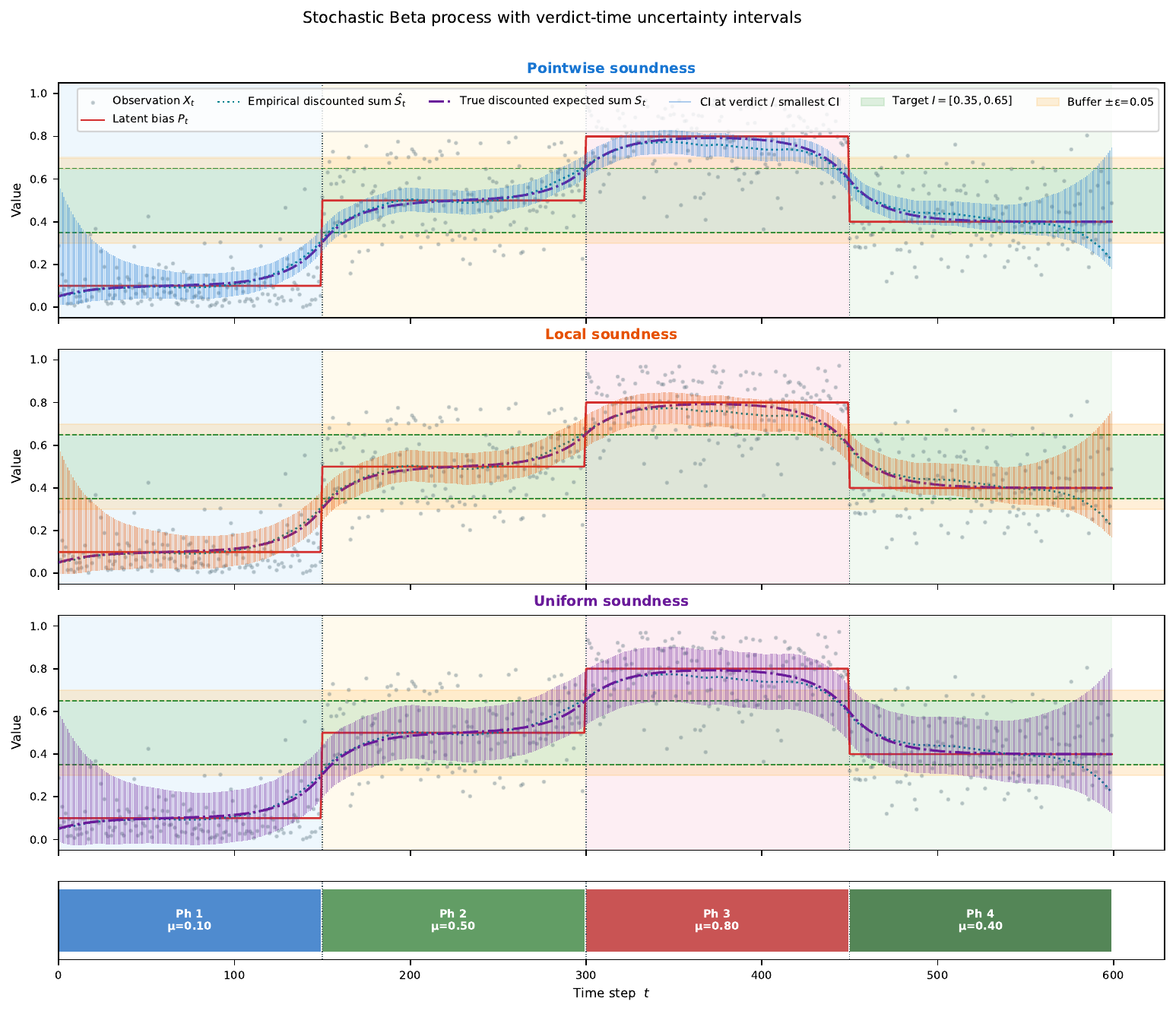}

    \caption{Statistical uncertainty intervals at the point of first decisive verdict and at the end of the run.}
    \label{fig:stat:process_ci}
\end{figure}

\subparagraph{Monte Carlo violation rates.}
We use a Monte Carlo simulation with $1000$ repetitions to evaluate the coverage of the statistical uncertainty intervals, the verdict frequency, the verdict delay, and the verdict correctness.
Table~\ref{tab:mc_metrics} reports Monte Carlo averages with standard deviations over independent runs.
The first block uses the setting of this section.
The second block uses a lower-variance setting, where all Beta parameters are multiplied by $10$, i.e., $\mathrm{Beta}(10a,10b)$, giving a sub-Gaussian norm upper bound of approximately $\sgn=0.05$.
We report the average statistical uncertainty interval violation rate, both pointwise over all $(t,n)$ pairs and run-wise as the fraction of runs containing at least one interval violation.
We also report the fraction of monitored time indices that eventually receive a decisive verdict, the average release delay $n-t$, the fraction of incorrect first verdicts among released verdicts, and the fraction of runs containing at least one incorrect verdict.
The statistical uncertainty intervals are conservative in these experiments.
Only the pointwise bound exceeds the nominal error probability in the run-wise interval violation metric, which is expected because pointwise soundness does not protect flexible release or whole-run events.
Compared with the lower-variance setting, the higher-variance setting issues fewer verdicts and has larger release delays.
Across all configurations, incorrect verdicts are rare.

\begin{table}[t]
\centering
\caption{Aggregate statistics of the Monte Carlo experiments with $1000$ repetitions.
\emph{Coverage:}
Interval viol.\ is the fraction of statistical uncertainty interval violations;
Any interval is the fraction of runs with at least one statistical uncertainty interval violation.
\emph{Release:}
Released is the fraction of monitored time indices that receive a decisive verdict;
Delay is the average release delay.
\emph{Verdicts:}
Wrong rate is the fraction of incorrect first verdicts among released verdicts;
Any wrong is the fraction of runs with at least one incorrect verdict.}
\label{tab:mc_metrics}
\begin{tabular}{lcccccc}
\toprule
$(a,b)$ & \multicolumn{2}{c}{Coverage} & \multicolumn{2}{c}{Release} & \multicolumn{2}{c}{Verdicts} \\
Sound. & Interval viol. & Any interval & Released & Delay & Wrong rate & Any wrong \\
\midrule
Pointwise & 0.000 $\pm$ 0.003 & 0.021 $\pm$ 0.143 & 0.924 $\pm$ 0.014 & 21.279 $\pm$ 0.920 & 0.000 $\pm$ 0.000 & 0.002 $\pm$ 0.045 \\
Local & 0.000 $\pm$ 0.000 & 0.000 $\pm$ 0.000 & 0.877 $\pm$ 0.025 & 24.204 $\pm$ 1.067 & 0.000 $\pm$ 0.000 & 0.000 $\pm$ 0.000 \\
Uniform & 0.000 $\pm$ 0.000 & 0.000 $\pm$ 0.000 & 0.662 $\pm$ 0.024 & 33.304 $\pm$ 1.436 & 0.000 $\pm$ 0.000 & 0.000 $\pm$ 0.000 \\
\bottomrule
\toprule
$(10a,10b)$ & \multicolumn{2}{c}{Coverage} & \multicolumn{2}{c}{Release} & \multicolumn{2}{c}{Verdicts} \\
Sound. & Interval viol. & Any interval & Released & Delay & Wrong rate & Any wrong \\
\midrule
Pointwise & 0.000 $\pm$ 0.002 & 0.014 $\pm$ 0.118 & 0.958 $\pm$ 0.002 & 14.372 $\pm$ 0.178 & 0.000 $\pm$ 0.000 & 0.000 $\pm$ 0.000 \\
Local & 0.000 $\pm$ 0.000 & 0.000 $\pm$ 0.000 & 0.955 $\pm$ 0.002 & 15.277 $\pm$ 0.191 & 0.000 $\pm$ 0.000 & 0.000 $\pm$ 0.000 \\
Uniform & 0.000 $\pm$ 0.000 & 0.000 $\pm$ 0.000 & 0.944 $\pm$ 0.003 & 18.335 $\pm$ 0.243 & 0.000 $\pm$ 0.000 & 0.000 $\pm$ 0.000 \\
\bottomrule
\end{tabular}
\end{table}

\subsection{RQ4: Statistical Uncertainty Intervals}

Our objective is to assess how the parameter choices affect the half-width of the statistical uncertainty interval.
For all experiments in this subsection, we set the past discount factor to zero, i.e., $r=0$, and monitor the expected discounted average. If not explicitly state we set the error probability to $\delta=0.01$.

\subparagraph{Convergence behaviour.}
In Fig.~\ref{fig:stat:parameter_convergence}, we study how the statistical error term and the deterministic tail error converge to their limiting values.
We vary the future discount factor, the sub-Gaussian norm, and the soundness notion.
Higher sub-Gaussian norms lead to larger statistical error terms and therefore wider statistical uncertainty intervals.
Larger future discount factors spread the normalised weights over a longer effective window, which reduces the limiting statistical error for the discounted average.
At the same time, larger future discount factors increase the deterministic tail error for finite observation horizons, so more observations are required before the tail error becomes negligible.

\begin{figure}
    \centering
    \includegraphics[width=1.0\linewidth]{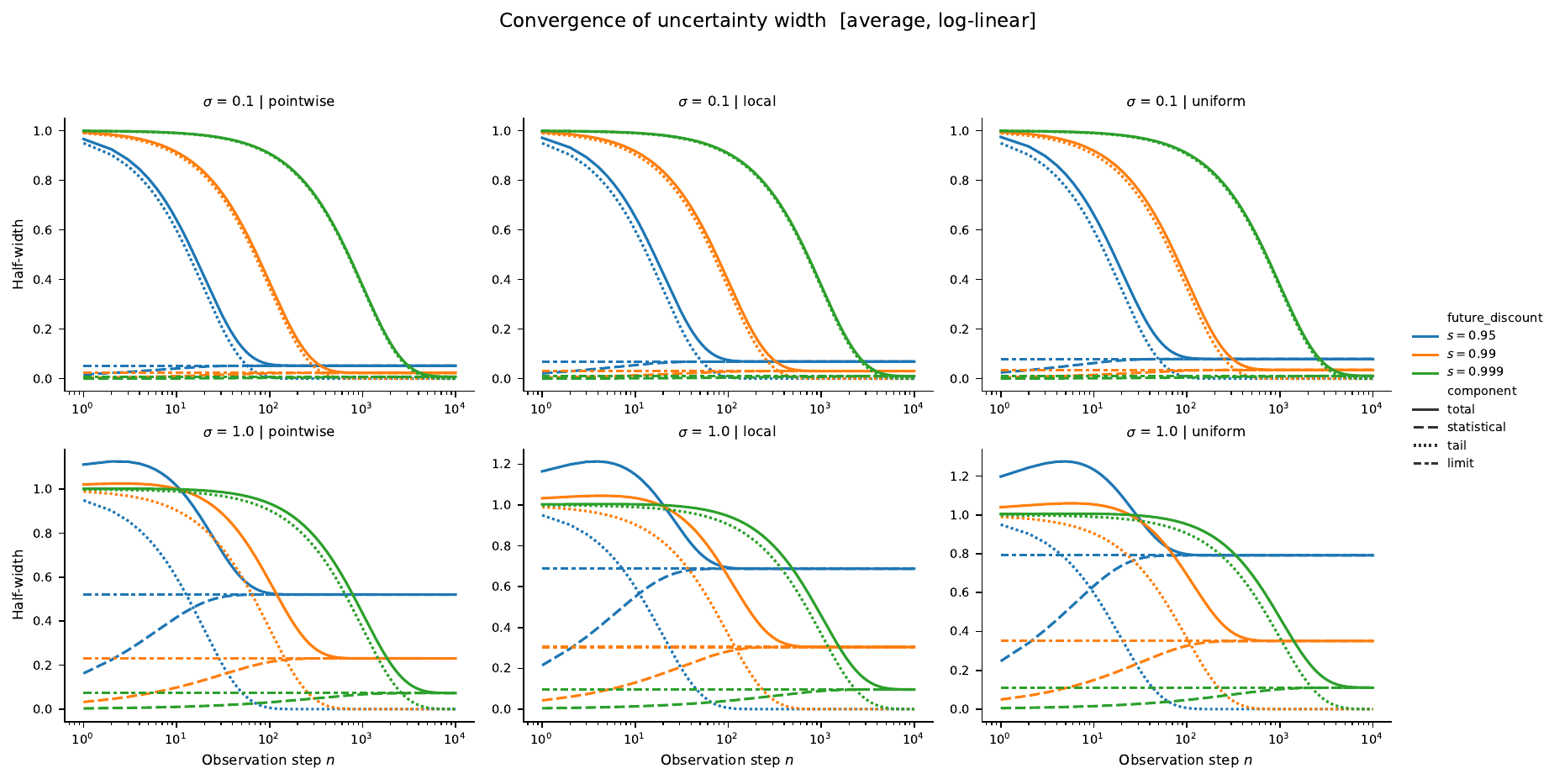}
    \caption{Convergence of the error components to their limiting values.
    The figure shows different sub-Gaussian norms (rows) and future discount factors (hue) for each soundness notion (columns).}
    \label{fig:stat:parameter_convergence}
\end{figure}

\subparagraph{Limit behaviour.}
In Fig.~\ref{fig:stat:parameter_limit}, we study the limiting statistical error bound as a function of the future discount factor.
The error probability has a comparatively small effect relative to the sub-Gaussian norm.
Decreasing the error probability increases the bound logarithmically, while increasing the sub-Gaussian norm increases the bound directly.
The bounds become practically useful mainly for large discount factors close to $1$, where the discounted average aggregates over a longer effective window.

\begin{figure}
    \centering
    \includegraphics[width=1.0\linewidth]{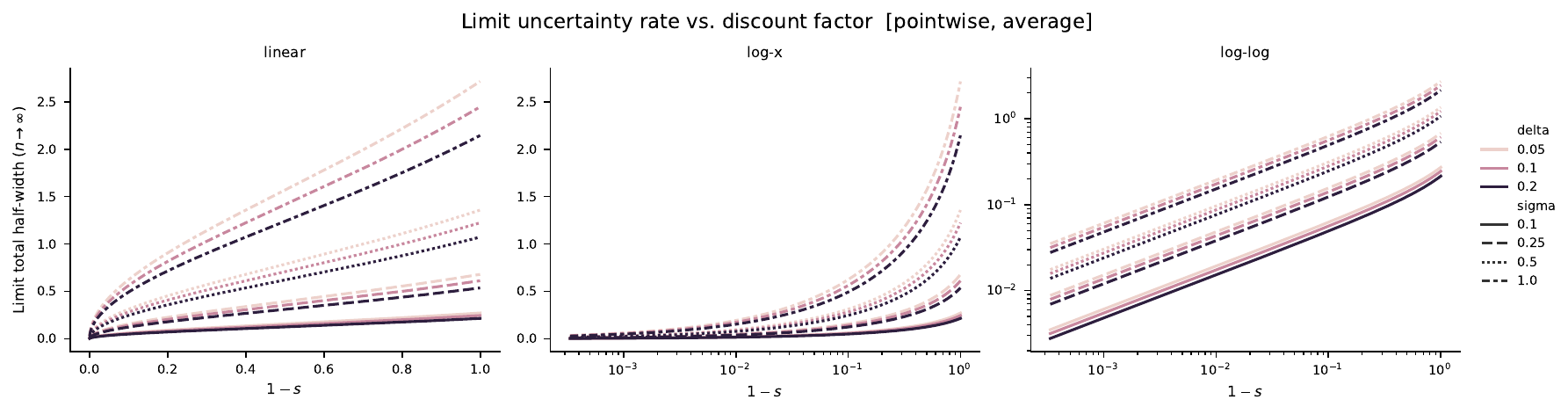}
    \caption{Effect of the future discount factor on the limiting statistical error bound for different error probabilities $\delta$ (hue) and sub-Gaussian norms $\sgn$ (style).
    To visualise the effects, the axis scales vary between columns.}
    \label{fig:stat:parameter_limit}
\end{figure}

\subparagraph{Uniform soundness.}
In Fig.~\ref{fig:stat:uniform}, we study the effect of the monitored time index $t$ on the uniform statistical error bound.
The uniform bound uses a time-dependent error budget, and therefore increases with $t$.
As expected, the increase is logarithmic in the monitored time index.

\begin{figure}
    \centering
    \begin{subfigure}{0.49\linewidth}
        \centering
        \includegraphics[width=\linewidth]{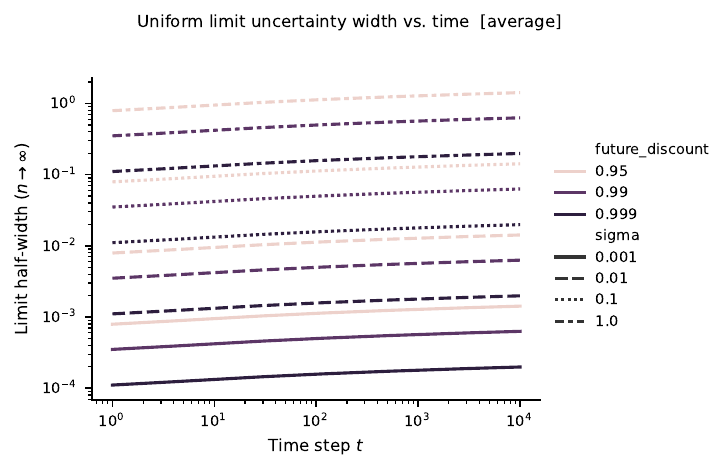}
    \end{subfigure}
    \hfill
    \begin{subfigure}{0.45\linewidth}
        \centering
        \includegraphics[width=\linewidth]{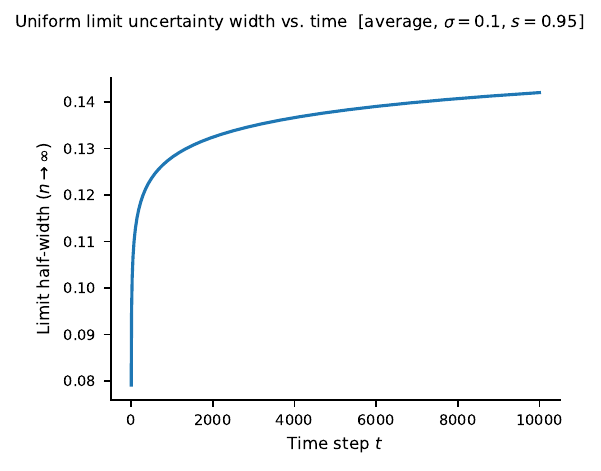}
    \end{subfigure}
    \caption{Effect of the monitored time index on the uniform statistical error bound.
    The left plot shows the effect on a log-log scale for different discount factors (hue) and sub-Gaussian norms (style).
    The right plot shows the same effect on a linear scale for one representative parameter choice.}
    \label{fig:stat:uniform}
\end{figure}

\end{document}